%% file: main.tex
\documentclass[a4paper,onecolumn,11pt,accepted=2025-08-22]{quantumarticle}
\pdfoutput=1
\usepackage[utf8]{inputenc}
\usepackage[english]{babel}
\usepackage[T1]{fontenc}
\usepackage[margin=1in]{geometry}
\usepackage{multicol,multirow}
\usepackage{amssymb,amsmath,amsthm,amsfonts}
\usepackage{textgreek}
\usepackage{mathtools}
\usepackage{enumitem}
\usepackage[numbers]{natbib}
\usepackage{graphicx}
\usepackage{caption,subcaption}
\usepackage{float}
\usepackage[ruled,vlined,linesnumbered]{algorithm2e}
\usepackage{footnote}
\usepackage{xcolor}
\usepackage{mathrsfs}
\usepackage{bbm}
\usepackage{svg}
\usepackage{siunitx}
\usepackage{physics}

\usepackage[colorlinks,linkcolor=cyan,citecolor=cyan, urlcolor=cyan]{hyperref}

\usepackage{wrapfig}
\usepackage{enumitem}

\usepackage[capitalize,noabbrev]{cleveref}

\usepackage{pifont}
\theoremstyle{plain}
\newtheorem{theorem}{Theorem}
\newtheorem{proposition}[theorem]{Proposition}
\newtheorem{lemma}[theorem]{Lemma}

\theoremstyle{definition}
\newtheorem{definition}{Definition}

\theoremstyle{remark}
\newtheorem{remark}{Remark}
\newtheorem{example}{Example}

\numberwithin{equation}{section}

\usepackage{thm-restate}

\newcommand{\eqn}[1]{\hyperref[eqn:#1]{(\ref*{eqn:#1})}}
\newcommand{\rem}[1]{\hyperref[rem:#1]{Remark~\ref*{rem:#1}}}
\newcommand{\thm}[1]{\hyperref[thm:#1]{Theorem~\ref*{thm:#1}}}
\newcommand{\cor}[1]{\hyperref[cor:#1]{Corollary~\ref*{cor:#1}}}
\newcommand{\defn}[1]{\hyperref[defn:#1]{Definition~\ref*{defn:#1}}}
\newcommand{\lem}[1]{\hyperref[lem:#1]{Lemma~\ref*{lem:#1}}}
\newcommand{\prop}[1]{\hyperref[prop:#1]{Proposition~\ref*{prop:#1}}}
\newcommand{\fig}[1]{\hyperref[fig:#1]{Figure~\ref*{fig:#1}}}
\newcommand{\tab}[1]{\hyperref[tab:#1]{Table~\ref*{tab:#1}}}
\newcommand{\algo}[1]{\hyperref[algo:#1]{Algorithm~\ref*{algo:#1}}}
\renewcommand{\sec}[1]{\hyperref[sec:#1]{Section~\ref*{sec:#1}}}
\newcommand{\append}[1]{\hyperref[append:#1]{Appendix~\ref*{append:#1}}}
\newcommand{\fac}[1]{\hyperref[fac:#1]{Fact~\ref*{fac:#1}}}
\newcommand{\lin}[1]{\hyperref[lin:#1]{Line~\ref*{lin:#1}}}
\newcommand{\fnote}[1]{\hyperref[fnote:#1]{Footnote~\ref*{fnote:#1}}}

\def\>{\rangle}
\def\<{\langle}
\def\trans{^{\top}}

\newcommand{\vect}[1]{\ensuremath{\boldsymbol{#1}}}

\newcommand{\Z}{\mathbb{Z}}
\newcommand{\R}{\mathbb{R}}
\newcommand{\C}{\mathbb{C}}

\newcommand{\Herm}
{\mathrm{Herm}}

\newcommand{\V}{\mathcal{V}}
\newcommand{\E}{\mathcal{E}}
\renewcommand{\S}{\mathcal{S}}

\renewcommand{\H}{\mathcal{H}}
\renewcommand{\d}{\mathrm{d}}
\newcommand{\vv}{\mathbf{v}}
\newcommand{\xx}{\mathbf{x}}

\newcommand{\Hebd}{H^{\mathrm{ebd}}}
\newcommand{\Hpen}{H^{\mathrm{pen}}}

\DeclareMathOperator{\poly}{poly}
\DeclareMathOperator{\spn}{span}

\DeclareMathOperator{\diag}{diag}

\DeclareMathOperator{\col}{col}

\begin{document}

\title{Expanding Hardware-Efficiently Manipulable Hilbert Space via Hamiltonian Embedding}

\author[1,3,4]{Jiaqi Leng$^*$}
\thanks{jiaqil@berkeley.edu}
\author[2,3]{Joseph Li$^*$}
\author[2,3]{Yuxiang Peng}
\author[2,3]{Xiaodi Wu}
\thanks{xiaodiwu@umd.edu}
\affil[1]{Department of Mathematics, University of Maryland, College Park, USA}
\affil[2]{Department of Computer Science, University of Maryland, College Park, USA}
\affil[3]{Joint Center for Quantum Information and Computer Science, University of Maryland}
\affil[4]{Department of Mathematics and Simons Institute for the Theory of Computing, University of California, Berkeley}
% \affil[*]{These authors contributed equally to this work.}
% \affil[$\dagger$]{\href{mailto:xiaodiwu@umd.edu}{xiaodiwu@umd.edu}}

\maketitle
\begin{abstract}
\def\thefootnote{*}\footnotetext{Equal contribution.}
    Many promising quantum applications depend on the efficient quantum simulation of an exponentially large sparse Hamiltonian, a task known as sparse Hamiltonian simulation, which is fundamentally important in quantum computation. Although several theoretically appealing quantum algorithms have been proposed for this task, they typically require a black-box query model of the sparse Hamiltonian, rendering them impractical for near-term implementation on quantum devices.

    In this paper, we propose a technique named \emph{Hamiltonian embedding}. This technique simulates a desired sparse Hamiltonian by embedding it into the evolution of a larger and more structured quantum system, allowing for more efficient simulation through hardware-efficient operations. We conduct a systematic study of this new technique and demonstrate significant savings in computational resources for implementing prominent quantum applications. As a result, we can now experimentally realize quantum walks on complicated graphs (e.g., binary trees, glued-tree graphs), quantum spatial search, and the simulation of real-space Schr\"odinger equations on current trapped-ion and neutral-atom platforms. Given the fundamental role of Hamiltonian evolution in the design of quantum algorithms, our technique markedly expands the horizon of implementable quantum advantages in the NISQ era.
\end{abstract}

\newpage
\setcounter{tocdepth}{2}
\tableofcontents

\newpage
\section{Introduction}
Quantum technology has achieved significant milestones, notably demonstrating quantum advantage over classical computers~\cite{arute2019quantum,zhong2020quantum}, and the realization of error-corrected digital qubits~\cite{bluvstein2023logical}. Meanwhile, various quantum algorithms targeting practical applications have emerged, including solving large-scale linear systems~\cite{harrow2009quantum}, simulating linear and nonlinear dynamics~\cite{an2023quantum,liu2021efficient}, and finding optimal solutions for mathematical optimization problems~\cite{leng2023quantum}. 
At the core of these quantum applications lies a fundamental subroutine named \emph{sparse Hamiltonian simulation}, i.e., to compute the matrix function $f(A) = e^{-iAt}$ for a given sparse Hermitian matrix $A$.
Over the past few decades, sparse Hamiltonian simulation has been a central topic in quantum computation research, and a rich variety of quantum algorithms (e.g.,~\cite{berry2007efficient,berry2012black,berry2014exponential,berry2015simulating,low2017optimal, gilyen2019quantum}) have been developed for this task. While these algorithms achieve exponential quantum speedups over known classical ones, their implementation requires executing deep quantum circuits. 
The absence of fault tolerance in near-term devices hinders our ability to utilize these advanced quantum algorithms for solving classically intractable computational tasks in science and business~\cite{beverland2022assessing}.

Delving into the low-level implementation of quantum simulation algorithms reveals two prominent factors that contribute to the depth of quantum circuits.
First, almost all sparse Hamiltonian simulation algorithms require a sophisticated quantum input model for the sparse matrix $A$, such as quantum oracles for sparse matrices~\cite{aharonov2003adiabatic,berry2007efficient,childs2011simulating}, QRAM~\cite{giovannetti2008quantum}, or block-encodings~\cite{low2018hamiltonian,gilyen2019quantum}. Implementing these input models usually introduces a significant gate count, even in relatively simple cases~\cite{camps2024explicit}.
Second, in mapping a quantum algorithm to executable quantum gates, typical circuit compilation routines begin with representing the algorithm using standardized gates, such as CNOT and Toffoli. The gate sequence is further decomposed into native 1- and 2-qubit gates~\cite{Qiskit, sorensen1999quantum, caldwell2018parametrically}. Failing to fully leverage the native programmability of the target quantum hardware, this compilation workflow often leads to significant practical overheads.

In this paper, we would like to explore an alternative approach to sparse Hamiltonian simulation that exploits both the sparsity structure of the input data and the resource efficiency within the underlying quantum hardware.
To this end, we first introduce a Hamiltonian-based model of quantum algorithms that encompass all \emph{hardware-efficient} operations in a specific quantum computer. Unlike quantum circuits, which act as a mathematical model agnostic to the underlying hardware, our new Hamiltonian model integrates low-level machine information to minimize the cost of end-to-end implementation.
Based on this new model, we have developed a technique named \emph{Hamiltonian embedding}, which enables us to simulate a sparse Hamiltonian as part of a larger quantum evolution that can be efficiently simulated on quantum hardware. With this technique, we can build input models directly with quantum Hamiltonians and perform sparse Hamiltonian simulations with extremely limited quantum resources.

Our Hamiltonian model of quantum computers is motivated by the low-level control logic of quantum hardware.
In general, physically realizable quantum platforms, such as transmon qubits~\cite{wallraff2004strong,devoret2013superconducting}, trapped ions~\cite{blatt2012quantum,monroe2021programmable}, neutral atoms~\cite{Henriet2020quantumcomputing}, are described as quantum systems whose evolution is governed by a quantum Hamiltonian involving 1- and 2-body interactions. 
For example, the emulation of analog quantum simulators, such as quantum annealers~\cite{johnson2011quantum,harris2010experimental} and Rydberg atom arrays~\cite{ebadi2021quantum,ebadi2022quantum}, is governed by a system Hamiltonian composed of 1- and 2-local Pauli operators.
In the case of digital quantum computers, their native continuously parameterized 1- and 2-qubit gates are typically specified by effective generating Hamiltonians.\footnote{For example, the arbitrary angle M\o lmer-S\o renson gate $\mathrm{MS}(\theta)=e^{-i(\theta/2) XX}$ in trapped-ion devices are realized by engineering the effective Hamiltonian $XX$ in the physical device for a variable time~\cite{molmer1999multiparticle,solano1999deterministic}.}
We propose to use the following quantum Hamiltonian to model hardware-efficient quantum operations: 
\begin{align}\label{eqn:ham-model}
    H(t) = \sum_j \alpha_j(t) H_j + \sum_{j,k} \beta_{j,k}(t) H_{j,k},
\end{align}
where $H_j$ (or $H_{j,k}$) is a Hamiltonian on qubit $j$ (or $j$ and $k$) that corresponds to a native operation in specific hardware, $\alpha_j(t)$ and $\beta_{j, k}(t)$ are time-dependent functions.
This Hamiltonian model can represent any implementable quantum circuits through piecewise-constant $\alpha_j$ and $\beta_{j, k}$.\footnote{An implementable quantum circuit is composed of a sequence of native 1- and 2-qubit gates. Each execution of a native gate can be represented as a rectangular control signal in the time-dependent functions.}
On the other hand, any unitary operations generated by the Hamiltonian $H(t)$ with smooth time dependence can be readily implemented by a sequence of native gates via product formulas~\cite{childs2021theory}.

To illustrate the idea of Hamiltonian embedding, we consider a simple example. Suppose that we have two (time-independent) Hamiltonian operators $A$ and $H$, and $H$ contains $A$ as a diagonal block such that $H = \diag(A, *)$, where $*$ represents another (possibly different) Hamiltonian block that evolves independently within its own subspace and is therefore irrelevant to the simulation of $A$. Then, the time evolution generated by $H$ is also block-diagonal,
\begin{align}
    e^{-iHt} = \begin{bmatrix}
        e^{-iAt} & 0\\
        0 & *
    \end{bmatrix},
\end{align}
where the upper left block is the time evolution of $A$. In other words, we can simulate a target Hamiltonian $A$ by embedding it into a larger block-diagonal Hamiltonian $H$.
Generalizing this intuition to ``approximately'' block-diagonal Hamiltonians, we have developed a formalism of Hamiltonian embedding with rigorous error analysis, see \thm{main}. Now, if the parameters in the Hamiltonian model $H(t)$ are programmed in a way such that $H(t)$ embeds $A$, we end up with a \emph{de facto} input model that allows efficient quantum simulation.\footnote{An $n$-qubit Hamiltonian model has at most $O(n^2)$ 1- and 2-local component Hamiltonians, where each component Hamiltonian can be natively simulated in the associated quantum hardware. Then, we can use product formulas to simulate the time evolution generated by $H(t)$ with $\poly(n)$ quantum resources.}

In theory, a Hamiltonian embedding can be constructed for any sparse matrix, as detailed in \sec{sparse-ham}. 
For a general $n$-dimensional sparse matrix without specific structures, an embedding may require $n$ qubits and $O(n^2)$ local interaction terms. In this case, Hamiltonian embedding could still offer polynomial speedups, although the actual advantage depends heavily on the problem class and requires case-by-case investigation.
% However, for large matrices, this construction process may incur an exponential cost, casting doubt on the feasibility of achieving quantum speedup in such cases.
Fortunately, we managed to identify various scenarios, including high-dimensional graphs created through graph product operations and specific linear differential operators, where the Hamiltonian embedding can be constructed using quantum resources that scale logarithmically in the input size $n$. This leads to exponential quantum speedups using our methodology.
When certain algebraic structures, such as addition, multiplication, composition, and tensor product, emerge in a sparse Hamiltonian $A$, we can decompose its Hamiltonian embedding into a small number of basic building blocks (see \thm{rules}). We then provide six embedding schemes that can be employed to construct elementary building blocks of the full Hamiltonian embedding. These embedding schemes work for matrices with particular sparsity patterns (e.g., band, banded circulant, $s$-sparse), as detailed in~\sec{sparse-ham}. It's important to note that many of these embedding schemes are already documented in existing literature, albeit under different names. Our contribution lies in developing a unifying formalism, which facilitates the systematic application of these schemes specifically for simulating large sparse Hamiltonian.

A key feature of Hamiltonian embedding is that it directly harnesses hardware-efficient operations to build the input model, which significantly reduces the quantum resources needed in Hamiltonian simulation tasks and quantum algorithms based on these. This technique enables us to implement several experiments on existing open-access cloud-based quantum computers, showcasing interesting quantum applications (see \sec{exp}). In contrast, Hamiltonian simulation methods using traditional query models (such as quantum oracles or block-encodings) are impractical on these quantum computers. Even a single implementation of the query model would deplete the available quantum resources~\cite{camps2024explicit}. A detailed discussion on different quantum input models is provided in \sec{review-input-model}.
It is worth noting that our Hamiltonian embedding technique can also be adapted for analog quantum simulators,\footnote{An experimental demonstration with Rydberg atom arrays is provided in \sec{sim-real-space-dynamics}.} unlike existing sparse Hamiltonian simulation methods which are tailored exclusively for gate-based quantum computers. This adaptability broadens the possibilities for analog quantum computation.

In the experiments, we also provide comprehensive resource analyses (in terms of native gate count) to better understand the empirical scaling of Hamiltonian embeddings regarding various problem sizes. For comparison, we also estimate the quantum resources needed to perform the same Hamiltonian simulation tasks using the basic Pauli access model, which directly represents sparse Hamiltonians as a sum of Pauli operators, without embedding (see \sec{review-input-model} for details). While both Hamiltonian embedding and the Pauli access model use the Hamiltonian as inputs, the latter is unaware of the machine programmability and it utilizes the full Hilbert space to represent the sparse data.\footnote{We remark that obtaining the Pauli access model for sparse matrices could require exponential classical pre-processing time. In the resource analysis, we assume the Pauli access model of a given $A$ is already known.} Due to these differences, we use the Pauli access model as a baseline to investigate the resource efficiency of our technique.
It turns out that there always exists at least one Hamiltonian embedding that is more resource-efficient than the baseline in both asymptotic (i.e., scaling in system size) and non-asymptotic (i.e., actual gate count) metrics, as detailed in the panel B of Figure \ref{fig:fig_2}, \ref{fig:fig_3}, and \ref{fig:fig_4}.
Our findings indicate that Hamiltonian embedding, despite increasing the size of the global Hilbert space, provides a means to expand the hardware-efficiently manipulable Hilbert space. This leads to quantum simulations that are more resource-efficient compared to traditional hardware-agnostic approaches.

\paragraph{Contribution.}
Our main contributions are threefold.
First, we propose to use the quantum Hamiltonian to model native operations in a quantum computer. This new model enables efficient Hamiltonian simulations without going through a hardware-agnostic compilation process.
Second, we develop a new technique named Hamiltonian embedding for hardware-efficient sparse Hamiltonian simulation. We provide a general framework with rigorous error analysis and a flexible construction approach with concrete instances. Given some special structures in the sparse Hamiltonian simulation tasks, Hamiltonian embedding allows us to achieve up to exponential quantum speedups on a small fault-tolerant quantum computer.
Last but not least, we showcase our technique to realize some interesting quantum applications that are almost infeasible with traditional Hamiltonian simulation methods. Comprehensive resource analysis shows that our methodology demonstrates both empirical and asymptotic advantages compared to the baseline method.

\paragraph{Code Availability.} 
The source codes of the real-machine experiments and resource analysis are available at \url{https://github.com/jiaqileng/hamiltonian-embedding}.

\subsection{Review: quantum input models}\label{sec:review-input-model}

In this subsection, we briefly review three commonly used quantum input models for Hamiltonian simulation and numerical linear algebra.

\paragraph{Sparse-input oracle.}
Let $A$ be a matrix that is $s$-sparse, i.e., every row or column of $A$ has at most $s$ nonzero elements. The sparse-input oracles of $A$ refer to a procedure (implemented by quantum circuits) that can perform the following mappings:
$$O_r \colon \ket{i}\ket{k}\to \ket{i}\ket{r_{ik}},\quad O_c \colon \ket{\ell}\ket{j}\to \ket{c_{\ell j}}\ket{j},$$
$$O_A\colon \ket{i}\ket{j}\ket{0}^{\otimes b} \to \ket{i}\ket{j}\ket{a_{ij}},$$
where $r_{ik}$ is the index for the $k$-th nonzero entry of the $i$-th row of $A$, $c_{\ell j}$ is the index for the $\ell$-th nonzero entry of the $j$-th column of $A$, and $a_{ij}$ is a $b$-bit binary description of the $(i,j)$-matrix element of $A$.
This black-box query model was first considered by Aharonov and Ta-Shma~\cite{aharonov2003adiabatic}, then has been widely assumed in quantum algorithms for Hamiltonian simulation~\cite{berry2007efficient,berry2012black,childs2011simulating,low2017optimal} and numerical linear algebra~\cite{childs2017quantum,gilyen2019quantum}. However, even for sparse matrices with regular sparsity patterns and a small number of nonzero elements, implementing sparse-input oracles by efficient quantum circuits is a highly nontrivial task~\cite{camps2024explicit}.

\paragraph{Block-encoding.}
A unitary $U_{A}$ is a block-encoding of a matrix $A$ if 
$$\left\|A - \alpha (\bra{0}^{\otimes a} \otimes I) U_A (\ket{0}^{\otimes a} \otimes I)\right\| \leq \epsilon,$$
where $\alpha$ is a normalization factor, $a$ is the number of ancilla qubits, and $\epsilon$ is an error parameter.
Block-encoding arises naturally as an input model for algorithms based on quantum signal processing (QSP)~\cite{low2017optimal} and the quantum singular value transformation (QSVT)~\cite{gilyen2019quantum}.
However, it is not possible to construct circuits for block-encoding arbitrary sparse matrices with space and time complexity both logarithmic in the matrix dimension~\cite{zhang2024circuit}.
Efficient circuit constructions for block-encoding have only been studied for certain structured matrices~\cite{camps2024explicit, sunderhauf2024block}, but these constructions require sequences of multi-qubit controlled gates which are impractical for implementation on current devices.
While the first steps towards implementing QSP have been demonstrated for a small-scale problem~\cite{kikuchi2023realization}, due to the high cost of implementing a block-encoding, it is generally expected that the scalable implementation of QSP-based algorithms will only be possible in the deep fault-tolerant regime.

\paragraph{Pauli access model.}
The Pauli access model of an $n$-qubit Hamiltonian $A$ assumes that $A$ can be specified as a linear combination of Pauli operators $Q_s$,
\begin{align}\label{eqn:pauli-sum}
    A = \sum_s a_s Q_s,
\end{align}
where $\{Q_s\}$ is the set of all $n$-qubit Pauli operators, and the coefficients $a_s$ can be computed by $a_s = \frac{1}{2^n}\Tr[AQ_s]$. In the literature, the representation \eqn{pauli-sum} is sometimes referred to as the standard binary encoding of $A$, e.g., see~\cite{sawaya2020resource,sawaya2022mat2qubit}.
This input model has been considered in some variational quantum algorithms~\cite{bravo2023variational,xu2021variational} and randomized algorithms for linear algebra~\cite{wan2022randomized,wang2024qubit}.
Unlike sparse-input oracles and block-encodings, the Pauli access model does not require a coherent quantum circuit implementation. However, for general sparse Hamiltonians, computing its Pauli operator decomposition is not scalable because the coefficient element $a_s$ in \eqn{pauli-sum} requires evaluation of the trace $\Tr[AQ_s]$, which could run for an exponentially long time on classical computers.
Also, evolving a Pauli operator that involves more than 2-site interactions (e.g., $XYX$) is expensive on gate-based quantum computers since the resulting unitary operators need to be decomposed to native 1- and 2-qubit gates.
As a comparison, the construction of Hamiltonian embedding is scalable since it does not require computing the trace of large matrices. When the matrix $A$ has certain sparsity patterns (e.g., band, banded circulant, etc.), the resulting Hamiltonian embedding only involves $2$-site interactions and thus is readily implementable on physical hardware without further decomposition.

\subsection{Relevant work}
\paragraph{Hamming encoding in Quantum Hamiltonian Descent.}
Recently, Leng et al.~\cite{leng2023quantum} proposed a quantum optimization algorithm named Quantum Hamiltonian Descent (QHD). QHD addresses continuous optimization problems by simulating quantum evolution. In addition to providing a theoretical analysis of the quantum algorithm, the authors also developed an analog implementation of QHD, which they termed \emph{Hamming encoding}. The Hamming encoding method can be seen as the initial instance of Hamiltonian embedding, in which the quantum algorithm is directly executed with a quantum Hamiltonian, rather than any existing quantum input models. Our current work expands this concept into a formalized approach, encompassing more constructions and a wider range of applications.

\paragraph{Encodings of quantum operators.}
Specific constructions of Hamiltonian embedding similar to our own have been extensively studied in different contexts.
The facilitation or antiblockade phenomenon has been investigated via quantum dynamics confined to certain subspaces in arrays of Rydberg atoms~\cite{marcuzzi2017facilitation, ostmann2019synthetic, liu2022localization}.
The encoding of $d$-level quantum systems (i.e., qudits) within a multi-qubit Hamiltonian has also been explored for various encodings~\cite{sawaya2020resource, sawaya2020connectivity, kyaw2021quantum, kyaw2023variational}, such as the standard binary, Gray, and one-hot codes. 
A software package named \verb|mat2qubit|~\cite{sawaya2022mat2qubit} automates the compilation of these encoding schemes, and in the case of the one-hot code, yields encoded operators identical to our (penalty-free) construction. 
Compared to our construction, these schemes generally lead to more complicated encoding operators because they completely disallow leakage to the orthogonal complement.
The $XY$ model has long been observed to give rise to quantum walk via the one-hot code \cite{christandl2005perfect}, and restrictions to higher excitation subspaces have also been studied for the task of permuting a quantum state \cite{albanese2004mirror}.
In the context of quantum adiabatic optimization~\cite{chancellor2019domain, hadfield2019quantum}, the unary and one-hot codes have been studied to encode combinatorial optimization problems in the ground-energy subspace of certain penalty operators. While these optimization works do not consider the task of Hamiltonian simulation, a similar penalty operator is utilized in our constructions of Hamiltonian embeddings.
From this perspective, Hamiltonian embedding can be viewed as a unifying framework that encompasses several well-studied encodings both with and without penalty.

\paragraph{Explicit construction of block-encodings for sparse matrices.}
The framework of quantum signal processing (QSP) has been shown to give an optimal algorithm for sparse Hamiltonian simulation~\cite{low2017optimal}.
In general, QSP-based algorithms require a block-encoding oracle of the Hamiltonian which is nontrivial to construct in general.
Explicit circuit constructions of block-encodings have only been obtained for specific matrices~\cite{gilyen2019quantum, camps2024explicit, sunderhauf2024block}.
In particular, \cite{camps2024explicit} constructs circuits for block-encoding tri-diagonal and banded circulant matrices, for which we also consider Hamiltonian embeddings in this paper.
Nevertheless, the circuits for block-encoding are constructed using multi-qubit controlled gates which require further decomposition into elementary one- and two-qubit gates.
While the overall gate complexities are polylogarithmic in the matrix dimension (and thus considered efficient), the actual gate counts required even for a single oracle call are prohibitively expensive in practice. 
For instance, the circuit for block-encoding an $8 \times 8$ banded circulant matrix requires roughly 171 one-qubit gates and 114 two-qubit gates when compiling to Pauli-$X$, Pauli-$Y$, and $XX$ rotations.
On the other hand, Hamiltonian embedding provides an alternative approach that avoids such oracle constructions and enables near-term implementation on noisy quantum computers.

\section{Hamiltonian embedding}\label{sec:ham-embedding}
\subsection{General formulation}
Let $A$ be an $n$-dimensional Hermitian operator and $H$ be a $q$-qubit operator.

\begin{definition}[Hamiltonian embedding]\label{defn:embedding}
     Let $\eta, \epsilon > 0$ be positive scalars. We say $H$ is a $(q, \eta, \epsilon)$-\textbf{embedding} of $A$ if there exists a subspace $\S \subset \C^{2^q}$ and a unitary operator $U$ such that
    \begin{enumerate}
        \item $P_\S(U^\dagger H U)P_{\S^\perp} = 0$, i.e., $U^\dagger H U$ is block-diagonal in $\S$ and $\S^\perp$,
        \item $\|I - U\| \le \eta$, where $I$ is the identity operator in $\C^{2^q}$,
        \item $\|(U^\dagger H U)|_{\S} - A\| \le \epsilon$, where $(\cdot)|_{\S}\coloneqq P_\S (\cdot)P_\S$.
    \end{enumerate} 
We call the subspace $\S$ as the \textbf{embedding subspace}.
\end{definition}

By the definition, the operator $H$ is approximately block-diagonal (up to a minor basis change $U$) with respect to the embedding subspace $\S$ and its orthogonal complement $\S^\perp$. The target Hamiltonian $A$ is embedded in the upper left block of $U^\dagger H U$ up to an additive error $\epsilon$. 
Clearly, a $q$-qubit Hamiltonian $H$ is a $(q, 0, 0)$-embedding of itself.
For sufficiently small $\eta$ and $\epsilon$, we show that the embedding Hamiltonian $H$ simulates the time evolution generated by the target Hamiltonian $A$. The proof is given in \append{proof-main}.

\begin{restatable}[Hamiltonian simulation with Hamiltonian embedding]{theorem}{MainThm}\label{thm:main}
    Suppose that $H$ is a $(q, \eta, \epsilon)$-embedding of $A$. Then, for a fixed evolution time $t \ge 0$, we have that
    \begin{align}\label{eqn:error_bound}
        \left\|\left(e^{-iH t}\right)\Big|_{\S} - e^{-iAt}\right\| \le (2\eta\|H\|+\epsilon)t.
    \end{align}
\end{restatable}

In addition, we find that a complicated Hamiltonian embedding can be built from simpler ones through the four composing rules, including addition, multiplication, composition, and tensor product. In what follows, we give an informal version of these rules. See \append{building} for a formal restatement and proof.

\begin{theorem}[Rules for building Hamiltonian embeddings; informal]\label{thm:rules}~
\begin{enumerate}
    \item (Addition) For $j = 1,2$, let $H_j$ be a $(q,\eta,\epsilon_j)$-embedding of $A_j$, then $H_1+H_2$ is a $(q,\eta,\epsilon_1+\epsilon_2)$-embedding of $A_1+A_2$.
    \item (Multiplication) Let $H$ be a $(q,\eta,\epsilon_j)$-embedding of $A$, then for a real scalar $\alpha$, $\alpha H$ is a $(q,\eta,|\alpha|\epsilon)$-embedding of $\alpha A$.
    \item (Composition) For $j = 1,2$, let $H_j$ be a $(q_j,\eta_j,\epsilon_j)$-embedding of $A_j$, then $H_1\otimes I+I\otimes H_2$ is a $(q_1+q_2,\eta_1+\eta_2,\epsilon_1+\epsilon_2)$-embedding of $A_1\otimes I+I\otimes A_2$.
    \item (Tensor product) For $j = 1,2$, let $H_j$ be a $(q_j,\eta_j,\epsilon_j)$-embedding of $A_j$, then $H_1\otimes H_2$ is a $(q_1+q_2,\eta_1+\eta_2,\|A_1\|\epsilon_2+\|A_2\|\epsilon_1+\epsilon_1\epsilon_2)$-embedding of $A_1\otimes A_2$.
\end{enumerate}
\end{theorem}

\subsection{Perturbative Hamiltonian embedding}
As we have seen, a target Hamiltonian $A$ can be simulated with a block-diagonal embedding Hamiltonian. However, identifying such a corresponding block-diagonal embedding could be non-trivial for many sparse Hamiltonians $A$ arising from real-world applications. Here we give an explicit construction of Hamiltonian embedding based on perturbation theory.

First, we assume there is a $q$-qubit quantum operator $Q$ and a subspace $\S$ such that $Q|_{\S} = A$. We express the operator $Q$ in a block-matrix form by projecting it down to the subspaces $\S$ and $\S^\perp$, respectively,
\begin{align}
    Q = \begin{bmatrix}
        A & R^\dagger\\
        R & B
    \end{bmatrix},
\end{align}
where $A = P_{\S} Q P_{\S}$, $R = P_{\S^\perp}Q P_{\S}$, $B= P_{\S^\perp} Q P_{\S^\perp}$. Since the off-diagonal block $R$ is not necessarily zero, we can not directly simulate $A$ by evolving the Hamiltonian $Q$ because in general $P_\S e^{-iQt}P_\S \neq e^{-iAt}$. This means that a quantum state initialized in the subspace $\S$ could be driven out of this subspace in the course of quantum evolution.

To suppress the leakage from the embedding subspace $\S$, we introduce another $q$-qubit operator $\Hpen$ as the \emph{penalty Hamiltonian}. We assume $\Hpen$ has an $n$-fold degenerate ground-energy subspace $\S$ with the ground energy being zero, i.e., the least eigenvalue of $\Hpen$ is zero. For a fixed positive number $g > 0$, we define the following quantum operator,
\begin{align}\label{eqn:perturbative-embedding}
    H = g\Hpen + Q.
\end{align}
Similarly, this new operator $H$ can be expressed in a block-matrix form (we denote $C = P_{\S^\perp}\Hpen P_{\S^\perp}$ and $G=B+gC$):
\begin{align}\label{eqn:decompo}
    H = \begin{bmatrix}
        A & R^\dagger\\
        R & B+gC
    \end{bmatrix} = \underbrace{\begin{bmatrix}
        A & 0\\0 & G\end{bmatrix}}_{\text{diagonal}} + \underbrace{\begin{bmatrix}
        0 & R^\dagger\\R & 0\end{bmatrix}}_{\text{off-diagonal}},
\end{align}
The Hamiltonian $H$ decomposes into an off-diagonal part and a diagonal part. For sufficiently large $g > 0$, there is a gap between the spectrum of $A$ and that of $G$ and the width of this gap is proportional to $g$, i.e., 
\begin{align}
    \Delta = \lambda_{\min}(G) - \lambda_{\max}(A) > 0,\quad \Delta \sim g.
\end{align}
When $g$ is large enough such that $\|R\|/g \ll 1$, the off-diagonal part in \eqn{decompo} can be treated as a perturbative term. In this case, we prove that $H$ is an embedding of $A$, as shown in the following theorem. A formal statement and the proof are provided in \append{perturbation}.

\begin{theorem}[Perturbative Hamiltonian embedding; informal version of \thm{dir-rotation}]\label{thm:perturbative-thm}
    Let the operators $H$, $A$, and $R$ be the same as above. For sufficiently large $g > 0$, the Hamiltonian $H$ is a $(q,\eta,\epsilon)$-embedding of $A$, where $\eta \sim \|R\|/g$, $\epsilon \sim \|R\|^2/g$.
\end{theorem}

Given that $\|R\|\le \|H\|$, this result immediately implies that the simulation error using perturbative Hamiltonian embedding is of the order $\mathcal{O}(\|R\|\|H\|t/g)$, where the error bound follows from \thm{main}. This error bound can be improved to $\mathcal{O}(\|R\|^2t/g)$ by leveraging the particular operator splitting structure as shown in \eqn{decompo}, see \thm{perturb}. In practice, the embedding subspace $\S$ is often much smaller than the full Hilbert space; consequently, it is often sufficient to assume $\|R\| = O(1)$.
Therefore, to achieve a simulation error $\delta > 0$, we need to choose a penalty coefficient $g = \mathcal{O}(t/\delta)$. However, this choice of $g$ would incur a $\poly(1/\delta)$ overhead in the overall gate complexity when we implement the Hamiltonian simulation on gate-based quantum computers using standard product formulas. Fortunately, the penalty Hamiltonian $\Hpen$ is usually fast-forwardable, i.e., the time-evolution operator $e^{-i\Hpen t}$ can be simulated using quantum resources that scale sub-linear in $t$. In this case, we find the unfavorable $\poly(1/\delta)$ overhead can be mitigated by utilizing interaction-picture quantum simulation algorithms like continuous qDRIFT~\cite{berry2020time}. See~\append{sim_perturbative} for a detailed discussion.

\begin{figure}[ht!]
    \centering
    \includegraphics[width=16cm]{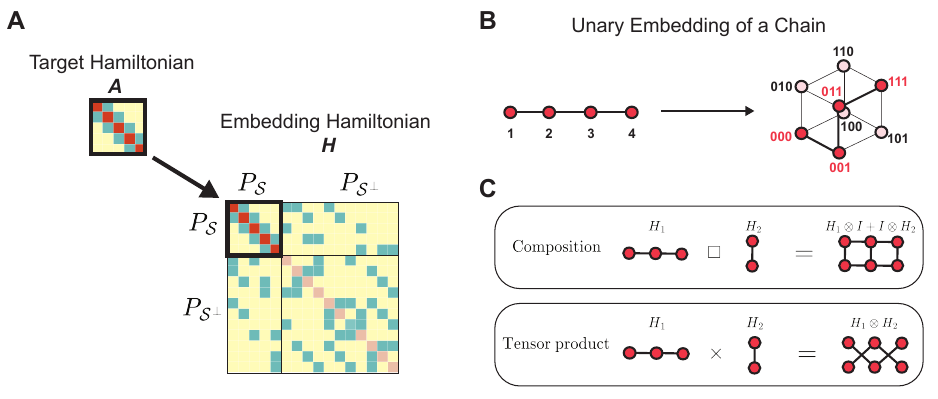}
    \caption{\small \textbf{Hamiltonian embedding.}
    \textbf{A.}
    Illustration of perturbative Hamiltonian embedding.
    The dynamics of a target Hamiltonian $A$ are embedded within a subspace $\S$ in a larger hardware Hilbert space.
    With respect to $\S$ and $\S^\perp$, $H$ is block-diagonal with $A$ being the upper-left block located by the projector $P_{\S}$.
    \textbf{B.}
    Graphical illustration of Hamiltonian embedding for the case of a chain (i.e. tridiagonal matrix) using the unary embedding.
    From a graph theoretic perspective, a path graph is located as a subgraph of a hypercube graph. Red vertices indicate basis states in the embedding subspace $\S$, while pink vertices indicate basis states in $\S^{\perp}$.
    \textbf{C.}
    Visualization of composition and tensor product rules for constructing Hamiltonian embeddings.
    When considering the graph of nonzero entries of a Hamiltonian, these rules are analogous to two different notions of graph products (see \rem{graph_vis} in the appendix for more details).}
    \label{fig:schematic}
\end{figure}

\subsection{Hamiltonian embedding of sparse matrices}\label{sec:sparse-ham}
For Hermitian matrices with certain special sparsity structures, we can construct explicit Hamiltonian embeddings for them. In what follows,
we provide 6 different Hamiltonian embedding schemes for 3 types of sparse Hermitian matrices: 
\begin{enumerate}
    \item \textbf{Band matrix}: an $n$-by-$n$ matrix $A$ is a band matrix of bandwidth $d$ if $A_{i,j} = 0$ for any $i,j=1,\dots, n$ such that $|i-j|>d$.
    \item \textbf{Banded circulant matrix}: an $n$-by-$n$ matrix $A$ is a banded circulant matrix of bandwidth $d$ if \emph{(i)} in the first row, $A_{1,j} = 0$ for any $j = d+2,\dots,n-d$, and \emph{(ii)} all row vectors are composed of the same elements and each row vector is rotated one element to the right relative to the preceding row vector.
    \item \textbf{$s$-sparse matrices}: an $n$-by-$n$ matrix $A$ is $s$-sparse if each row/column of $A$ has at most $s$ non-zero elements.
\end{enumerate}

In \tab{sparse-embedding}, we list all the embedding schemes discussed in this paper that can be applied to simulate at least one type of sparse Hermitian matrices. 
Among our embedding schemes, the unary and one-hot encodings are well-known~\cite{chancellor2019domain, hadfield2019quantum, sawaya2020resource, sawaya2022mat2qubit, christandl2005perfect}, while the others have not been systematically studied in the existing literature.
Note that, except for the last item (``Penalty-free one-hot''), all embeddings are perturbative Hamiltonian embeddings and thus require a penalty Hamiltonian.\footnote{For perturbative Hamiltonian embedding, the simulation error can be made arbitrarily small by increasing the penalty coefficient $g$. Thus, for simplicity, we do not explicitly specify the tuple $(q,\eta,\epsilon)$ for each Hamiltonian embedding listed in the table.} The ``Max. weight'' column shows the maximal weight of the Pauli operators involved in the corresponding embedding Hamiltonian ($d$ represents the bandwidth of the target Hamiltonian $A$). Full details of these embedding schemes can be found in \append{sparse_matrix}.

\begin{table}[!ht]
    \centering
    \begin{tabular}{|p{4cm}|p{4cm}|p{3cm}|p{3cm}|}
     \hline
     \textbf{Embedding scheme} & \textbf{Applications} & \textbf{Max. weight} & \textbf{Details}\\
     \hline
     Unary & band & $\max(d,2)$ & \append{unary}\\
     \hline
     Antiferromagnetic & band & $\max(d,2)$ & \append{antiferro}\\
     \hline
     Circulant unary & banded circulant & $\max(d,2)$ &  \append{circ-unary}\\
     \hline
     Circulant antiferromagnetic & banded circulant & $\max(d,2)$ & \append{circ-antiferro}\\
     \hline
     One-hot & band, banded circulant, $s$-sparse & $2$ & \append{onehot}\\
     \hline
     Penalty-free one-hot & band, banded circulant, $s$-sparse & $2$ & \append{onehot_penfree}\\
     \hline
    \end{tabular}
    \caption{Six Hamiltonian embeddings for sparse Hamiltonian simulation}
    \label{tab:sparse-embedding}
\end{table}

A notable feature of these embedding schemes is that their maximal weight depends on the special sparsity pattern of the target Hamiltonian $A$. In particular, for band matrices (or banded circulant matrices) with bandwidth $d \le 2$, the corresponding unary/antiferromagnetic (or circulant unary/circulant antiferromagnetic) embeddings have a maximal weight of $2$, which means these embedding Hamiltonians can be directly simulated using native gates. This is also true for one-hot/penalty-free one-hot embeddings that encode arbitrary $s$-sparse matrices.

We also note that all these schemes utilize only an exponentially small subspace of the full Hilbert space for embedding. Therefore, direct use of these schemes to build Hamiltonian embeddings does not lead to meaningful quantum speedups in Hamiltonian simulation. To achieve exponential quantum speedups, one strategy is to first decompose the target sparse Hamiltonian using the \emph{composition} and \emph{tensor product} rules in \thm{rules}, and then use the embedding schemes listed in \tab{sparse-embedding} to construct elementary building blocks of size $O(1)$. Using this strategy, we can build Hamiltonian embeddings for symmetric/Hermitian matrices arising from graph theory and differential equations with logarithmic quantum resources, as detailed in~\sec{exp}.

Now, we provide some concrete examples of the embedding of a small sparse matrix. 
We consider a $5$-by-$5$ tridiagonal matrix $A$ given by 
\begin{align}\label{eqn:5-node-chain}
    A = \begin{bmatrix}
-1 & 1 & 0 & 0 & 0\\
1 & -2 & 1 & 0 & 0\\
0 & 1 & -2 & 1 & 0\\
0 & 0 & 1 & -2 & 1\\
0 & 0 & 0 & 1 & -1\\
\end{bmatrix},
\end{align}
which is the Laplacian matrix of a $5$-node chain graph.
In \tab{embedding_tridiagonal}, we list the operators $Q$ and $\Hpen$ for embedding $A$ using various embedding schemes (with more details available in \append{qwalk}).
In this paper, we use little-endian ordering, meaning that we enumerate the bit of a string \textit{from right to left}.
Here, $X = \begin{bmatrix} 0 & 1 \\ 1 & 0 \end{bmatrix}$ and $Z = \begin{bmatrix} 1 & 0 \\ 0 & -1\end{bmatrix}$ are the Pauli-$X$ and Pauli-$Z$ matrices, and $\hat{n} = \frac{1}{2}(I-Z)$.
The overall embedding Hamiltonian is $H=g\Hpen + Q$, where $g>0$ is a sufficiently large penalty coefficient.
The embedding subspace $\S$ depends on the embedding scheme.
For sufficiently large $g > 0$, \thm{perturbative-thm} implies that $e^{-iAt} \approx P_\S e^{-iH_A t} P_\S$.

\begin{table}[!ht]
    \centering
    \begin{tabular}{|p{4cm}|p{6cm}|p{4cm}|}
        \hline
        \textbf{Embedding scheme} & $Q$ & $\Hpen$ \\
        \hline
        Unary & $-\hat{n}_1+\hat{n}_{4} + \sum^{4}_{j=1} X_j$ & $-\sum^{3}_{j=1} Z_{j+1}Z_j + Z_1 - Z_{4}$ \\
        \hline
        Antiferromagnetic & $- \hat{n}_1 - \hat{n}_{4} + \sum^{4}_{j=1} X_j$ & $\sum^{3}_{j=1} Z_{j+1}Z_j + Z_1 + Z_{4}$ \\
        \hline
        One-hot & $\left(-\hat{n}_1 - \hat{n}_5 -2\sum^{4}_{j=2} \hat{n}_j\right)$ $+ \sum^{4}_{j=1} X_{j+1}X_j$ & $\left(\sum_{j=1}^{5} \hat{n}_j - 1\right)^2$\\
        \hline
        Penalty-free one-hot & $\left(-\hat{n}_1 - \hat{n}_5 -2\sum^{4}_{j=2} \hat{n}_j\right)$ $+ \frac{1}{2}\sum^{4}_{j=1} \left(X_{j+1}X_j + Y_{j+1}Y_j\right)$ & 0 \\
        \hline
    \end{tabular}
    \caption{Hamiltonian embeddings of the tridiagonal matrix $A$ in \eqn{5-node-chain}.}
    \label{tab:embedding_tridiagonal}
\end{table}

While there could be several possible Hamiltonian embedding schemes for a fixed target Hamiltonian $A$, we remark that there is no single criterion to determine which one would perform the best. We provide a few aspects to compare various Hamiltonian embeddings in practice. First, the Hamiltonian embedding must match the hardware programmability. For example, the antiferromagnetic embedding scheme fits well with Rydberg atom arrays, see \append{quera-antiferro} for details.
Second, we want to use the Hamiltonian embedding with minimal simulation error. The simulation error could come from the perturbative construction of an embedding and/or a specific Hamiltonian simulation algorithm (e.g., a Trotter formula, continuous qDRIFT, etc.). 

\subsection{Connection to quantum Hamiltonian complexity}
Hamiltonian embedding is closely related to previously studied notions of simulating one Hamiltonian by another Hamiltonian, including the isometry-based definition of \emph{simulation} in~\cite{bravyi2017complexity} used to study the complexity of $2$-local stoquastic Hamiltonians, and \emph{Hamiltonian encodings} used to show the universality of certain spin-lattice models~\cite{cubitt2018universal,zhou2021strongly}.

The simulation introduced in~\cite{bravyi2017complexity} quantifies how close are two quantum Hamiltonians in terms of their low-energy spectrum. It is defined as an isometry transformation and its application has been limited to stoquastic local Hamiltonians~\cite{bravyi2007complexity}, while our Hamiltonian embedding applies to a broader class of sparse matrices.

A Hamiltonian encoding (aka, encoding transformation) is a map that encodes a Hamiltonian $H$ into some other Hamiltonian $H'$. Specifically, this map needs to fulfill a few basic requirements, including the preservation of locality, spectrum, and real-linearity~\cite{cubitt2018universal}. Hamiltonian encoding is more general than simulation since an encoding does not need to be an isometry. The construction of Hamiltonian encodings heavily utilizes perturbative gadgets~\cite{kempe2006complexity, oliveira2008complexity, jordan2008perturbative}, a theoretical tool originally developed to prove hardness results in Hamiltonian complexity theory.

Compared to~\cite{cubitt2018universal}, our definition of Hamiltonian embedding (\defn{embedding}) more closely resembles the \emph{simulation} defined in~\cite{bravyi2017complexity}, in the sense that both do not impose any locality-related restriction on the encoding/embedding maps.
Technically speaking, our perturbative Hamiltonian embedding can be viewed as an explicit construction of perturbative gadgets, but two major differences distinguish our work from prior arts. 
First, existing techniques are almost exclusively applied to (local) quantum Hamiltonians arising from many-body physics (i.e., bosons, fermions, and local spin/qudit Hamiltonians), while our focus is on the simulation of sparse matrices.
Second, Hamiltonian encodings often involve complicated basis change, while our Hamiltonian embedding preserves the energy spectrum and dynamical evolution without introducing a nontrivial basis change. This allows us to measure the embedded system using a subset of computational basis. However, this favorable feature is achieved at a cost of universality. The construction technique in \cite{cubitt2018universal} can be applied to arbitrary local Hamiltonians, while we only provide explicit constructions of Hamiltonian embedding for certain families of sparse Hermitian matrices.

\section{Real-machine experiments}\label{sec:exp}
We conduct experiments to demonstrate the use of Hamiltonian embeddings for computational tasks, including (continuous-time) quantum walk on graphs, spatial search, and simulating real-space quantum dynamics. 
We construct Hamiltonian embeddings for each task and deploy them on current digital and analog quantum computers, including IonQ's ion trap systems~\cite{ionq_api} and QuEra's neutral atom systems~\cite{wurtz2023aquila}.
We also exhibit the efficiency and scalability of our approach over the conventional Pauli access approach (i.e., the standard binary encoding) through a detailed resource analysis.
If given the Pauli decomposition of a Hamiltonian, the Pauli access model enables a straightforward approach to Hamiltonian simulation via product formulas.
On the other hand, the sparse-input and block-encoding input models require coherent circuit implementations of oracles which are highly nontrivial even for structured matrices.
While there has been recent progress for constructing block-encodings~\cite{camps2024explicit, sunderhauf2024block}, the overhead associated with these schemes make them applicable only in the fault-tolerant regime.
Consequently, we compare our embedding schemes with the standard binary encoding, which serves as a more reasonable baseline.
More details on resource analysis are available in \append{methodology}.

\subsection{Hardware-efficient Hamiltonian models of quantum computers}
In this subsection, we give the hardware-efficient Hamiltonian models for the IonQ and QuEra quantum computers. For both devices, their Hamiltonian models can be formulated as
$$H(t) = \sum_j \alpha_j(t) H_j + \sum_{j,k} \beta_{j,k}(t) H_{j,k},$$
but with different native 1- and 2-qubit component Hamiltonians.

The IonQ Aria-1 quantum computer natively supports the GPi gate, GPi2 gate~\cite{wright2019benchmarking}, virtual-Z gate, and arbitrary angle MS (M\o lmer-S\o renson) gate~\cite{ionq_native_gates}. The hardware-efficient Hamiltonian model for the IonQ Aria-1 device is composed of the following 1- and 2-qubit components:
\begin{align}
    H_j\colon aX + bY + cZ,\quad H_{j,k}\colon (\cos(\phi_1)X_{j} + \sin(\phi_1)Y_{j}) \otimes (\cos(\phi_2)X_{k} + \sin(\phi_2)Y_{k}),
\end{align}
where $a, b, c, \phi_1, \phi_2$ can be any real-valued scalars.

The QuEra Aquila quantum computer is an analog quantum simulator. It allows users to program certain parameters, including the Rabi frequency, local detuning, and atom-atom distance, in the effective Hamiltonian describing the Rydberg atom arrays. More details on the QuEra device are provided in \append{quera-antiferro}. We can formulate the abstract model for the QuEra Aquila device using the following 1- and 2-qubit component Hamiltonians:
\begin{align}
    H_j\colon aX + bY + cZ,\quad H_{j,k} \colon \alpha Z\otimes Z,
\end{align}
where $a, b, c$ are real-valued scalars. The parameter $\alpha$ is engineered by Rydberg interactions and thus must be non-negative.

\subsection{Traversing the glued trees graph}

\begin{figure}[ht!]
    \centering
    \includegraphics[width=16cm]{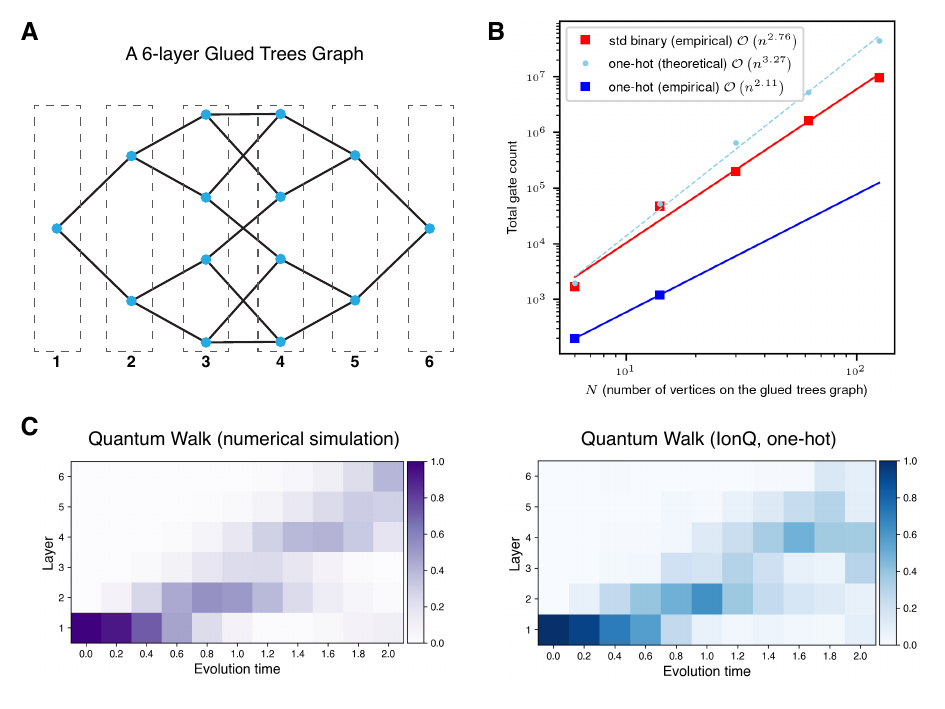}
    \caption{\small \textbf{Quantum walk on glued trees graphs.}
    \textbf{A.} A 6-layer glued trees graph with 14 vertices.
    \textbf{B.} Gate counts of simulating CTQW on $N$-vertex glued trees graphs using standard binary encoding and one-hot embedding. Both the theoretical (worst-case) upper bound and the empirical data suggest that the one-hot embedding has better asymptotic scaling.
    \textbf{C.} Propagation of the quantum wave packet on the 6-layer glued trees graph. Wave packet clustered on each layer. Left: numerical simulation; Right: real-machine results. }
    \label{fig:fig_2}
\end{figure}

We first demonstrate the simulation of continuous-time quantum walk (CTQW) on undirected graphs~\cite{childs2010relationship}. In this section, we focus on the simulation of quantum walk on a glued trees graph, but the embedding and simulation of CTQW for other graphs are presented in \append{qwalk}.
The problem of traversing the glued trees graph is proven to be efficiently solvable for quantum computers while computationally hard for classical computers in the oracle setting~\cite{childs2003exponential}. 
Although our experiment simulates the desired quantum dynamics, this result itself does not imply achieving exponential quantum speedup on NISQ devices because the Hamiltonian embedding is not built using logarithmic quantum resources.

In the existing literature, CTQW on glued trees has been experimentally realized only in very limited settings.
\cite{shi2020quantum} used a photonic chip to simulate a one-dimensional quantum walk reduced from the CTQW on glued trees.
An earlier photonic implementation~\cite{tang2018experimental} realizes quantum walk on a hexagonal variant of the glued trees graph.
Implementations of CTQW have been explored more extensively for other graphs on various experimental platforms, such as NMR~\cite{du2003experimental}, neutral atoms~\cite{preiss2015strongly, young2022tweezer}, superconducting qubits~\cite{gong2021quantum}, and photonic processors~\cite{qu2022experimental}.
However, these implementations are hardware-specific and generally require qubits proportional to the order (i.e. number of vertices) of the graph, thus unable to accommodate other families of graphs.
The work in~\cite{qiang2021implementing, wang2022large} simulates CTQW in a photon interference experiment, but the desired evolution operators are calculated classically to configure linear optical circuits.
Discrete-time quantum walks have also been realized on a trapped-ion quantum processor~\cite{huerta2020quantum}, which uses a dense encoding (thus being space-efficient) but is essentially hard-coded for a specific choice of the initial state.

In \fig{fig_2}A, we illustrate a $14$-node glued trees graph $G = (V, E)$ containing two balanced binary trees. Here, $V$ is the set of vertices, and $E$ is the set of all edges.
The CTQW on this graph is described by the time-evolution operator $e^{-itA}$, where $t$ is the evolution time, and $A$ is a $14$-by-$14$ real symmetric matrix representing the adjacency matrix of the graph:
\begin{align}
    A_{jk}=\begin{cases}
        1, & j\neq k, (j, k)\in E, \\
        0, & j\neq k, (j, k)\not\in E.
    \end{cases}
\end{align}

To simulate the Hamiltonian $A$, we construct the following 14-qubit penalty-free one-hot embedding (details in \append{binary-glued}),
\begin{align}
\label{eqn:one-hot-L}
    H=\sum_{(j, k)\in E} \frac{1}{2}(X_jX_k+Y_jY_k),
\end{align}
where the embedding subspace $\S$ is spanned by the one-hot codewords $\{\ket{0..01},\ket{0..10},\dots,\ket{10..0}\}$. In other words, $\S$ is the single-excitation subspace in the $14$-qubit Hilbert space.
It is readily verified that $\S$ is an invariant subspace of $H$ and $H|_{\S} = A$, so the simulation of $A$ is embedded within the dynamics of $H$.

\fig{fig_2}B shows the total gate counts for the circuit implementation via standard binary encoding and Hamiltonian embedding (i.e., penalty-free one-hot). The total gate counts (represented by solid dots) are estimated such that the simulation error is suppressed to a fixed accuracy.
The results from extrapolation indicate that in terms of both exact and asymptotic measures, Hamiltonian embedding outperforms standard binary encoding.
A detailed discussion on the resource analysis for this experiment is available in \append{gluedtrees-resource}.

We implement the Hamiltonian embedding as shown in~\eqn{one-hot-L} on the IonQ Aria-1 processor to simulate the CTQW on the 14-node glued trees graph, where we employ the randomized first-order Trotter formula to decompose the simulation into circuits with IonQ's native gate set.
We initialize the quantum state with a single walker (starting from the entrance node in layer 1) at time $t=0$ and simulate the propagation of the quantum wave function through the layers over time. The real-machine and numerical simulation results are illustrated in~\fig{fig_2}C.
In both subplots, a clear pattern of population migration over time is observed.
The real-machine result depicts a strong similarity to the numerical simulation, while a slightly shifted population distribution is witnessed near the end time $T=2$, possibly caused by the accumulated device noise.
In \tab{glued_trees_gate_counts}, we list the empirical resource usage (qubit and gate counts) for traversing the glued trees graph using the one-hot code as well as an estimate of the resources required by the standard binary code.
While the standard binary code uses fewer qubits, the gate counts needed are an order of magnitude larger.
The savings in gate count for the penalty-free one-hot embedding enables the simulation of sparse Hamiltonians which would otherwise be infeasible on current devices.
\begin{table}[]
    \centering
    \begin{tabular}{|c|c|c|c|}
        \hline
        \textbf{Encoding/embedding} & \textbf{\# of qubits} & \textbf{\# of 1-qubit gates} & \textbf{\# of 2-qubit gates} \\
        \hline
        Standard binary & 4 & 6088 & 932 \\
        \hline
        Penalty-free one-hot & 14 & 1 & 160\\
        \hline
    \end{tabular}
    \caption{1- and 2-qubit gate counts for simulating CTQW on a 14-node glued trees graph. The accuracy is chosen to correspond to the Trotter number $r=4$ for the one-hot code as used in the real-machine experiment.}
    \label{tab:glued_trees_gate_counts}
\end{table}

\subsection{Spatial search on square grids}

\begin{figure}[ht!]
    \centering
    \includegraphics[width=16cm]{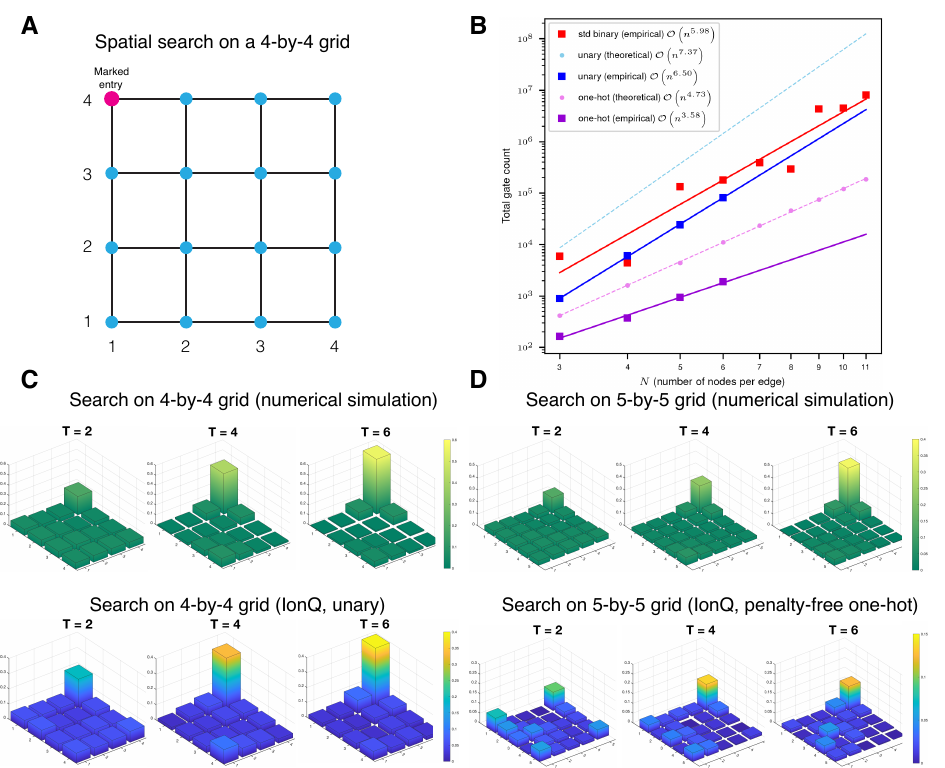}
    \caption{\small \textbf{Spatial search by quantum walk.}
    \textbf{A.} Illustration of a spatial search problem on a 4-by-4 grid. 
    \textbf{B.} Gate counts of implementing spatial search on $N$-by-$N$ grids. The penalty-free one-hot embedding has the best asymptotic scaling in both the worst-case estimate and empirical extrapolation.
    \textbf{C.} Quantum distribution at $T = 2, 4, 6$ shows the algorithm solves the 4-by-4 spatial search problem. Top: numerical simulation; Bottom: results obtained on IonQ Aria-1 with unary embedding.
    \textbf{D.} Quantum distribution at $T = 2, 4, 6$ shows the algorithm solves the 5-by-5 spatial search problem. Top: numerical simulation; Bottom: results obtained on IonQ Aria-1 with penalty-free one-hot embedding.}
    \label{fig:fig_3}
\end{figure}

Another prominent application for quantum computers is to find a target entry in a database, a task known as \emph{quantum search}. The famous Grover's algorithm~\cite{grover1997quantum} tackles the quantum search task with $O(\sqrt{M})$ queries to entries in a database of size $M$, while classical algorithms require at least $\Omega(M)$ queries, showcasing a quadratic quantum speedup.
In this section, we implement a quantum search algorithm proposed by Childs and Goldstone~\cite{childs2004spatial} for a structured database (i.e., \emph{spatial search}). The search spaces are two-dimensional square grids with $N^2$ entries.
By the composition rule of Hamiltonian embedding, we can implement the aforementioned spatial search algorithm using only $2N$ qubits.
In general, spatial search over a $d$-dimensional grid (with $N^d$ entries) can be implemented using $dN$ qubits by composing Hamiltonian embeddings (see \append{building}). 

There are only a few experimental demonstrations of spatial search documented in the literature. Some photonic-based implementations have been shown in~\cite{benedetti2021quantum, qu2022deterministic}, while these proposals only work for specific graphs and thereby lack programmability. A proof-of-principle demonstration of two-dimensional spatial search was performed using neutral atoms~\cite{young2022tweezer}, where $N^2$ atoms were used to represent all entries in the search space.
In this regard, our implementation of spatial search serves as the first demonstration of spatial search on a lattice that effectively exploits the matrix structure of the Hamiltonian.
The use of Hamiltonian embedding significantly reduces the required qubit and gate counts, thereby enabling a digital implementation of 2-dimensional spatial search that has not been demonstrated in any previous literature.

In 2004, Childs and Goldstone~\cite{childs2004spatial} proposed a quantum search algorithm via quantum walk. Given a search space represented by a graph $G=(V,E)$ and a marked entry $v \in V$, their quantum algorithm requires simulating the following Hamiltonian,
\begin{align}\label{eqn:spatial-search-ham}
    H=-\gamma L-H_{v},
\end{align}
where $L$ is the Laplacian of the graph $G$, $H_v=\ket{v}\bra{v}$ is a projector onto the marked entry $v$ (known as the \emph{oracle} Hamiltonian), and $\gamma>0$ is a parameter minimizing the spectral gap in $H$. 
When the search space is a two-dimensional grid graph (with $N^2$ entries $(j, k)$ for $j, k=1,\dots, N$) and the marked entry is $\vv = (v_x,v_y)$, the graph Laplacian reads
\begin{align}\label{eqn:spatial-search-laplacian}
    L = L_1 \otimes I_N + I_{N}\otimes L_1,
\end{align}
where $L_1$ is the graph Laplacian of a $N$-node one-dimensional chain graph (an example for $N=5$ is given in \eqn{5-node-chain}), and $I_N$ is the $N$-by-$N$ identity operator. Meanwhile, the oracle Hamiltonian is given by 
\begin{align}\label{eqn:spatial-search-oracle}
    H_{\vv} = \ket{v_x,v_y}\bra{v_x,v_y} =  \ket{v_x}\bra{v_x}\otimes \ket{v_y}\bra{v_y}.
\end{align}
In \fig{fig_3}A, we illustrate a spatial search problem on a $4$-by-$4$ grid with the marked entry at the upper left corner. By indexing the entries from the lower left corner, the oracle Hamiltonian for the marked entry is $H_{\vv} = \ket{1,4}\bra{1,4}$. 

To implement the quantum algorithm by Childs and Goldstone to solve the two-dimensional spatial search problem, we need to simulate the target Hamiltonian $H$ given in \eqn{spatial-search-ham}, which is of dimension $N^2$. Since $H$ is a sparse Hamiltonian, we could in principle construct a one-hot embedding using $N^2$ qubits. However, the decomposition structure in the graph Laplacian \eqn{spatial-search-laplacian} and the tensor product structure in the oracle Hamiltonian \eqn{spatial-search-oracle} allow us to build a more compact Hamiltonian embedding using Hamiltonian embeddings of $L_1$, $\ket{v_x}\bra{v_x}$, and $\ket{v_y}\bra{v_y}$.  
Notably, for both the unary and one-hot Hamiltonian embeddings, we only require $O(N)$ qubits and $O(N)$ $2$-local Pauli operators. 
More details of the construction for spatial search (including the case of periodic boundaries) may be found in \append{spatial_search}.

\fig{fig_3}B shows the total gate counts in the circuit implementation via standard binary encoding and Hamiltonian embeddings (i.e., unary, penalty-free one-hot). 
The extrapolation results suggest that the penalty-free one-hot embedding (both worst-case and empirical estimates) outperforms the standard binary encoding in both exact and asymptotic performance measures. However, the unary embedding appears to use more elementary gates than the standard binary encoding, potentially due to the large penalty coefficient. More discussions on the resource analysis for this experiment can be found in \append{spatial-resource}.

Experimental results for two-dimensional spatial search on $4$-by-$4$ and $5$-by-$5$ grids are shown in panels C and D in \fig{fig_3}. We choose the unary embedding for the $4$-by-$4$ experiment and the penalty-free one-hot embedding for the $5$-by-$5$ experiment. 
More details of the experiment setup, including the state preparation procedure, are provided in \append{spatial_search_experiments}.
The results show that our implementations via Hamiltonian embedding on the IonQ device are in good agreement with the expected quantum dynamics described in the algorithm.
\begin{table}[ht!]
    \centering
    \begin{tabular}{|c|c|c|c|}
        \hline
        \textbf{Encoding/embedding} & \textbf{\# of qubits} & \textbf{\# of 1-qubit gates} & \textbf{\# of 2-qubit gates} \\
        \hline
        Standard binary & 4 & 831 & 123 \\
        \hline
        Unary & 6 & 132 & 114 \\
        \hline
    \end{tabular}
    \caption{Resource counts for simulating spatial search on a $4\times 4$ lattice.}
    \label{tab:4_by_4_spatial_search_gate_counts}
\end{table}
\begin{table}[ht!]
    \centering
    \begin{tabular}{|c|c|c|c|}
        \hline
        \textbf{Encoding/embedding} & \textbf{\# of qubits} & \textbf{\# of 1-qubit gates} & \textbf{\# of 2-qubit gates} \\
        \hline
        Standard binary & 6 & 26100 & 4464 \\
        \hline
        Penalty-free one-hot & 10 & 22 & 181 \\
        \hline
    \end{tabular}
    \caption{Resource counts for simulating spatial search on a $5\times 5$ lattice.}
    \label{tab:5_by_5_spatial_search_gate_counts}
\end{table}
In \tab{4_by_4_spatial_search_gate_counts} and \tab{5_by_5_spatial_search_gate_counts}, we list the empirical resource usage for spatial search on $4\times 4$ and $5\times 5$ lattices, respectively.
Compared to our embedding schemes, the standard binary encoding requires orders of magnitude more gates for the same target accuracy.
Consequently, Hamiltonian embedding demonstrates a significant resource savings that enables the implementation of quantum search on current quantum computers.

\subsection{Simulating real-space quantum dynamics}\label{sec:sim-real-space-dynamics}

\begin{figure}[ht!]
    \centering
    \includegraphics[width=16cm]{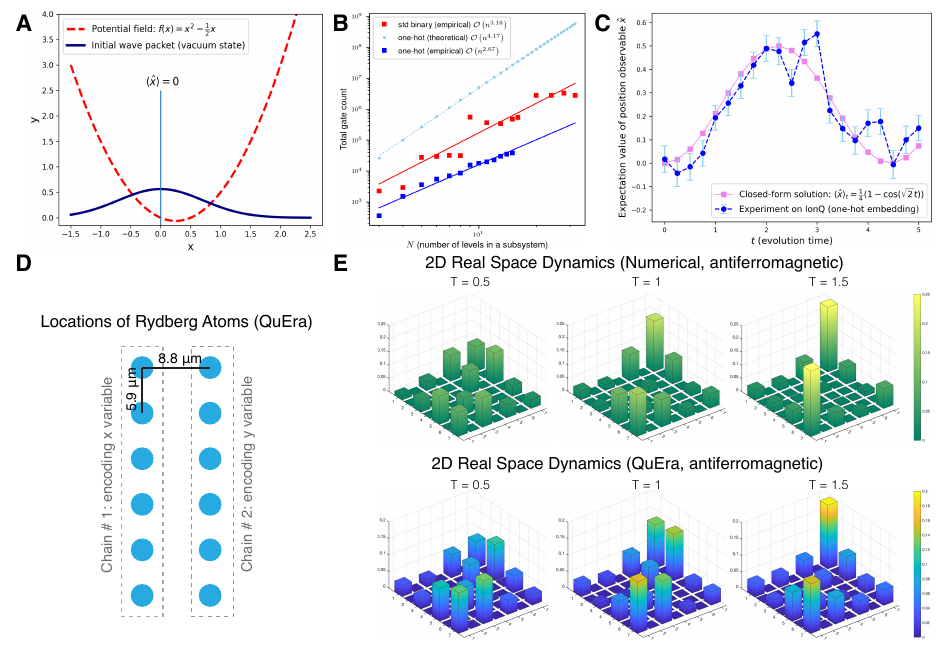}
    \caption{\small \textbf{Simulating real-space quantum dynamics.}
    \textbf{A.} The potential field and initial state in the IonQ experiment. 
    \textbf{B.} Resource count of simulating real-space quantum dynamics using an $N$-level truncated Hamiltonian. The empirical data suggests that one-hot embedding has a better asymptotic scaling than the standard binary encoding.
    \textbf{C.} Expectation value of the position observable $\hat{x}$ depicted as a function of evolution time $t$. Experiment data obtained from IonQ Aria-1.
    \textbf{D.} Locations of the 12 Rydberg atoms in the experiment of simulating quantum dynamics in a 2D space.
    \textbf{E.} Measuring the 2D real-space quantum dynamics at $T = 1, 1.5, 2$. Top: numerical simulation; Bottom: experiment results obtained on QuEra Aquila.}
    \label{fig:fig_4}
\end{figure}

The real-space quantum dynamics are governed by the time-dependent Schr\"odinger equation over the $d$-dimensional Euclidean space,
\begin{align}\label{eqn:real-space-schr}
    &i \frac{\partial}{\partial t}\Psi(t,x) = \left[-\frac{1}{2}\nabla^2 + f(x)\right] \Psi(t,x),\\
    &\text{subject to}\quad \Psi(0,x) = \Psi_0(x),
\end{align}
where $\Psi(t,x):[0, T]\times \mathbb{R}^d\to \mathbb{C}$ is the quantum wave function, $f\colon \mathbb{R}^{d} \to \mathbb{R}$ is a potential field, and the initial state is given by $\Psi_0(x)$. 

Although physically relevant systems are of dimension $3$, high-dimensional Schr\"odinger equations find ubiquitous applications for simulating multi-particle systems in condensed matter physics~\cite{babbush2018low} and quantum chemistry~\cite{kassal2008polynomial,babbush2019quantum}. 
Also, some quantum algorithms for numerical optimization require the simulation of Schr\"odinger equations in high-dimensional Euclidean spaces~\cite{zhang2021quantum,liu2023quantum,leng2023quantum}.
Several quantum algorithms for real-space quantum simulation have been proposed, e.g., \cite{wiesner1996simulations,zalka1998efficient,kivlichan2017bounding,an2021time,an2022time,childs2022quantum,jin2022quantum}, while most of them require fully fault-tolerant quantum computers that are currently out of reach.
Recently, Chang et al.~\cite{chang2022improving} utilize variants of Gray code to reformulate the continuous-space Schr\"odinger equation as spin Hamiltonians, leading to exponentially better space complexity over classical methods. However, the spin operators in \cite{chang2022improving} involve 3-body interaction terms and thus require further decomposition in the implementation. In our experiment, the different choices of embedding schemes allow us to represent the Schr\"odinger equation using 1- and 2-body spin operators, eliminating the Pauli compilation overhead.

We showcase the flexibility and versatility of the Hamiltonian embedding technique by implementing the real-space quantum simulation task on two different experimental platforms: a trapped-ion quantum processor (IonQ Aria-1) and a programmable Rydberg atom array (QuEra Aquila).

On the trapped-ion processor, we simulate a 1-dimensional Schr\"odinger equation with a quadratic potential field and a Gaussian initial state,
$$f(x) = x^2-\frac{1}{2}x,\quad \Psi_0(x) = \left(\frac{1}{2\pi}\right)^{1/4} e^{-\frac{x^2}{4}},$$
as illustrated in \fig{fig_4}A.
We use a truncated Fock space method with $N=5$ levels to discretize the Hamiltonian operator and obtain a finite-dimensional sparse Hamiltonian.
Next, we employ the penalty-free one-hot embedding to build the corresponding embedding Hamiltonian
\begin{equation}
    H_{\mathrm{real-space}} = \frac{1}{2} Q_{\hat{p}^2}^{\mathrm{one-hot}} - \frac{1}{2}Q_{\hat{x}^2}^{\mathrm{one-hot}} + \frac{1}{2} Q_{\hat{x}}^{\mathrm{one-hot}},
\end{equation}
where $Q_{\hat{p}^2}^{\mathrm{one-hot}}$, $Q_{\hat{x}^2}^{\mathrm{one-hot}}$, and $Q_{\hat{x}}^{\mathrm{one-hot}}$ are all 2-local Hamiltonians.
More details of the truncation method and embedding are discussed in \append{fock_space_truncation}.

In \fig{fig_4}B, we estimate and compare the gate counts required in the circuit implementation for the standard binary encoding and the penalty-free one-hot embedding.
The extrapolation based on empirical data shows that our Hamiltonian embedding implementation uses fewer elementary gates and has slightly better asymptotic scaling in terms of system size, while the theoretical (i.e., worst-case) asymptotic scaling of our embedding does not show an advantage.

We simulate the embedding Hamiltonian $H_{\mathrm{real-space}}$ on the IonQ Aria-1 processor and compute the expectation value of the position observable $\hat{x}$ as a function of time as shown \fig{fig_4}C.
Notably, computing $\hat{x}$ only requires measurements in the $x$-basis when using Hamiltonian embedding. 
The same measurements using the standard binary code would require performing measurements in several different bases corresponding to the Pauli decomposition of $\hat{x}$, which demonstrates another advantage of Hamiltonian embedding.
The experimental data matches the closed-form solution and we observe a full oscillation from $t=0$ to $t=5$. 
More details regarding the experiment setup and results are provided in \append{real-setup}.
Furthermore, we list in \tab{real_space_gate_counts} the empirical resource usage for simulating the real-space Schr\"odinger equation with the one-hot code as well as an estimate of the resources required by the standard binary code.
Although we consider only a 1-dimensional case in this example, the one-hot code demonstrates a significant advantage in gate count.
In particular, the one-hot code only requires a single 1-qubit gate (for state preparation), while the standard binary code requires over 1800 single-qubit gates, which would not be feasible on current quantum hardware.

\begin{table}[ht]
    \centering
    \begin{tabular}{|c|c|c|c|}
        \hline
        \textbf{Encoding/embedding} & \textbf{\# of qubits} & \textbf{\# of 1-qubit gates} & \textbf{\# of 2-qubit gates} \\
        \hline
        Standard binary & 3 & 1826 & 220 \\
        \hline
        Penalty-free one-hot & 5 & 1 & 154\\
        \hline
    \end{tabular}
    \caption{1- and 2-qubit gate counts for simulating the real-space Schr\"odinger equation with $N=5$ levels. The accuracy is chosen to correspond to the Trotter number $r=11$ for the one-hot code as used in the real-machine experiment.}
    \label{tab:real_space_gate_counts}
\end{table}

We also use programmable Rydberg atom arrays to simulate a 2-dimensional Schr\"odinger equation.
First, we apply the finite difference method to discretize the Hamiltonian operator, yielding a finite-dimensional Hamiltonian,
\begin{equation}\label{eqn:real-space-fdm}
    \hat{H} = -\frac{1}{2} (D \otimes I + I \otimes D) + U,
\end{equation}
where $D$ is an $N$-by-$N$ tridiagonal matrix representing the finite-difference discretization of the second-order differential operator $\frac{\partial^2}{\partial x^2}$, $I$ is the $N$-by-$N$ identity matrix and $U$ is a $N^2$-by-$N^2$ diagonal matrix corresponding to the potential field $f(x,y)$.
The native Hamiltonian of Rydberg atom arrays allows us to form a Hamiltonian embedding of \eqn{real-space-fdm} using the antiferromagnetic scheme, while other schemes (such as unary) are not possible because the Rydberg interaction coefficients must be positive (see \append{quera-antiferro} for details).
The desired Hamiltonian embedding can be realized by arranging the atoms into two chains as shown in \fig{fig_4}D, where each chain represents an individual continuous variable ($x$ for chain 1, $y$ for chain 2). Note that the lack of local detuning in QuEra Aquila poses a significant restriction on the shape of the potential field. In our experiment, the \emph{effective} potential field $U$ and the penalty Hamiltonian are both engineered by the pairwise Rydberg interactions between atoms. 
As shown in \fig{fig_4}E, the QuEra-implemented quantum simulation results have good agreement with the numerical simulation. The quantum distributions are obtained using 1000 shots/measurements per time step (for $T = 0.5, 1, 2$).
The details of the finite difference discretization, together with the construction and implementation of Hamiltonian embedding, are presented in \append{quera_experiments}.

\section{Open questions and discussion}
In this work, we propose the Hamiltonian embedding technique for hardware-efficient sparse Hamiltonian simulation. This approach has the potential to achieve exponential quantum speedups when a large sparse Hamiltonian can be efficiently decomposed using the rules as described in \thm{rules}. We then identify several instances from graph theory, combinatorial optimization, and differential equations where Hamiltonian embeddings can be constructed in poly-logarithmic time. It remains an interesting question whether Hamiltonian embeddings can be efficiently constructed for practical problems with less regular structures, potentially leading to practical quantum advantages even with limited quantum resources. 

Our analysis of the perturbative Hamiltonian embedding (see \thm{perturbative-thm}) indicates a non-negligible simulation error for a large evolution time, which can only be suppressed by a large penalty coefficient. This is because we only utilize the first-order Schrieffer–Wolff theory to craft perturbative Hamiltonian embedding. It is of interest to explore if higher-order Schrieffer–Wolff theories could lead to more efficient Hamiltonian embeddings.

In the resource analysis, the empirical scalings of Hamiltonian embedding are usually much better than the results suggested by the theoretical (worst-case) analysis. This may suggest that we need new analytical tools to better understand the resource efficiency of Hamiltonian embedding.

Our real-machine demonstrations of Hamiltonian embedding are limited to small-scale toy model problems, and the on-device experiment results do not fully match the ideal numerical simulation. This could be caused by the accumulated machine noise.
With the evolution of quantum hardware in the next few years, we are excited about the opportunity to further explore the practical quantum advantages that Hamiltonian embedding can offer in broad application domains, such as simulating non-quantum linear dynamics or even nonlinear dynamics at a large scale.

Last but not least, while the primary focus of this work is to expand the implementable quantum advantages for near-term (NISQ) devices, the Hamiltonian embedding technique also opens new possibilities for quantum input models and could be highly relevant even for fault-tolerant quantum computation. For example, the construction of Hamiltonian embeddings relies on the Pauli input model, which can be naturally leveraged by quantum algorithms such as quantum signal processing or quantum singular value transformation without the need to explicitly construct block encodings~\cite{chakraborty2025quantum}. 
Moreover, Hamiltonian embedding appears to offer a trade-off between system connectivity (local interactions) and space complexity, an active area of research in quantum compilation and circuit routing.

\section*{Acknowledgment}
We thank Alexey Gorshkov, Lei Fan, Lexing Ying, and Lin Lin for helpful discussions and insightful feedback.
This work was partially funded by the U.S. Department of Energy, Office of Science, Office of Advanced Scientific Computing Research, Accelerated Research in Quantum Computing under Award Number DE-SC0020273, the Air Force Office of Scientific Research under Grant No. FA95502110051, the U.S. National Science Foundation grant CCF-1816695 and CCF-1942837 (CAREER), and a Sloan research fellowship. JL is partially supported by DOE QSA Grant No. FP00010905, a Simons Quantum Postdoctoral Fellowship, and a Simons Investigator award through Grant No. 825053. 

\newpage
\bibliographystyle{quantum}
\bibliography{reference}

%%%%%%%%%%%%%%%%%%%%%%%%%%%%%%%%%%%%%%%%%%%%%%%%%%%%%%%%%%%%%%%
\newpage
\appendix
\include{appendix}

\end{document}

%% file: appendix.tex
\noindent{\huge \textbf{Appendices}}

\section{Quantum simulation using Hamiltonian embedding} 
\paragraph{Notation.}
Let $\H$ be a Hilbert space, and $A\colon \H \to \H$ be a linear operator over $\H$. Let $\S$ be a subspace of $\H$, and $\S^\perp$ be the orthogonal complement of $\S$ such that $\H = \S \oplus \S^\perp$. We denote $P_\S$ (or $P_{\S^\perp}$) as the projection onto $\S$ (or $\S^\perp$). We write $A|_{\S}\coloneqq P_\S A P_\S$ as the restriction of $A$ in the subspace $\S$. $\Herm(\C^n)$ represents the space of all $n$-by-$n$ complex Hermitian matrices.

\subsection{Proof of \texorpdfstring{\thm{main}}{Theorem~\ref*{thm:main}}}\label{append:proof-main}

\MainThm*

\begin{proof}
We define $\mathscr{E}(t) \coloneqq e^{-iHt} - U^\dagger e^{-iHt} U$, and it follows that
\begin{align}
   \dot{\mathscr{E}}(t) &= (-iH)e^{-iHt} - U^\dagger (-i H) e^{-i H t} U = (-iH)\mathscr{E}(t) - i\left[H - \tilde{H}\right]U^\dagger e^{-i H t} U,
\end{align}
where $\tilde{H} = U^\dagger H U$. By the variation-of-parameter formula, we have
\begin{align}
   \mathscr{E}(t) = -i\int^t_0 e^{-iH(t-\tau)} \left[H - \tilde{H}\right]e^{-i \tilde{H} \tau} ~\d \tau,
\end{align}
and it follows that
\begin{align*}
   \|P_{\S}\mathscr{E}(t)P_{\S}\| \le \|(H - \tilde{H})P_{\S}\|t = \|[H, U]P_{\S}\|t.
\end{align*}

We denote $H_0 = P_\S \tilde{H}P_\S$. By the definition of Hamiltonian embedding, $\Tilde{H}$ is block-diagonal and $\|H_0 - A\| \le \epsilon$. Then, we have
\begin{align*}
   \|P_\S U^\dagger e^{-iHt} U P_\S - e^{-iAt}\| = \|e^{-i H_0 t} - e^{-i A t}\| \le \|H_0 - A\| t \le \epsilon t.
\end{align*}
It follows from the triangle inequality that
\begin{align}
    \|\left(e^{-iHt}\right)\Big|_\S - e^{-iAt} \| \le \|P_\S \mathscr{E}(t) P_\S\| + \|P_\S U^\dagger e^{-iHt} U P_\S-e^{-iAt}\| \le \left(\|[H, U]P_{\S}\| + \epsilon\right)t.
\end{align}
Note that $\|[H, U]\| = \|[H, I] - [H, (I-U)]\| \le 2 \eta \|H\|$, which implies the error bound \eqn{error_bound}.
\end{proof}

\subsection{Rules for building Hamiltonian embeddings}\label{append:building}

In this section, we introduce several basic rules for building more sophisticated Hamiltonian embeddings from existing ones.

\begin{lemma}[Addition]\label{lem:addition}
    Let $A_1, A_2 \in \Herm(\C^n)$. For $j = 1,2$, we suppose that $H_j$ is a $(q,\eta, \epsilon_j)$-embedding of $A_j$ with unitary transformation $U$ and embedding subspace $\S$. Then, $H = H_1 + H_2$ is a $(q, \eta, \epsilon_1+\epsilon_2)$-embedding of $A = A_1 + A_2$ with the same unitary transformation $U$ and embedding subspace $\S$.  
\end{lemma}
\begin{proof}
    Since the unitary $U$ block-diagonalizes both $H_1$ and $H_2$, it also block-diagonalizes the $H = H_1+H_2$. Moreover, by the triangle inequality, $\|(U^\dagger H U)_\S - (A_1 + A_2)\| \le \|(U^\dagger H_1 U)_\S - A_1\|+\|(U^\dagger H_2 U)_\S - A_2\| \le \epsilon_1+\epsilon_2$. 
\end{proof}

\begin{lemma}[Multiplication]\label{lem:multiplication}
    Let $H$ be a $(q,\eta,\epsilon)$-embedding of $A$ with unitary transformation $U$ and embedding subspace $\S$. Then, for any real number $\alpha$, $\alpha H$ is a $(q,\eta,|\alpha| \epsilon)$-embedding of $\alpha A$ with the same unitary transformation $U$ and embedding subspace $\S$.
\end{lemma}
\begin{proof}
    Due to the linearity of matrix multiplication, we have that $$\|(U^\dagger \alpha H U)_\S - \alpha A\| = |\alpha| \|(U^\dagger H U)_\S - A\| \le |\alpha| \epsilon.$$
\end{proof}

\begin{lemma}[Composition]\label{lem:composition}
    For $j = 1,2$, let $H_j$ be a $(q_j, \eta_j, \epsilon_j)$-embedding of $A_j$ with unitary transformation $U_j$ and embedding subspace $\S_j$, respectively. Then, $H = H_1\otimes I + I \otimes H_2$ is a $(q_1+q_2,\eta_1+\eta_2,\epsilon_1+\epsilon_2)$-embedding of $A = A_1\otimes I + I\otimes A_2$ with unitary transformation $U = U_1\otimes U_2$ and embedding subspace $\S = \S_1\otimes \S_2$.
\end{lemma}
\begin{proof}
    In the full Hilbert space $\H' = \C^{2^{q_1+q_2}}$, we have the embedding subspace $\S = \S_1\otimes \S_2$, so its orthogonal complement is $$\S^\perp = (\S^\perp_1\otimes \S_2)\oplus (\S_1\otimes \S^\perp_2)\oplus (\S^\perp_1\otimes \S^\perp_2).$$
    First, we check the rotated Hamiltonian $\tilde{H}\coloneqq U^\dagger H U$ is block-diagonal in $\S$ and $\S^\perp$. We observe that $P_{\S^\perp} = P_{\S^\perp_1\otimes \S_2} + P_{\S_1\otimes \S^\perp_2} + P_{\S^\perp_1\otimes \S^\perp_2}$. It is readily verified that
    $$P_{\S^\perp_1\otimes \S_2} \left(U^\dagger HU\right)P_{\S_1\otimes \S_2} = \left(P_{\S^\perp_1}\tilde{H}_1P_{\S_1}\right)\otimes I + \left(P_{\S^\perp_1}P_{\S_1}\right)\otimes \left(P_{\S_2}\tilde{H}_2P_{\S_2}\right)=0.$$
    Similarly, we can show that
    $$P_{\S_1\otimes \S^\perp_2}\left(U^\dagger HU\right)P_{\S_1\otimes \S_2} = 0,\quad P_{\S^\perp_1\otimes \S^\perp_2}\left(U^\dagger HU\right)P_{\S_1\otimes \S_2} = 0,$$
    which implies $P_{\S^\perp}\tilde{H} P_{\S} = 0$, i.e., $\tilde{H}$ is block-diagonal. Also, for the unitary $U = U_1\otimes U_2$, we have
    \begin{align*}
        \|I - U_1\otimes U_2\| \le \|(I - U_1)\otimes I\| + \|U_1 \otimes (I - U_2)\| \le \eta_1 + \eta_2.
    \end{align*}
    Finally, we check the approximation error,
    \begin{align*}
        &\|P_\S U^\dagger_1\otimes U^\dagger_2 (H_1\otimes I + I \otimes H_2) U_1\otimes U_2 P_\S - (A_1\otimes I + I \otimes A_2)\| \\
        &= \Big\|\left[(U^\dagger_1 H_1 U_1)|_{\S_1} - A_1\right]\otimes I_{\S_2} +I_{\S_1}\otimes\left[(U^\dagger_2 H_2 U_2)|_{\S_2} - A_2\right] \Big\| \le \epsilon_1 + \epsilon_2.
    \end{align*}
\end{proof}

\begin{lemma}[Tensor product]
    For $j = 1,2$, let $H_j$ be a $(q_j, \eta_j, \epsilon_j)$-embedding of $A_j$ with unitary transformation $U_j$ and embedding subspace $\S_j$, respectively. Then, $H = H_1\otimes H_2$ is a $(q_1+q_2, \eta_1+\eta_2,\|A_1\|\epsilon_2 + \|A_2\|\epsilon_1 + \epsilon_1\epsilon_2)$-embedding of $H = A_1\otimes A_2$ with unitary transformation $U = U_1\otimes U_2$ and embedding subspace $\S = \S_1\otimes \S_2$.
\end{lemma}
\begin{proof}
    Similar to Rule 3, we can check that $P_{\S^\perp}(U^\dagger H_1\otimes H_2 U)P_\S = 0$ and $\|I - U_1\otimes U_2\| \le \eta_1 + \eta_2$. For $j = 1,2$, we have that $\|(U^\dagger_j H_j U_j)|_{\S_j}\| \le \|A_j\|+\epsilon_j$. Using this fact, we can estimate the approximation error:
    \begin{align*}
        &\|P_\S U^\dagger_1\otimes U^\dagger_2 (H_1\otimes H_2) U_1\otimes U_2 P_\S - (A_1\otimes A_2)\| = \|(U^\dagger_1 H_1 U_1)|_{\S_1} \otimes (U^\dagger_2 H_2 U_2)|_{\S_2} - A_1\otimes A_2\|\\
        &\le \|[A_1 - (U^\dagger_1 H_1 U_1)|_{\S_1}]\otimes A_2\| + \|(U^\dagger_1 H_1 U_1)|_{\S_1} \otimes [A_2 - (U^\dagger_2 H_2 U_2)|_{\S_2}]\|\\
        &\le \|A_2\|\epsilon_1 + \|A_1\|\epsilon_2 + \epsilon_1\epsilon_2.
    \end{align*}
\end{proof}

\begin{remark}\label{rem:graph_vis}
    The nonzero entries of a Hamiltonian give rise to the adjacency matrix of a graph (ignoring diagonal entries for simplicity).
    As a result, the composition and tensor product rules can be visualized as two different notions of graph products.
    For graphs $G_1=(V_1,E_1)$ and $G_2=(V_2,E_2)$, the \textit{Cartesian product} of $G_1$ and $G_2$ is the graph $G_1 \square G_2$ with vertex set $V_1 \times V_2$ such that vertices $(u_1,v_1)$ and $(u_2,v_2)$ are adjacent iff either $u_1=u_2$ and $v_1$ is adjacent to $v_2$, or $v_1=v_2$ and $u_1$ is adjacent to $u_2$.
    The \textit{tensor product} of $G_1$ and $G_2$ is the graph $G_1 \times G_2$ with vertex set $V_1 \times V_2$ such that vertices $(u_1,v_2)$ and $(u_2,v_2)$ are adjacent iff both $u_1$ is adjacent to $u_2$ and $v_1$ is adjacent to $v_2$.
    From the definitions, it follows that the adjacency matrix of $G_1 \square G_2$ is $A_1 \otimes I + I \otimes A_2$, while the adjacency matrix of $G_1 \times G_2$ is $A_1 \otimes A_2$.
    Therefore, the composition and tensor product rules can be visualized as taking the graph Cartesian and tensor products, respectively.
\end{remark}

\subsection{Perturbative Hamiltonian embedding}\label{append:perturbation}

\subsubsection{Schrieffer-Wolff theory}\label{append:perturb_theory}
We consider the $q$-qubit operator given in \eqn{perturbative-embedding}, i.e., $H = g\Hpen + Q$, where $Q$ and $\Hpen$ can be expressed in block-matrix forms (w.r.t. subspaces $\S$ and $\S^\perp$):
$$Q = \begin{bmatrix}A & R^\dagger\\R & B\end{bmatrix},\quad \Hpen = \begin{bmatrix}0 & 0\\0 & C\end{bmatrix}.$$
Therefore, the operator $H = \begin{bmatrix}A & R^\dagger\\R & G\end{bmatrix}$, where $G = B+gC$.

We define $\Delta_0 \coloneqq \lambda_{\min}(B) - \lambda_{\max}(A)$ as the spectral distance between the two diagonal blocks in $Q$, and $\Delta \coloneqq \lambda_{\min}(G) - \lambda_{\max}(A)$ being the spectral distance between the two diagonal blocks in $H$. Note that the spectral gaps $\Delta_0$ and $\Delta$ are not necessarily positive. By Weyl's inequality, we have
$$\Delta \ge g \lambda_1 + \Delta_0,$$
where $\lambda_1 > 0$ is the minimum eigenvalue of $C$ (recall that $\S$ is the ground-energy subspace of $\Hpen$).

By choosing sufficiently large $g > 0$, we can make $\Delta > 0$ and $\Delta \sim g$. For large $\Delta$ such that $\|R\|/\Delta \ll 1$, the off-diagonal part in $H$ can be regarded as a perturbation and we can invoke the Schrieffer-Wolff transformation~\cite{davis1970rotation,bravyi2011schrieffer}.

\begin{theorem}[Schrieffer-Wolff transformation]\label{thm:dir-rotation}
    Define $\kappa \coloneqq \|R\|/\Delta$. We assume that $\Delta > 0$ and $\kappa < 1/2$. Then, there exists a unitary transformation $U$ such that $\|I-U\|\le 2\sqrt{2}\kappa$ and $U^\dagger H U$ is block-diagonal. Moreover, if $\kappa < 1/16$, there exists an absolute constant $c> 0$ such that 
    \begin{align}\label{eqn:error_embedding}
        \left\|\left(U^\dagger H U\right)_{\S} - A\right\| \le c \|R\|\kappa.
    \end{align}
    In other words, $H$ is a $(q, 2\sqrt{2}\kappa, c\|R\|\kappa)$-embedding of $A$ if $\kappa < 16$.
\end{theorem}
\begin{proof}
    We define $H_0 = \begin{bmatrix}A & 0\\0 & G\end{bmatrix}$ and $V = \begin{bmatrix}0 & R\\R^\dagger & 0\end{bmatrix}$, so $H = H_0 + V$. Note that $\|V\| = \|R\|$. When $\kappa = \|V\|/\Delta< 1/2$, there is a unique Schrieffer-Wolff transformation $U$ such that $U(H_0 + V)U^\dagger$ is block-diagonal (see~\cite{bravyi2011schrieffer} for detailed discussions). By~\cite[Section 1]{davis1970rotation}, there exists an \emph{angle} operator $\Theta$ such that $\|I-U\|= 2\|\sin \frac{1}{2}\Theta\|$. Let $\mathcal{T}= U \S$ be the rotated subspace, we also have that $\left\|P_{\S} - P_{\mathcal{T}}\right\| = \|\sin\Theta\|$. On the other hand, by \cite[Lemma 3.1]{bravyi2011schrieffer}, we can estimate the distance between $\S$ and $\mathcal{T}$: 
    $$\left\|P_{\S} - P_{\mathcal{T}}\right\| \le 2\kappa.$$

    Since we have $\kappa < 1/2$, the angle operator is bounded by $\pi/2$~\cite{davis1970rotation}. Note that for any angle $|\theta |< \pi/2$, we have
    $$\left|\sin \frac{1}{2}\theta\right| = \sqrt{\frac{1}{2}(1-\cos \theta)} = \sqrt{\frac{1}{2}(1-\sqrt{1-\sin^2\theta})} \le \frac{\sqrt{2}}{2}|\sin\theta|,$$
    which implies that $\|\sin \frac{1}{2}\Theta \|\le \frac{\sqrt{2}}{2}\|\sin \Theta\|$. It turns out that
    $$\|I-U\|\le \sqrt{2}\|P_{\S}-P_{\mathcal{T}}\| \le 2\sqrt{2}\kappa.$$

    We define the effective low-energy Hamiltonian $H_{\mathrm{eff}}\coloneqq P_{\S}U^\dagger H UP_{\S}$. The effective Hamiltonian $H_{\mathrm{eff}}$ allows a Taylor series expansion~\cite{bravyi2011schrieffer,bravyi2017complexity}:
    \begin{align}\label{eqn:H_series}
        H_{\mathrm{eff}} = \sum^\infty_{j=1}H_{\mathrm{eff},j}.
    \end{align}
    We define $H_{\mathrm{eff}}(k) = \sum^k_{j=1}H_{\mathrm{eff},j}$. When $\kappa < 1/16$, the series \eqn{H_series} converge absolutely and there exists an absolute constant $c> 0$ such that
    $$\left\|H_{\mathrm{eff}}- H_{\mathrm{eff}}(k)\right\| \le c\frac{\|V\|^{k+1}}{\Delta^k} = c\|R\|\kappa^k.$$
    Note that $H_{\mathrm{eff}}(1) = A$. Therefore, we prove that $\left\|H_{\mathrm{eff}}- A\right\| \le c\|R\|\kappa$.
\end{proof}

By \thm{dir-rotation}, a bigger $g$ implies a more accurate Hamiltonian embedding $H$. The idea is that a large parameter $g$ suppresses the leakage of the quantum evolution from the embedding subspace $\S$ to its orthogonal complement $\S^\perp$. Based on this intuition, we refer to $g$ as the \emph{penalty coefficient} and the operator $\Hpen$ as the \emph{penalty Hamiltonian}.

\vspace{4mm}
\begin{remark}
    In a perturbative Hamiltonian embedding with penalty coefficient $g > 0$, the error characterization parameters are $\eta \le 2\sqrt{2}\kappa \sim \frac{1}{g}$, $\epsilon \le c\|R\|\kappa \sim \frac{\|R\|}{g}$.
    By choosing a large enough $g$, the error parameters $\eta$ and $\epsilon$ can be made arbitrarily small. On the other hand, when $\Hpen$ is fast-forwardable (which is always the case for embedding schemes discussed in this paper, see \append{sparse_matrix}), there exist quantum algorithms (e.g., continuous qDRIFT, see \append{sim_perturbative}) that can efficiently simulate the Hamiltonian embedding even for large $g$. Therefore, for simplicity, we will not explicitly specify $\eta$ and $\epsilon$ when discussing a perturbative Hamiltonian embedding in the rest of this paper.
\end{remark}

\vspace{4mm}
All the basic rules of constructing Hamiltonian embeddings can be reformulated for perturbative Hamiltonian embedding. The proof is similar to \append{building}.

\begin{proposition}
    \phantom{}
    \begin{enumerate}
        \item (Addition) For $j = 1, 2$, suppose $H_j = g\Hpen+Q_j$ is an embedding of $A_j$ with the same embedding subspace $\S$. Then, $H = g\Hpen + Q_1 + Q_2$ is an embedding of $A = A_1 + A_2$ with embedding subspace $\S$.
        \item (Multiplication) Suppose $H = g\Hpen+Q$ is an embedding of $A$ with embedding subspace $\S$. Then, for any real number $\alpha$, $H' = g\Hpen + \alpha Q$ is an embedding of $A' = \alpha A$.
        \item (Composition) For $j = 1, 2$, suppose that the $q_j$-qubit operator $H_j = g\Hpen_j + Q_j$ is an embedding of $A_j$ with the embedding subspace $\S_j$ being the ground-energy subspace of $\Hpen_j$. Then, $H = g \Hpen + Q$ with
        \begin{align}
            \Hpen = \Hpen_1\otimes I + I \otimes \Hpen_2,~Q = Q_1\otimes I + I \otimes Q_2
        \end{align}
        is an embedding of $A= A_1\otimes I + I \otimes A_2$ with embedding subspace $\S = \S_1\otimes \S_2$.
        \item (Tensor product) For $j = 1,2$, suppose that $H_j = g \Hpen_j + Q_j$ is an embedding of $A_j$ with the embedding subspace $\S_j$ being the ground-energy subspace of $\Hpen_j$. Then, $H = g \Hpen + Q_1\otimes Q_2$ with
    \begin{align}
        \Hpen = \Hpen_1\otimes I + I \otimes \Hpen_2
    \end{align}
    is an embedding of $H = A_1\otimes A_2$ with embedding subspace $\S = \S_1\otimes \S_2$.
    \end{enumerate}
\end{proposition}

\subsubsection{Improved simulation error analysis}\label{append:proof_perturb}

\begin{theorem}[Quantum simulation with perturbative embedding]\label{thm:perturb}
Let all the notations be the same as above.
We assume $\Delta > 0$ and $\kappa < 1/2$. Then, for any fixed $t > 0$, we have
    \begin{align}
        \|P_{\S}e^{-iH t} P_{\S} - e^{-iAt}\| \le 4\sqrt{2}\kappa\|R\|t.
    \end{align}
\end{theorem}
\begin{proof}
Let $H_0 = P_\S H P_\S = \diag(A,0)$, and we define:
\begin{align}
    \mathscr{A}(t) = P_\S e^{-iHt}P_\S,\quad \mathscr{B}(t) = P_\S e^{-iH_0 t}P_\S.
\end{align}
We also define the $\mathscr{E}(t) \coloneqq \mathscr{A}(t) - \mathscr{B}(t)$. Note that $\mathscr{E}(0) = 0$. We find that
\begin{align*}
    \frac{\d}{\d t}\mathscr{A}(t) &= P_\S \left(-i H\right)\left(P_\S + P_{\S^\perp}\right) e^{-iHt} P_\S= -i H_0\mathscr{A}(t) -i R e^{-iHt} P_\S,\\
    \frac{\d}{\d t}\mathscr{B}(t) &= P_\S \left(-iH_0\right)e^{-iH_0 t}P_\S=-i H_0 \mathscr{B}(t).
\end{align*}
Therefore, we have $\dot{\mathscr{E}}(t) =\dot{\mathscr{A}}(t)-\dot{\mathscr{B}}(t) = -i H_0 \mathscr{E}(t) - i Re^{-iHt} P_\S$. By the variation-of-parameter formula \cite[Theorem~4.9]{knapp2007basic}, we represent the function $\mathscr{E}(t)$ as the time integral:
\begin{align}
    \mathscr{E}(t) = -i\int^t_0 \exp(-iH_0 (t-\tau))R \exp(-iH\tau) P_{\S} ~\d\tau,
\end{align}
so we can estimate the error (recall that $R = P_\S H P_{\S^\perp}$):
\begin{align*}
    \|\mathscr{E}(t)\| &\le \int^t_0 \|P_\S H P_{\S^\perp} \exp(-iH\tau) P_{\S}\| ~\d \tau \le \int^t_0 \|RP_{\S^\perp} \exp(-iH\tau) P_{\S}\|~\d \tau \\
    &\le \int^t_0 \|R\|\left(4\sqrt{2}\kappa\right)~\d \tau = 4\sqrt{2}\kappa\|R\|t.
\end{align*}
In the second step, we use $P_{\S^\perp} = P^2_{\S^\perp}$. In the second to last step, we invoke \lem{cross}.
\end{proof}

\begin{lemma}\label{lem:cross}
Assume that $\kappa < 1/2$ and let $\tau \ge 0$ be an arbitrary real number. Then, we have
\begin{align}
    \|P_{\S^\perp} \exp(-iH\tau) P_{\S}\| \le 4\sqrt{2}\kappa.
\end{align}
\end{lemma}
\begin{proof}
By \thm{dir-rotation}, there exists a Schrieffer-Wolff transformation $U$ such that
$$U^\dagger H U = \begin{bmatrix}
    H_{\mathrm{eff}} & 0\\0 & *
\end{bmatrix},$$
where $H_{\mathrm{eff}} = P_{\S}U^\dagger H U P_{\S}$.
We remark that $P_{\S^\perp} U^\dagger \exp(-iH\tau) U P_\S=0$. To see this, note that $U$ block-diagonalizes $H$, so does $\exp(-iUHU^\dagger\tau)$:
$$P_{\S^\perp}U \exp(-iH\tau) U^\dagger P_{\S} = P_{\S^\perp}\exp(-iUHU^\dagger\tau)P_{\S} = 0.$$
Now, we bound the cross term $P_{\S^\perp} \exp(-iH\tau) P_{\S}$ by,
\begin{align*}
    \|P_{\S^\perp} \exp(-iH\tau) P_{\S}\| 
    &= \|P_{\S^{\perp}} \exp(-iH\tau) P_{\S} - P_{\S^{\perp}}U^\dagger \exp(-iH\tau) U P_{\S}\|\\
    &\leq \|\exp(-iH\tau) - U^\dagger \exp(-iH\tau)U\|
    = \|U\exp(-iH\tau) - \exp(-iH\tau)U\|.
\end{align*}
Then, by \thm{dir-rotation}, we conclude that 
\begin{align*}
    \|P_{\S^\perp} \exp(-iH\tau) P_{\S}\| 
    \le \|(U-I)\exp(-iH\tau) - \exp(-iH\tau)(U-I)\|
    \leq 2 \| U-I \|
    \leq 4\sqrt{2} \kappa.
\end{align*}
\end{proof}

\subsubsection{Quantum simulation of perturbative embedding}\label{append:sim_perturbative}
The perturbative Hamiltonian embedding $\Hebd = g\Hpen + Q$ often comes with a large penalty coefficient $g$. This could lead to significant simulation error if one uses a standard quantum algorithm (e.g., a product formula) to simulate $\Hebd$. In this section, we discuss two types of quantum simulation protocols (product formulas and \texttt{qDRIFT}) for perturbative Hamiltonian embedding. It turns out that \texttt{qDRIFT} is friendly to near-term devices with minimal overheads in terms of $g$.

First, we consider the standard first-order Trotter formula:
\begin{align}
    e^{-i\Hebd t} \approx U_1(t) \coloneqq \prod^r_{k=1} e^{-iQt/r}e^{-ig\Hpen t/r}.
\end{align}
By \cite{childs2021theory}, the Trotter error scales like
\begin{align}
    \|e^{-i\Hebd t} - U_1(t)\| \le O\left(\frac{g[\Hpen, Q]t^2}{r}\right).
\end{align}
To simulate the embedding Hamiltonian up to an additive error $\epsilon$, we need to choose the Trotter number 
\begin{align}
    r = O\left(\frac{g[\Hpen, Q]t^2}{\epsilon}\right).
\end{align}

Similarly, the second-order Trotter-Suzuki formula requires the Trotter number
\begin{align}
    r = O\left(\frac{\Gamma t^{3/2}}{\epsilon^{1/2}}\right),
\end{align}
where $\Gamma = \left(\frac{g}{12}\|[Q,[Q,\Hpen]]\| + \frac{g^2}{24}\|[\Hpen,[\Hpen,Q]]\|\right)^{1/2}$.

We can improve the first-order Trotter formula by randomization~\cite{childs2019faster}, which turns out to be effectively a second-order product formula:
\begin{align}
    \|e^{-i\Hebd t} - U_{1,\mathrm{rand}}(t)\| \le O\left(\frac{\Lambda^3 t^3}{r^2}\right),
\end{align}
where $U_{1,\mathrm{rand}}(t)$ is the randomized first-order Trotter formula, and $\Lambda = \max\{g\|\Hpen\|, \|Q\|\}$. It follows that we need the Trotter number
\begin{align}
    r = O\left(\frac{\Lambda^{3/2} t^{3/2}}{\epsilon^{1/2}}\right)
\end{align}
in order to achieve an $\epsilon$-accurate simulation of $H'$. A second-order (deterministic) Trotter formula will require that $r = O(gt^{3/2}\epsilon^{-1/2})$, but potentially have a better pre-factor in terms of the commutators of $\Hpen$ and $Q$~\cite{childs2021theory}.

\vspace{4mm}
A major drawback of the aforementioned quantum simulation protocols we discussed above is that their simulation error has explicit dependence on the penalty coefficient $g$, which could be very huge in practice. To address this issue, we consider quantum simulation in the interaction picture. For a Hamiltonian $H = A + B$, we transform it to the interaction picture:
\begin{align}
    H_I(t) = e^{iAt}B e^{-iAt},
\end{align}
and consider the Schr\"odinger equation 
\begin{align}
    i\frac{\d }{\d t}\ket{\psi_I(t)} = H_I(t)\ket{\psi_I(t)},
\end{align}
subject to $\ket{\psi_I(0)} = \ket{\psi_0}$. One can easily verify that the state vector $\ket{\psi(t)} = e^{-iHt}\ket{\psi_0}$ can be recovered by
\begin{align}
    \ket{\psi(t)} = e^{-iAt}\ket{\psi_I(t)}. 
\end{align}
Based on this observation, digital quantum simulation protocols for time-dependent Hamiltonian that has $L^1$-norm scaling can be used to simulate $H_I(t)$ to achieve an error scaling only depends on $\|B\|$, given that $A$ can be fast-forwardable. 

In our construction of Hamiltonian embedding, $\Hpen$ is often fast-forwardable, so we can leverage the interaction-picture quantum simulation protocols to simulate 
\begin{align}
    \Hebd_I(t) = e^{ig\Hpen t}Qe^{-ig\Hpen t}
\end{align}
and achieve $g$-independent error scaling. There are a handful of quantum algorithms that can achieve $l^1$-norm scaling for interaction picture Hamiltonian, e.g., truncated Dyson series~\cite{low2018hamiltonian, berry2020time}, continuous qDRIFT~\cite{berry2020time}, qHOP~\cite{an2022time}, etc. Unfortunately, the truncated Dyson series method and the qHOP method require complicated control flows and advanced input models (e.g., block-encodings) that are infeasible for near-term quantum computers. On the other hand, continuous qDRIFT relies on a simple product formula and a classical stochastic sampling protocol, which is possible to implement on near-term devices. The details of \texttt{qDRIFT} are presented in \algo{qDRIFT}.

\begin{algorithm}[htbp]
    \SetKwInOut{Input}{input}\SetKwInOut{Output}{output}
    \caption{Quantum simulation by \texttt{qDRIFT}.}
    \label{algo:qDRIFT}   
    
    \Input{An initial state $\ket{\psi_0}$, number of samples $M$, Trotter number $K$, total evolution time $T$.}
    \Output{An estimate of the measurement result $\Tilde{E}$.}
    \BlankLine
    \emph{Let $\Delta t = T/K$}\;
    \For{$j\leftarrow 1$ \KwTo $M$}{
    \For{$k\leftarrow 1$ \KwTo $K$}{
    Sample a uniformly random number $\xi_k \in [(k-1)\Delta t, k\Delta t]$\;
    Propagate the quantum state
    $$\ket{\psi_k} = e^{-i\Hebd_I(\xi_k)\Delta t} \ket{\psi_{k-1}} = e^{ig\Hpen \xi_k}e^{-iQ\Delta t}e^{-ig\Hpen \xi_k}\ket{\psi_{k-1}}.$$
    }
    Measure the quantum state $\ket{\psi_K}$ with the observable $M_I = e^{ig\Hpen T}O e^{-ig\Hpen T}$ and compute $m_j = \bra{\psi_K}M_I\ket{\psi_K}$\;
    }
    Compute the sample mean $\Tilde{E} = \frac{1}{M}\sum^M_{j=1}m_j$.
\end{algorithm}

We can prove that when $M, K$ are large enough, 
\begin{align}
    \bra{\psi_0}e^{i\Hebd T} O e^{-i\Hebd T}\ket{\psi_0} \approx \Tilde{E}.
\end{align}
More precisely, by \cite{berry2020time}, the Trotter number required by continuous qDRIFT to simulate the interaction-picture Hamiltonian $H'_I(t)$ is 
\begin{align}\label{eqn:err_qdrfit}
    r = O\left(\frac{\|Q\|^2 t^2}{\epsilon}\right).
\end{align}
Notably, the error bound in \eqn{err_qdrfit} is independent of the penalty coefficient $g$, which allows us to simulate the embedding Hamiltonian $\Hebd$ with large $g$.

\section{Hamiltonian embedding of sparse matrices}\label{append:sparse_matrix}
In this section, we construct several Hamiltonian embeddings whose locality has explicit dependence on the sparsity structure of the target matrix. Such embeddings are advantageous for special sparse matrices, e.g., tridiagonal matrices.

\subsection{Band matrix}\label{append:band}
A band matrix is a sparse matrix whose non-zero entries are confined to a diagonal band. We say a matrix $A$ has bandwidth $d$ if $A_{i,j} = 0$ for any $i,j=1,\dots, n$ such that $|i-j|>d$. For example, a tri-diagonal matrix is of bandwidth $1$. Clearly, a band matrix with bandwidth $d$ is $(2d+1)$-sparse. 

In this section, we discuss the Hamiltonian embedding of Hermitian band matrices. We will introduce two types of embeddings whose locality explicitly depends on the bandwidth of a target matrix. 

\subsubsection{Unary embedding}\label{append:unary}
\begin{definition}[Unary code]
    Given an integer $j \in \{1,\dots, n\}$, the unary coding of $j$, denoted by $u_j$, is the $(n-1)$-bit string whose first\footnote{In this paper, we use little-endian ordering, i.e., we enumerate the bit of a string \textit{from right to left}. For example, in the string $100$, we refer the first bit to be $0$, and the last bit to be $1$.} $(j-1)$ bits are $1$'s and the last $(n-j)$ bits are $0$'s.
\end{definition}

In \tab{unary_code}, we list the unary code for integers $1 \le j \le 8$.

\begin{table}[!ht]
    \centering
        \begin{tabular}{|p{2cm}|p{2cm}||p{2cm}|p{2cm}|}
     \hline
     \multicolumn{4}{|c|}{\textbf{Unary code ($n = 8$)}} \\
     \hline
     Basis index & Codeword & Basis index & Codeword\\
     \hline
     1 & 0000000 & 5 & 0001111\\
     2 & 0000001 & 6 & 0011111\\
     3 & 0000011 & 7 & 0111111\\
     4 & 0000111 & 8 & 1111111\\
     \hline
    \end{tabular}
    \caption{$7$-bit unary code representing integers $\{1,2,\dots, 8\}$.}
    \label{tab:unary_code}
\end{table}

The unary embedding of an $n$-by-$n$ Hermitian matrix $A$ is an $(n-1)$-qubit quantum operator with the embedding subspace spanned by the unary code, i.e., $\S = \spn\{\ket{u_1},\dots,\ket{u_{n}}\}$. It is easily verified that $\S$ is the ground-energy subspace of the following quantum operator: 
\begin{align}
    \Hpen_{\mathrm{unary}} &= -\sum^{n-2}_{j=1} Z_{j+1}Z_j + Z_1 - Z_{n-1}.
\end{align}

\begin{example}
    For $n = 4$, the penalty Hamiltonian $\Hpen_{\mathrm{unary}}$ has 2 degenerate eigenspaces:
    \begin{align}
        E_0 &= -2: \ket{000}, \ket{001}, \ket{011},\ket{111},\\
        E_1 &= 2: \ket{010},\ket{100}, \ket{101},\ket{110},
    \end{align}
    and the energy gap between the two eigenspaces is $4$. In fact, for any $n \ge 3$, the spectral gap between the ground and the first-excited energy levels of $\Hpen_{\mathrm{unary}}$ is always $4$.
\end{example}

For $1 \le j < k\le n$, we find that 
\begin{align}
    &\bra{u_j}\left(X_{k-1} \otimes \dots \otimes X_{j+1} \otimes X_j\right)\ket{u_k} = 1,~
    \bra{u_k}\left(X_{k-1} \otimes \dots \otimes X_{j+1} \otimes X_j\right)\ket{u_j} = 1,\\
    &\bra{u_j}\left(X_{k-1} \otimes \dots \otimes X_{j+1} \otimes Y_j\right)\ket{u_k} = -i,~
    \bra{u_k}\left(X_{k-1} \otimes \dots \otimes X_{j+1} \otimes Y_j\right)\ket{u_j} = i.
\end{align}

We denote $\hat{n}_j$ as the number operator at site $j$:
\begin{align}\label{eqn:number-operator}
    \hat{n}_j \coloneqq \frac{1}{2}\left(I - Z_j\right).
\end{align}

\begin{proposition}[Unary embedding]
    Given an $n$-dimensional Hermitian matrix $A$, we define the following $(n-1)$-qubit quantum operator:
    \begin{align}
        Q_A &= \left(\alpha_1 I + \sum_{2\le j \le n} (\alpha_{j}-\alpha_{j-1}) \hat{n}_{j-1}\right) + \sum_{1\le j < k \le n} X_{k-1}\otimes \dots \otimes X_{j+1} \otimes (\alpha_{j,k} X_j - \beta_{j,k} Y_j),
    \end{align}
    where $\alpha_j$ is the $j$-th diagonal entry of $A$, $\alpha_{j,k}$ and $\beta_{j,k}$ are the real and imaginary part of the $(j,k)$-th entry of $A$. Then, the quantum operator
    \begin{align}
        \Hebd_A = g \Hpen_{\mathrm{unary}} + Q_A
    \end{align}
    is an embedding of $A$. Moreover, the embedding Hamiltonian $\Hebd_A$ is $\max(d,2)$-local if $A$ has band width $d$.
\end{proposition}
\begin{proof}
    For any $k=1,\dots,n$,
    \begin{equation}
        \bra{u_k} \left(\alpha_1 I + \sum_{2\le j \le n} (\alpha_{j}-\alpha_{j-1}) \hat{n}_{j-1}\right) \ket{u_k} = \alpha_1 + \sum_{2\leq j \leq k} (\alpha_{j} - \alpha_{j-1}) = \alpha_{k}.
    \end{equation}
    For any $\ell,m \in \{1,\dots,n\}$, 
    \begin{align}
        \bra{u_{\ell}} Q_A \ket{u_{m}}
        &= \bra{u_{\ell}} \left(\sum_{1\le j < k \le n} X_{k-1}\otimes \dots \otimes X_{j+1} \otimes (\alpha_{j,k} X_j - \beta_{j,k} Y_j)\right) \ket{u_{m}}\\
        &= \bra{u_{\ell}} \left(X_{m-1} \otimes \dots \otimes X_{\ell+1}\otimes (\alpha_{\ell,m} X_{\ell} - \beta_{\ell,m} Y_{\ell})\right) \ket{u_{m}}\\
        &= \alpha_{\ell,m} + \beta_{\ell,m} i.
    \end{align}
    Letting $P_{\S} = \sum_{k=1}^{n} \ketbra{u_k}$, it follows that $P_{\S} Q_{A} P_{\S} = A$ and thus $\Hebd_A$ is an embedding of $A$.
    If $A$ has bandwidth $d$, then $\alpha_{j,k}=\beta_{j,k}=0$ for $|j-k|>d$, making $Q_{A}$ $d$-local.
    Since $\Hpen_{\mathrm{unary}}$ is 2-local, $\Hebd_A$ is $\max(d,2)$-local.
\end{proof}

\subsubsection{Antiferromagnetic embedding}\label{append:antiferro}
\begin{definition}[Antiferromagnetic code]
    Given an integer $j \in \{1,\dots, n\}$, the antiferromagnetic coding of $j$, denoted by $a_j$, is an $(n-1)$-bit string defined as follows:
    \begin{enumerate}
        \item when $j = 1$, $a_j$ is the $0$-$1$ alternating string starting with $0$;
        \item when $j = 2,\dots, n$, $a_j$ is generated by flipping the $(j-1)$-th bit in the previous codeword $a_{j-1}$.
    \end{enumerate}
\end{definition}

In \tab{antiferro_code}, we list the antiferromagnetic code for integers $1 \le j \le 8$.

\begin{table}[!ht]
    \centering
    \begin{tabular}{|p{2cm}|p{2cm}||p{2cm}|p{2cm}|}
     \hline
     \multicolumn{4}{|c|}{\textbf{Antiferromagnetic code ($n = 8$)}} \\
     \hline
     Basis index & Codeword & Basis index & Codeword\\
     \hline
     1 & 0101010 & 5 & 0100101\\
     2 & 0101011 & 6 & 0110101\\
     3 & 0101001 & 7 & 0010101\\
     4 & 0101101 & 8 & 1010101\\
     \hline
    \end{tabular}
    \caption{$7$-bit antiferromagnetic code representing integers $\{1,2,\dots, 8\}$.}
    \label{tab:antiferro_code}
\end{table}

The antiferromagnetic embedding of an $n$-by-$n$ Hermitian matrix $A$ is an $(n-1)$-qubit quantum operator with the embedding subspace spanned by the antiferromagnetic code, i.e., $\S = \spn\{\ket{a_1},\dots,\ket{a_n}\}$. It is easily verified that $\S$ is the ground-energy subspace of the following quantum operator:
\begin{align}\label{eqn:pen_antiferro}
    \Hpen_{\mathrm{antiferro}} &= \sum^{n-2}_{j=1}Z_{j+1}Z_j + Z_1 + (-1)^{n-1}Z_{n-1}.
\end{align}

\begin{example}
    For $n = 4$, the penalty Hamiltonian $\Hpen_{\mathrm{antiferro}}$ has 2 degenerate eigenspaces:
    \begin{align}
        E_0 &= -2: \ket{001}, \ket{010}, \ket{011},\ket{101},\\
        E_1 &= 2: \ket{000},\ket{100}, \ket{110},\ket{111},
    \end{align}
    and the energy gap between the two eigenspaces is $4$. In fact, for any $n \ge 3$, the spectral gap between the ground and the first-excited energy levels of $\Hpen_{\mathrm{antiferro}}$ is always $4$.
\end{example}

For any $1 \le j < k \le n$, we find that 
\begin{align}
    &\bra{a_j}\left(X_{k-1} \otimes \dots \otimes X_{j+1} \otimes X_j\right)\ket{a_k} = 1,~
    \bra{a_k}\left(X_{k-1} \otimes \dots \otimes X_{j+1} \otimes X_j\right)\ket{a_j} = 1,\\
    &\bra{a_j}\left(X_{k-1} \otimes \dots \otimes X_{j+1} \otimes Y_j\right)\ket{a_k} = -i,~
    \bra{a_k}\left(X_{k-1} \otimes \dots \otimes X_{j+1} \otimes Y_j\right)\ket{a_j} = i.
\end{align}

\begin{proposition}[Antiferromagnetic embedding]
    Given an $n$-dimensional Hermitian matrix $A$, we define the following $(n-1)$-qubit quantum operator:
    \begin{align}
        Q_A &= \left(\gamma I + \sum_{2\le j \le n} (-1)^{j}(\alpha_{j} - \alpha_{j-1})\hat{n}_{j-1}\right) + \sum_{1\le j < k \le n} X_{k-1}\otimes \dots \otimes X_{j+1} \otimes (\alpha_{j,k} X_j - \beta_{j,k} Y_j),
    \end{align}
    where $\alpha_j$ is the $j$-th diagonal entry of $A$, $\alpha_{j,k}$ and $\beta_{j,k}$ are the real and imaginary part of the $(j,k)$-th entry of $A$, and 
    \begin{align}
        \gamma \coloneqq \sum^n_{j=1} (-1)^{j+1}\alpha_j.
    \end{align}
    Then, the quantum operator
    \begin{align}
        \Hebd_A = g \Hpen_{\mathrm{antiferro}} + Q_A
    \end{align}
    is an embedding of $A$. Moreover, the embedding Hamiltonian $\Hebd_A$ is $\max(d,2)$-local if $A$ has band width $d$.
\end{proposition}
\begin{proof}
    Note that
    \begin{equation}
        \left(\gamma I + \sum_{2\le j \le n} (-1)^{j}(\alpha_{j} - \alpha_{j-1})\hat{n}_{j-1}\right)
        = \alpha_1 + \sum_{2\leq j \leq n} (\alpha_j - \alpha_{j-1})\left(\frac{I-(-1)^{j}Z_{j-1}}{2}\right).
    \end{equation}
    Then for any $k=1,\dots,n$,
    \begin{equation}
        \bra{a_k} \left(\alpha_1 + \sum_{2\leq j \leq n} (\alpha_j - \alpha_{j-1})\left(\frac{I-(-1)^{j}Z_{j-1}}{2}\right)\right) \ket{a_k}
        = \alpha_1 + \sum_{2 \leq j \leq k} (\alpha_j - \alpha_{j-1})
        = \alpha_k.
    \end{equation}
    For $1 \leq \ell < m \leq n$, 
    \begin{align}
        \bra{a_{\ell}} Q_A \ket{a_{m}}
        &= \bra{a_{\ell}} \left(\sum_{1\le j < k \le n} X_{k-1}\otimes \dots \otimes X_{j+1} \otimes (\alpha_{j,k} X_j - \beta_{j,k} Y_j)\right) \ket{a_{m}}\\
        &= \bra{a_{\ell}} \left(X_{m-1} \otimes \dots \otimes X_{\ell+1}\otimes (\alpha_{\ell,m} X_{\ell} - \beta_{\ell,m} Y_{\ell})\right) \ket{a_{m}}\\
        &= \alpha_{\ell,m} + \beta_{\ell,m} i.
    \end{align}
    Letting $P_{\S} = \sum_{k=1}^{n} \ketbra{a_k}$, it follows that $P_{\S} Q_{A} P_{\S} = A$ and thus $\Hebd_A$ is an embedding of $A$.
    If $A$ has bandwidth $d$, then $\alpha_{j,k}=\beta_{j,k}=0$ for $|j-k|>d$, making $Q_{A}$ $d$-local.
    Since $\Hpen_{\mathrm{unary}}$ is 2-local, $\Hebd_A$ is $\max(d,2)$-local.
\end{proof}

\subsection{Banded circulant matrix}\label{append:circulant}
In this section, we discuss how to embed (real) symmetric circulant matrices of bandwidth $1$ (i.e., tridiagonal) into $2$-local quantum operators. Our construction can be readily generalized to circulant matrices of bandwidth $d$.
It is worth noting that the embedding discussed in this section only applies to circulant matrices of even dimensions. 

Given an even number $n$, we consider the following $n$-by-$n$ circulant matrix:
\begin{align}\label{eqn:circulant}
    C = \begin{bmatrix}
        0 & 1 & 0 & \dots & 1 \\
        1 & 0 & 1 & \dots & 0\\
        \dots & \dots & \dots & \dots & \dots\\
        0 & \dots & 1 & 0 & 1\\
        1 & \dots & 0 & 1 & 0
    \end{bmatrix}.
\end{align}
The matrix $C$ is precisely the adjacency matrix of the (unweighted) $n$-node cycle graph. We introduce two types of embeddings for the matrix $C$.

\subsubsection{Circulant unary embedding}\label{append:circ-unary}
\begin{definition}[Circulant unary code]
    Given an integer $j \in \{1,\dots, n\}$, the circulant unary coding of $j$, denoted by $c_j$, is an $\frac{n}{2}$-bit string defined as follows:
    \begin{enumerate}
        \item when $j = 1,\dots, n/2$, the first $(j-1)$ bits of $c_j$ are $1$'s and the remaining $(n/2+1-j)$ bits are $0$'s;
        \item when $j = n/2+1,\dots, n$, the codeword $c_j$ is the negation of $c_{j-n/2}$.
    \end{enumerate}
\end{definition}

In \tab{unary_code_circulant}, we list the circulant unary code for integers $1 \le j \le 8$.

\begin{table}[!ht]
    \centering
    \begin{tabular}{|p{2cm}|p{2cm}||p{2cm}|p{2cm}|}
     \hline
     \multicolumn{4}{|c|}{\textbf{Circulant unary code ($n = 8$)}} \\
     \hline
     Basis index & Codeword & Basis index & Codeword\\
     \hline
     1 & 0000 & 5 & 1111\\
     2 & 0001 & 6 & 1110\\
     3 & 0011 & 7 & 1100\\
     4 & 0111 & 8 & 1000\\
     \hline
    \end{tabular}
    \caption{$4$-bit circulant unary code representing integers $\{1,2,\dots, 8\}$.}
    \label{tab:unary_code_circulant}
\end{table}

The circulant unary embedding of the matrix $C$ is an $\frac{n}{2}$-qubit quantum operator with the embedding subspace spanned by the circulant unary code, i.e., $\S = \spn\{\ket{c_1},\dots,\ket{c_{n}}\}$. It is easily verified that $\S$ is the ground-energy subspace of the following quantum operator:
\begin{align}
    \Hpen_{\mathrm{circ-unary}} = -\sum^{n/2-1}_{j=1} Z_{j+1}Z_j  + Z_{n/2}Z_1.
\end{align}

\begin{example}
    For $n = 8$, the penalty Hamiltonian $\Hpen_{\mathrm{circ-unary}}$ has 2 degenerate eigenspaces:
    \begin{align}
        E_0 &= -2: \ket{0000}, \ket{0001}, \ket{0011},\ket{0111},\ket{1000},\ket{1100},\ket{1110},\ket{1111},\\
        E_1 &= 2: \ket{0010},\ket{0100}, \ket{0101},\ket{0110},\ket{1001},\ket{1010},\ket{1011},\ket{1101},
    \end{align}
    and the energy gap between the two eigenspaces is $4$. In fact, for any $n \ge 4$, the spectral gap between the ground and the first-excited energy levels of $\Hpen_{\mathrm{circ-unary}}$ is always $4$.
\end{example}

\begin{proposition}[Circulant unary embedding]
    We define the following $\frac{n}{2}$-qubit quantum operator:
    \begin{align}
        Q_C = \sum^{n/2}_{j=1} X_j.
    \end{align}
    Then, the quantum operator
    \begin{align}
        \Hebd_C = g \Hpen_{\mathrm{circ-unary}} + Q_C
    \end{align}
    is an embedding of the circulant matrix $C$ as in~\eqn{circulant}. 
\end{proposition}
\begin{proof}
    Observe that $Q_C$ is the adjacency matrix of the $(n/2)$-dimensional hypercube, containing edges between bitstrings with Hamming distance 1.
    Therefore, for any $1\leq k < \ell \leq n$,
    \begin{equation}
        \bra{c_k} Q_c \ket{c_\ell} =
        \begin{cases}
            1 & \text{if $|k-\ell|=1$,}\\
            0 & \text{otherwise.}
        \end{cases}
    \end{equation}
\end{proof}

\subsubsection{Circulant antiferromagnetic embedding}\label{append:circ-antiferro}
\begin{definition}[Circulant antiferromagnetic code]
    Given an integer $j \in \{1,\dots,n\}$, the circulant antiferromagnetic coding of $j$, denoted by $f_j$, is an $\frac{n}{2}$-bit string defined as follows:
    \begin{enumerate}
        \item when $j = 1$, the first codeword $f_1$ is the $0$-$1$ alternating string starting with $0$;
        \item when $j = 2,\dots, n/2$, $f_j$ is generated by flipping the $(j-1)$-th bit in the previous codeword $f_{j-1}$;
        \item when $j = n/2+1,\dots, n$, the codeword $f_j$ is the negation of $f_{j-n/2}$.
    \end{enumerate}
\end{definition}

In \tab{antiferro_code_circulant}, we list the circulant antiferromagnetic code for integers $1\le j \le 8$.

\begin{table}[!ht]
    \centering
    \begin{tabular}{|p{2cm}|p{2cm}||p{2cm}|p{2cm}|}
     \hline
     \multicolumn{4}{|c|}{\textbf{Circulant antiferromagnetic code ($n = 8$)}} \\
     \hline
     Basis index & Codeword & Basis index & Codeword\\
     \hline
     1 & 1010 & 5 & 0101\\
     2 & 1011 & 6 & 0100\\
     3 & 1001 & 7 & 0110\\
     4 & 1101 & 8 & 0010\\
     \hline
    \end{tabular}
    \caption{$4$-bit circulant antiferromagnetic code representing integers $\{1,2,\dots, 8\}$.}
    \label{tab:antiferro_code_circulant}
\end{table}

The circulant antiferromagnetic embedding of the matrix $C$ is an $\frac{n}{2}$-qubit quantum operator with the embedding subspace spanned by the circulant antiferromagnetic code, i.e., $\S = \{\ket{f_1},\dots,
\ket{f_n}\}$. It is easily verified that $\S$ is the ground-energy subspace of the following quantum operator:
\begin{align}
    \Hpen_{\mathrm{circ-antiferro}} = \sum^{n/2-1}_{j=1}Z_{j+1}Z_{j} - (-1)^{n/2} Z_{n/2}Z_1.
\end{align}

\begin{example}
    For $n = 8$, the penalty Hamiltonian $\Hpen_{\mathrm{circ-antiferro}}$ has 2 degenerate eigenspaces:
    \begin{align}
        E_0 &= -2: \ket{0010},\ket{0100}, \ket{0101},\ket{0110},\ket{1001},\ket{1010},\ket{1011},\ket{1101},\\
        E_1 &= 2: \ket{0000}, \ket{0001}, \ket{0011},\ket{0111},\ket{1000},\ket{1100},\ket{1110},\ket{1111},
    \end{align}
    and the energy gap between the two eigenspaces is $4$. In fact, for any $n \ge 4$, the spectral gap between the ground and the first-excited energy levels of $\Hpen_{\mathrm{circ-antiferro}}$ is always $4$.
\end{example}

\begin{proposition}[Circulant antiferromagnetic embedding]
    We define the following $\frac{n}{2}$-qubit quantum operator:
    \begin{align}
        Q_C = \sum^{n/2}_{j=1} X_j.
    \end{align}
    Then, the quantum operator
    \begin{align}
        \Hebd_C = g \Hpen_{\mathrm{circ-antiferro}} + Q_C
    \end{align}
    is an embedding of the circulant matrix $C$ as in~\eqn{circulant}.
\end{proposition}
\begin{proof}
    Observe that $Q_C$ is the adjacency matrix of the $(n/2)$-dimensional hypercube, containing edges between bitstrings with Hamming distance 1.
    Therefore, for any $1\leq k < \ell \leq n$,
    \begin{equation}
        \bra{f_k} Q_c \ket{f_\ell} =
        \begin{cases}
            1 & \text{if $|k-\ell|=1$,}\\
            0 & \text{otherwise.}
        \end{cases}
    \end{equation}
\end{proof}

\subsection{Arbitrary sparse matrices: one-hot embedding}\label{append:one-hot}
For arbitrary Hermitian matrix $H$, the bandwidth of $H$ could be huge even if it is sparse. For such target matrices, the unary or antiferromagnetic embedding protocol will lead to a Hamiltonian embedding with a high locality (depending on the bandwidth of $H$). In this section, we introduce a new type of Hamiltonian embedding based on the one-hot code that has a constant locality.

\begin{definition}[One-hot code]
    Given an integer $j \in \{1,\dots, n\}$, the one-hot coding of $j$, denoted by $h_j$, is an $n$-bit string whose all but the $j$-th bit are $0$'s. 
\end{definition}

In \tab{one_hot_code}, we list the one-hot code for integers $1\le j \le 8$.

\begin{table}[!ht]
    \centering
    \begin{tabular}{|p{2cm}|p{2cm}||p{2cm}|p{2cm}|}
     \hline
     \multicolumn{4}{|c|}{\textbf{One-hot code ($n = 8$)}} \\
     \hline
     Basis index & Codeword & Basis index & Codeword\\
     \hline
     1 & 00000001 & 5 & 00010000\\
     2 & 00000010 & 6 & 00100000\\
     3 & 00000100 & 7 & 01000000\\
     4 & 00001000 & 8 & 10000000\\
     \hline
    \end{tabular}
    \caption{$8$-bit antiferromagnetic code representing integers $\{1,2,\dots, 8\}$.}
    \label{tab:one_hot_code}
\end{table}

\subsubsection{One-hot embedding with penalty}\label{append:onehot}
We remark that the edit distance between any two codewords in the one-hot code is $2$. This nice property allows us to embed any matrix element with a $2$-local operator. Indeed, for $1 \le j < k \le n$, we find that 
\begin{align}
    &\bra{h_j}\left(X_{k} \otimes X_j\right)\ket{h_k} = 1,~
    \bra{h_k}\left(X_{k} \otimes X_j\right)\ket{h_j} = 1,\\
    &\bra{h_j}\left(X_{k} \otimes Y_j\right)\ket{h_k} = i,~
    \bra{h_k}\left(X_{k} \otimes Y_j\right)\ket{h_j} = -i.
\end{align}

Meanwhile, the matrix elements on the main diagonal of the target Hamiltonian can be embedded into the number operator. For $1\le j, k \le n$, we observe that (note that the number operator $\hat{n}_j$ is defined in \eqn{number-operator})
\begin{align}
    \bra{h_j}\hat{n}_{j}\ket{h_j} = \begin{cases}
        0 & (j \neq k)\\
        1 & (j = k).
    \end{cases}
\end{align}

Denote $\S$ as the embedding subspace spanned by the one-hot code states. Unfortunately, neither $X_k\otimes X_j$ nor $X_k\otimes Y_j$ is block-diagonalized by $\S$ and $\S^\perp$. Therefore, we need a penalty Hamiltonian in the full embedding Hamiltonian.

The one-hot code are $n$-bit strings with Hamming weight $1$, so they span the $(n-2)$-eigenspace of the ``sum-of-Z'' operator:
\begin{align}
    \Hpen_{\mathrm{one-hot}} = \sum^n_{j=1} Z_j.
\end{align}
We also note that the one-hot code spans the ground-energy subspace of the following operator:
\begin{align}\label{eqn:one-hot-ground}
    \Hpen_{\mathrm{one-hot,g}} = \left(\sum^n_{j=1} \hat{n}_j - 1\right)^2.
\end{align}

\begin{example}
    For $n = 3$, the penalty Hamiltonian $\Hpen_{\mathrm{one-hot}}$ has $4$ eigenspaces that correspond to computational basis with Hamming weight $3, 2, 1, 0$, respectively:
    \begin{align}
        E_0 &= -3: \ket{111},\\
        E_1 &= -1: \ket{011},\ket{101},\ket{110},\\
        E_2 &= 1: \ket{001}, \ket{010}, \ket{100},\\
        E_3 &= 3: \ket{000}.
    \end{align}
    Clearly, the one-hot code spans the $1$-eigenspace. Meanwhile, the penalty Hamiltonian $\Hpen_{\mathrm{one-hot,g}}$ has $3$ eigenspaces:
    \begin{align}
        E_0 &= 0: \ket{011},\ket{101},\ket{110},\\
        E_1 &= 1: \ket{000}, \ket{011},\ket{101},\ket{110},\\
        E_2 &= 4: \ket{111},\\
    \end{align}
    and the one-hot code spans the ground-energy subspace.
\end{example}

\begin{proposition}[One-hot embedding with penalty]
    Given an $n$-dimensional Hermitian matrix $A$, we define the following $n$-qubit quantum operator:
    \begin{align}
        Q_A = \left(\sum^n_{j=1} \alpha_j \hat{n}_j\right) + \left(\sum_{1 \le j < k \le n} \alpha_{j,k} X_k \otimes X_j + \beta_{j,k} X_k\otimes Y_j\right),
    \end{align}
    where $\alpha_j$ is the $j$-th diagonal entry of $A$, $\alpha_{j,k}$ and $\beta_{j,k}$ are the real and imaginary part of the $(j,k)$-th entry of $A$, respectively.
    Then, the $2$-local quantum operator 
    \begin{align}
        \Hebd_A = g \Hpen_{\mathrm{one-hot}} + Q_A 
    \end{align}
    is an embedding of $A$.
\end{proposition}
\begin{proof}
    For any $j=1,\dots,n$,
    \begin{equation}
        \bra{h_j} \Hebd_A \ket{h_j} = \bra{h_j} \left(\sum^n_{j=1} \alpha_j \hat{n}_j\right) \ket{h_j} = \alpha_j.
    \end{equation}
    For $1 \leq \ell < m \leq n$, 
    \begin{align}
        \bra{h_{\ell}} \Hebd_A \ket{h_{m}} 
        &= \frac{1}{2}\bra{h_{\ell}} \left(\sum_{1 \le j < k \le n} \alpha_{j,k} (X_k \otimes X_j) + \beta_{j,k} (X_k\otimes Y_j)\right) \ket{h_{m}}\\
        &= \bra{h_{\ell}} \left( \alpha_{\ell,m} (X_{m} \otimes X_{\ell} + \beta_{\ell,m} (X_{m} \otimes Y_{\ell})\right) \ket{h_{m}}\\
        &= \alpha_{\ell,m} + \beta_{\ell,m} i
    \end{align}
    Thus, $\Hebd_A$ is an embedding of $A$.
\end{proof}

\subsubsection{Penalty-free one-hot embedding}\label{append:onehot_penfree}
The one-hot code allows one to design the perturbation Hamiltonian to be block-diagonal, thereby eliminating the need for a penalty term.
In physics, the one-hot code is known as the single-excitation subspace, and the leakage can be fully eliminated by considering excitation-preserving Hamiltonians.
In the digital setting, this corresponds to excitation-preserving gates, which often arise in simulations of fermionic systems \cite{arute2020observation}. 

There are two parameterized gates which are particularly relevant to the one-hot code.
The first is the fSim gate \cite{foxen2020demonstrating} which acts as a Pauli-$X$ rotation in the $\{\ket{01},\ket{10}\}$ subspace.
The other is the RBS gate \cite{johri2021nearest, mathur2021medical}, which acts as a Pauli-$Y$ rotation in the $\{\ket{01},\ket{10}\}$ subspace.
For both of these gates, we can ignore the local phase in the $\{\ket{00},\ket{11}\}$ subspace.
Based on these observations, we consider the following two Hamiltonians:
\begin{align}
    \hat{X} &\coloneqq \begin{pmatrix}
        0 & 0 & 0 & 0\\
        0 & 0 & 1 & 0\\
        0 & 1 & 0 & 0\\
        0 & 0 & 0 & 0
    \end{pmatrix}
    = \frac{1}{2}(X \otimes X + Y \otimes Y) = \ket{10}\bra{01}+\ket{01}\bra{10},\\
    \hat{Y} &\coloneqq \begin{pmatrix}
        0 & 0 & 0 & 0\\
        0 & 0 & -i & 0\\
        0 & i & 0 & 0\\
        0 & 0 & 0 & 0
    \end{pmatrix}
    = \frac{1}{2}(Y \otimes X - X \otimes Y) = i\ket{10}\bra{01}-i\ket{01}\bra{10}.
\end{align}
Note that $\comm{X \otimes X}{Y \otimes Y}=\comm{X \otimes Y}{Y \otimes X}=0$, giving rise to a straightforward implementation of $e^{-i\hat{X}}$ and $e^{-i\hat{Y}}$ each using two parameterized two-qubit gates.

In particular, the subspace spanned by the one-hot code is an invariant subspace of $\hat{X}$ and $\hat{Y}$.
That is, for $1\leq j < k \leq n$,
\begin{align}
    &\frac{1}{2} \left(X_{k} \otimes X_{j} + Y_{k} \otimes Y_{j}\right) \ket{h_j} = \ket{h_k},~
    \frac{1}{2} \left(X_{k} \otimes X_{j} + Y_{k} \otimes Y_{j}\right) \ket{h_k} = \ket{h_j},\\
    &\frac{1}{2} \left(X_{k} \otimes Y_{j} - Y_{k} \otimes X_{j}\right) \ket{h_j} = i\ket{h_k},~
    \frac{1}{2} \left(X_{k} \otimes X_{j} - Y_{k} \otimes Y_{j}\right) \ket{h_k} = -i\ket{h_j}.
\end{align}

Similar to the one-hot embedding with penalty, any matrix element is embedded with a $2$-local operator.
Indeed, for $1 \le j < k \le n$, 
\begin{align}
    &\bra{h_j}\frac{\left(X_{k} \otimes X_j + Y_{k} \otimes Y_{j}\right)}{2}\ket{h_k} = 1,~
    \bra{h_k}\frac{\left(X_{k} \otimes X_{j} + Y_{k} \otimes Y_{j}\right)}{2}\ket{h_j} = 1,\\
    &\bra{h_j}\frac{\left(X_{k} \otimes Y_{j} - Y_{k} \otimes X_{j}\right)}{2}\ket{h_k} = i,~
    \bra{h_k}\frac{\left(X_{k} \otimes Y_{j} - Y_{k} \otimes X_{j}\right)}{2}\ket{h_j} = -i.
\end{align}
Meanwhile, for $1\le j, k \le n$, we observe that
\begin{align}
    \bra{h_j}n_{j}\ket{h_j} = \begin{cases}
        0 & (j \neq k)\\
        1 & (j = k).
    \end{cases}
\end{align}

\begin{proposition}[One-hot embedding without penalty]
    Given an $n$-dimensional Hermitian matrix $A$, we define the following $n$-qubit quantum operator:
    \begin{align}
        \Hebd_A = \left(\sum^n_{j=1} \alpha_j \hat{n}_j\right) + \frac{1}{2}\left(\sum_{1 \le j < k \le n} \alpha_{j,k} (X_k \otimes X_j + Y_k \otimes Y_j) + \beta_{j,k} (X_k\otimes Y_j - Y_k \otimes X_j)\right),
    \end{align}
    where $\alpha_j$ is the $j$-th diagonal entry of $A$, $\alpha_{j,k}$ and $\beta_{j,k}$ are the real and imaginary part of the $(j,k)$-th entry of $A$, respectively.
    Then, the $2$-local quantum operator $\Hebd_A$ is an embedding of $A$.
\end{proposition}
\begin{proof}
    For any $j=1,\dots,n$,
    \begin{equation}
        \bra{h_j} \Hebd_A \ket{h_j} = \bra{h_j} \left(\sum^n_{j=1} \alpha_j \hat{n}_j\right) \ket{h_j} = \alpha_j.
    \end{equation}
    For $1 \leq \ell < m \leq n$, 
    \begin{align}
        \bra{h_{\ell}} \Hebd_A \ket{h_{m}} 
        &= \frac{1}{2}\bra{h_{\ell}} \left(\sum_{1 \le j < k \le n} \alpha_{j,k} (X_k \otimes X_j + Y_k \otimes Y_j) + \beta_{j,k} (X_k\otimes Y_j - Y_k \otimes
 X_j)\right) \ket{h_{m}}\\
        &= \frac{1}{2}\bra{h_{\ell}} \left( \alpha_{\ell,m} (X_{m} \otimes X_{\ell} + Y_{m} \otimes Y_{\ell}) + \beta_{\ell,m} (X_{m} \otimes Y_{\ell} - Y_{m} \otimes X_{\ell})\right) \ket{h_{m}}\\
        &= \alpha_{\ell,m} + \beta_{\ell,m} i
    \end{align}
    Since the encoding subspace for the one-hot code is an invariant subspace of $\Hebd_A$, no penalty Hamiltonian is needed.
    Thus, $\Hebd_A$ is an embedding of $A$.
\end{proof}

\section{Methodology of resource analysis}\label{append:methodology}
For each of the quantum simulation tasks considered, we formulate a target goal which corresponds to simulating the Hamiltonian for some time $T$, typically depending on the system size.
We numerically compute an estimate of the resources (i.e. total gate count) required to perform each target task using Hamiltonian embedding.
In all of the resource estimates considered, we do not consider any hardware specific error or additional overheads that may come from error correction.

As a baseline, we consider the naive approach of computing the Pauli decomposition of a Hamiltonian and performing simulation by Trotter-Suzuki product formulas.
The Pauli decomposition of a Hamiltonian implicitly assumes the use of standard binary encoding to represent each computational basis state.
The standard binary encoding is also the typical approach used extensively in quantum algorithms (for example, in the commonly used circuit for the quantum Fourier transform).
When the system size $N$ is not a power of 2, we start from the problem Hamiltonian to a larger one of dimension $2^{\lceil \lg N \rceil} \times 2^{\lceil \lg N \rceil}$ and remove any off-diagonal terms corresponding to edges between the target subgraph and the unused component.
When the problem Hamiltonian has nontrivial diagonal terms, we additionally modify the diagonal terms to ensure a correct encoding of the desired problem Hamiltonian.

For each task experimentally demonstrated on the IonQ processor, we compare the estimated gate count of the standard binary encoding with other Hamiltonian embedding schemes.
For a fair comparison, we fix the same Hamiltonian simulation method (i.e. the choice of product formula) and the desired precision across the different embedding methods.
For the standard binary encoding, we use Qiskit \cite{Qiskit} to compute the Pauli decomposition of the (possibly modified) problem Hamiltonian and compile the circuit to one- and two-qubit gates.
In addition to optimizing the circuit using the Qiskit transpiler, we further optimize the circuit using TKET \cite{sivarajah2020t}. 
The basis gate set is chosen to include parameterized single-qubit Pauli rotations and $XX$-rotation gates: $\{R_x(\theta), R_y(\theta), R_z(\theta), R_{xx}(\theta)\}$.
The additive Trotter error is estimated by Monte Carlo sampling.
In all resource estimates performed, we use binary search to numerically compute the Trotter number to achieve error within $\epsilon = 5\times 10^{-2}$.
For reference, we also include theoretical estimates of the gate count, computed using the worst-case Trotter error bound for the unary and one-hot embeddings \cite{childs2019faster, childs2021theory}.
We fit the data to power laws to obtain estimates of the asymptotic gate complexity.

For the experiments on QuEra, we do not perform a resource comparison between the standard binary encoding and the use of Hamiltonian embedding for simulating real-space dynamics.
The programmability of analog quantum simulators naturally forbids the use of dense encodings (such as the standard binary encoding and the Gray code) to simulate the desired problem Hamiltonian directly.
For the standard binary code, adjacent vertices in the graph of the problem Hamiltonian may have Hamming distance larger than 1 (for instance, bitstrings 0111 and 1000 representing 7 and 8, respectively).
On the other hand, the Gray code suffers from the issue of having unwanted interactions, even if the Hamming distance is small (for instance, bitstrings 0111 and 1111 representing 7 and 15, respectively).

\section{Quantum walk on graphs}\label{append:qwalk}

Given a graph $G = (V, E)$, the continuous-time quantum walk on the graph $G$ is described by the Schr\"odinger equation,
\begin{align}\label{eqn:quantum-walk-equation}
    i \frac{\d}{\d t}\ket{\psi} = L \ket{\psi},
\end{align}
where $L$ is the graph Laplacian defined by:
\begin{align}\label{eqn:graph-laplacian}
    (L)_{j,k} = \begin{cases}
        0 & (j,k) \notin E\\
        1 & (j,k) \in E\\
        - \deg(j) & j = k
    \end{cases}
\end{align}

In this section, we construct Hamiltonian embeddings of the graph Laplacian $L$ for several important families of graphs, including finite-dimensional lattices and trees.

\subsection{Lattice graphs}

In this section, we discuss the embedding of $d$-dimensional Lattice graphs with two types of boundary conditions: the regular lattice graphs are grids with open boundaries, and the periodic lattice graphs have periodic boundary conditions (i.e., torus-like shape). For both types of graphs, the size of the graph Laplacian $L$ scales exponentially in dimension $d$. By leveraging a tensor-product decomposition of $L$, we construct $2$-local Hamiltonian embeddings of the graph Laplacian with $O(d)$ qubits.

\subsubsection{Regular lattice graphs}\label{append:reg_lattice}
\begin{definition}[Regular lattice graphs]
    Fix an integer $N \ge 1$, the $d$-dimensional regular lattice graph $\mathcal{G}_{\mathrm{reg}} = (\mathcal{V}_{\mathrm{reg}}, \mathcal{E}_{\mathrm{reg}})$ has $N^d$ vertices:
    \begin{align}
        \V_{\mathrm{reg}} = \{\mathbf{v} = (v_1,v_2,\dots, v_d): v_1,\dots,v_d \in \{1,\dots,N\}\},
    \end{align}
    and there is an edge between $\mathbf{u}$ and $\mathbf{v}$ if and only if their coordinates differ by 1 at a single site. In other words, $(\mathbf{u}, \mathbf{v}) \in \mathcal{E}_{\mathrm{reg}}$ if there exists an index $k\in \{1,\dots,d\}$ such that $|v_k - u_k| = 1$ and $v_j = u_j$ for all $j \neq k$.
\end{definition}

Regular lattice graphs are sometimes referred to as grid graphs. The graph Laplacian of the $d$-dimensional regular lattice graph is given by
\begin{align}\label{eqn:d-regular-laplacian}
    L_{\mathrm{reg}} = \sum^d_{k=1} I \otimes \dots \otimes \underbrace{L_{\mathrm{chain}}}_{\text{the $k$-th operator}} \otimes \dots \otimes I,
\end{align}
where $L_{\mathrm{chain}}$ is the graph Laplacian of the 1D chain graph (with $N$ nodes):
\begin{align}
    L_{\mathrm{chain}} = \begin{bmatrix} -1 & 1 & & \\
    1 & -2 & 1 & \\
    ...& ... & ... &...\\
    & 1 & -2 & 1\\
    & & 1 & -1\\
    \end{bmatrix}
\end{align}

From \eqn{d-regular-laplacian}, it is clear that the graph Laplacian $L_{\mathrm{reg}}$ is the composition of $d$ independent copies of $L_{\mathrm{chain}}$. Since the Laplacian of the 1D chain graph is a tridiagonal matrix, we can use an embedding scheme in \append{band} or the one-hot embedding\footnote{It is worth noting that we must use the penalty Hamiltonian defined in \eqn{one-hot-ground} because \lem{composition} requires the embedding subspace to be the ground-energy subspace of $\Hpen$.} (see \append{one-hot}) to construct an embedding of $L_{\mathrm{chain}}$.

\begin{enumerate}
    \item Unary embedding of $L_{\mathrm{chain}}$: $$\Hpen = -\sum^{N-2}_{j=1} Z_{j+1}Z_j + Z_1 - Z_{N-1}, \quad Q = (-I-\hat{n}_1+\hat{n}_{N-1}) + \sum^{N-1}_{j=1} X_j.$$
    \item Antiferromagnetic embedding of $L_{\mathrm{chain}}$: $$\Hpen = \sum^{N-2}_{j=1}Z_{j+1}Z_j + Z_1 + (-1)^{N-1}Z_{N-1}, \quad Q = -2(N-1)I - \hat{n}_1 + (-1)^N \hat{n}_{N-1} + \sum^{N-1}_{j=1} X_j.$$
    \item One-hot embedding of $L_{\mathrm{chain}}$ (with penalty): $$\Hpen = \left(\sum^N_{j=1} \hat{n}_j - 1\right)^2, \quad Q = \left(-\hat{n}_1 - \hat{n}_N -2\sum^{N-1}_{j=2} \hat{n}_j\right) + \sum^{N-1}_{j=1} X_{j+1}X_j.$$
    \item Penalty-free one-hot embedding of $L_{\mathrm{chain}}$: $$\Hebd = \left(-\hat{n}_1 - \hat{n}_N -2\sum^{N-1}_{j=2} \hat{n}_j\right) + \sum^{N-1}_{j=1} \left(X_{j+1}X_j + Y_{j+1}Y_j\right).$$
\end{enumerate}

Next, we use \lem{composition} to construct an embedding of $L_{\mathrm{reg}}$ using $d$ independent embeddings of $L_{\mathrm{chain}}$. In \tab{reg_resource}, we summarize the resources required in different Hamiltonian embeddings of $L_{\mathrm{reg}}$.

\begin{table}[!ht]
    \centering
    \begin{tabular}{|p{4cm}|p{3cm}|p{4.5cm}|p{2.5cm}|}
     \hline
     \textbf{Embedding scheme} & \textbf{\# of qubits} & \textbf{\# of 2-local operators} & \textbf{Max. weight}\\
     \hline
     Unary & $d(N-1)$ & $d(N-2)$ & 2\\
     \hline
     Antiferromagnetic & $d(N-1)$ & $d(N-2)$ & 2\\
     \hline
     One-hot w/ penalty & $dN$ & $d(N-1)(N+2)/2$ & 2\\
     \hline
     Penalty-free one-hot & $dN$ & $2d(N-1)$ & 2\\
     \hline
    \end{tabular}
    \caption{Resource count for various Hamiltonian embeddings of $L_{\mathrm{reg}}$.}
    \label{tab:reg_resource}
\end{table}

\subsubsection{Periodic lattice graphs}
\begin{definition}[Periodic lattice graphs]
    Fix an integer $N \ge 1$, the $d$-dimensional periodic lattice graph 
    $\mathcal{G}_{\mathrm{prd}} = (\mathcal{V}_{\mathrm{prd}}, \mathcal{E}_{\mathrm{prd}})$ has $N^d$ vertices:
    \begin{align}
        \V_{\mathrm{prd}} = \{\mathbf{v} = (v_1,v_2,\dots, v_d): v_1,\dots,v_d \in \Z_N\},
    \end{align}
    and there is an edge between $\mathbf{u}$ and $\mathbf{v}$ if and only if there exists and index $k\in \{1,\dots,d\}$ such that $v_k = u_k \pm 1$ and $v_j = u_j$ for all $j \neq k$.
\end{definition}

The graph Laplacian of the $d$-dimensional periodic lattice graph is given by
\begin{align}\label{eqn:d-periodic-laplacian}
    L_{\mathrm{prd}} = \sum^d_{k=1} I \otimes \dots \otimes \underbrace{L_{\mathrm{cycle}}}_{\text{the $k$-th operator}} \otimes \dots \otimes I,
\end{align}
where $L_{\mathrm{cycle}}$ is the graph Laplacian of the 1D cycle graph (with $N$ nodes):
\begin{align}
    L_{\mathrm{cycle}} = \begin{bmatrix} -2 & 1 & & 1\\
    1 & -2 & 1 & \\
    ...& ... & ... &...\\
    & 1 & -2 & 1\\
    1 & & 1 & -2\\
    \end{bmatrix}
\end{align}

Similarly, the graph Laplacian $L_{\mathrm{prd}}$ is the composition of $d$ independent copies of $L_{\mathrm{cycle}}$. We observe that $L_{\mathrm{cycle}}$ is a tridiagonal circulant matrix, so we can use an embedding scheme in \append{circulant} or the one-hot embedding to construct an embedding of $L_{\mathrm{cycle}}$. Note that we omit the $-2$ global phase in the following embeddings.

\begin{enumerate}
    \item Circulant unary embedding of $L_{\mathrm{cycle}}$ ($N$ must be an even number): $$\Hpen = -\sum^{N/2-1}_{j=1} Z_{j+1}Z_j  + Z_{N/2}Z_1, \quad Q = \sum^{N/2}_{j=1} X_j.$$
    \item Circulant antiferromagnetic embedding of $L_{\mathrm{cycle}}$ ($N$ must be an even number): $$\Hpen = \sum^{N/2-1}_{j=1}Z_{j+1}Z_{j} - (-1)^{N/2} Z_{N/2}Z_1, \quad Q = \sum^{N/2}_{j=1} X_j.$$
    \item One-hot embedding of $L_{\mathrm{cycle}}$ (with penalty, $N$ can be even or odd): $$\Hpen = \left(\sum^N_{j=1} \hat{n}_j - 1\right)^2, \quad Q = \sum^{N-1}_{j=1} X_{j+1}X_j + X_N X_1.$$
    \item Penalty-free one-hot embedding of $L_{\mathrm{cycle}}$: $$\Hebd = \sum^{N-1}_{j=1} \left(X_{j+1}X_j + Y_{j+1}Y_j\right) + X_N X_1 + Y_N Y_1.$$
\end{enumerate}

An embedding of $L_{\mathrm{prd}}$ can be obtained by \lem{composition}. In \tab{prd_resource}, we summarize the resources required in different Hamiltonian embeddings of $L_{\mathrm{prd}}$.

\begin{table}[!ht]
    \centering
    \begin{tabular}{|p{4cm}|p{3cm}|p{4.5cm}|p{2.5cm}|}
     \hline
     \textbf{Embedding scheme} & \textbf{\# of qubits} & \textbf{\# of 2-local operators} & \textbf{Max. weight}\\
     \hline
     Circ-Unary & $dN/2$ & $dN/2$ & 2\\
     \hline
     Circ-Antiferromagnetic & $dN/2$ & $dN/2$ & 2\\
     \hline
     One-hot w/ penalty & $dN$ & $dN(N+1)/2$ & 2\\
     \hline
     Penalty-free one-hot & $dN$ & $2dN$ & 2\\
     \hline
    \end{tabular}
    \caption{Resource count for various Hamiltonian embeddings of $L_{\mathrm{prd}}$.}
    \label{tab:prd_resource}
\end{table}

\subsubsection{Experiments on IonQ}
We conduct experiments for quantum walk on both the 1D chain and 1D cycle graphs on the IonQ Aria-1 processor.

For both the 1D chain and 1D cycle, we demonstrate quantum walk on $N=15$ vertices with the penalty-free one-hot embedding.
The initial state is taken to be the middle vertex in the case of the regular lattice, prepared using a single one-qubit gate.
Due to symmetry, the choice of the vertex does not matter for quantum walk on a periodic lattice.
The maximum evolution time is taken to be $T=4.0$ with $r=5$ Trotter steps using the randomized first-order Trotter-Suzuki formula.
The two-qubit gate counts correspond to the number of 2-local operators, as shown in \tab{reg_resource}.
For each circuit, the two-qubit gate counts are 140 and 150 for the 1D chain and 1D cycle, respectively.

For the experimental demonstrations of quantum walk on lattices, we adopt the definition of continuous time quantum walk using the adjacency matrix of the graph, rather than the graph Laplacian.
The choice of using the adjacency matrix is merely for simplicity, as it avoids introducing additional Pauli-$Z$ rotations.
Since the lattice graph without periodic boundary conditions is not a regular graph, this gives rise to different dynamics.
For the periodic lattice, the adjacency matrix differs from the Laplacian by a constant, so the dynamics are equivalent.

In panels A, B, D, and E of \fig{experiments_lattice}, we show experimental results simulating quantum walk on the 1D chain and cycle, each with $15$ vertices.
Due to the initial state, the dynamics are qualitatively similar between the chain and lattice.
Experimentally, we observe the expected propagation distance increases linearly with time, closely matching the numerically predicted result.

\begin{figure}[ht!]
    \centering
    \includegraphics{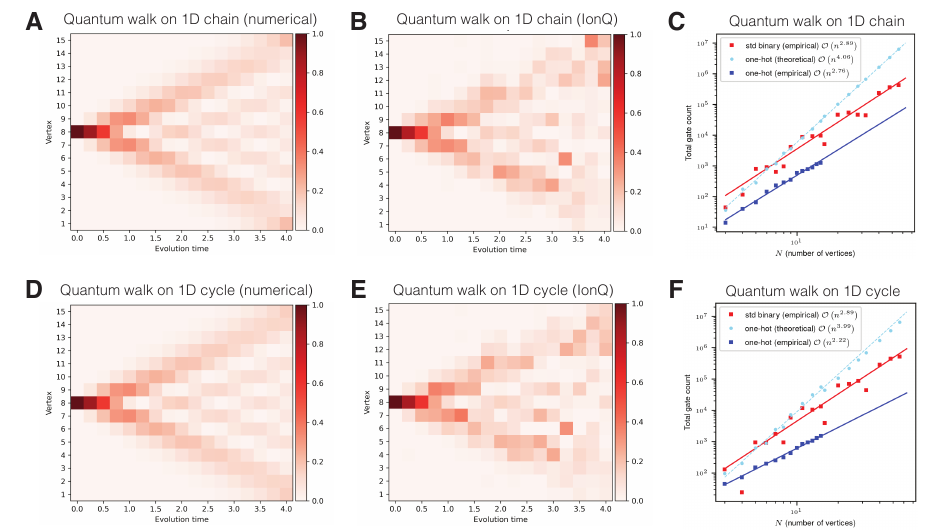}
    \caption{\small Quantum walk on 1-dimensional regular lattice graph (chain) and periodic lattice graph (cycle) of order $N=15$.}
    \label{fig:experiments_lattice}
\end{figure}

\subsubsection{Resource analysis}\label{append:resource_lattice}
We conduct a resource comparison between the standard binary code and Hamiltonian embedding to perform quantum walk on one-dimensional lattice graphs.
We assume the initial state is the first vertex $v_1$ on the chain or cycle and define the task as follows.

We define the expected propagation distance observable as
\begin{equation}
    O_{\mathrm{prop}}=\sum_{j=1}^{N} d(v_1,v_j)\ket{v_j}\bra{v_j},
\end{equation}
where $d(v_1,v_j)$ is the distance between $v_1$ and $v_j$ in the graph, and $\ket{v_j}$ denotes the encoded state for vertex $v_j$.
The task is then to evolve the system for time $T_N$ such that the propagation distance is at least $N/4$, i.e.
\begin{equation}
    T_N = \min_{t>0} \left\{t \mid \bra{\psi(t)}O_{\mathrm{prop}}\ket{\psi(t)} > N / 4\right\},
\end{equation}
where $\ket{\psi(t)}$ is the wave function at time $t$ as defined in \eqn{quantum-walk-equation}.
In panels C and F of \fig{experiments_lattice}, we show estimates of the total gate count required to simulate quantum walk on regular and periodic lattice graphs of order $N$ for evolution time $T_N$ using standard binary and the penalty-free one-hot embedding.
As in the experiments, we use the randomized first-order Trotter formula to simulate the Hamiltonian.
In both cases, the empirical data suggests that the penalty-free one-hot embedding has better constant factors as well as slightly better asymptotic scaling.

\subsection{Binary tree and glued trees graph}
\label{append:binary-glued}
A binary tree is a graph in which each node has at most two children. A glued trees graph consists of two copies of binary trees in which leaf nodes are ``glued'' together randomly to form a cycle.
Binary tree and the glued trees graph are important structures in computer science.
In particular, the glued trees graph is of significant theoretical interest because the problem of traversing the graph has been shown to exhibit an exponential separation between classical and quantum algorithms \cite{childs2003exponential}.
In this section, we discuss the embedding of perfect binary trees and glued trees graphs.

\begin{figure}[ht!]
    \centering
    \includegraphics[width=10cm]{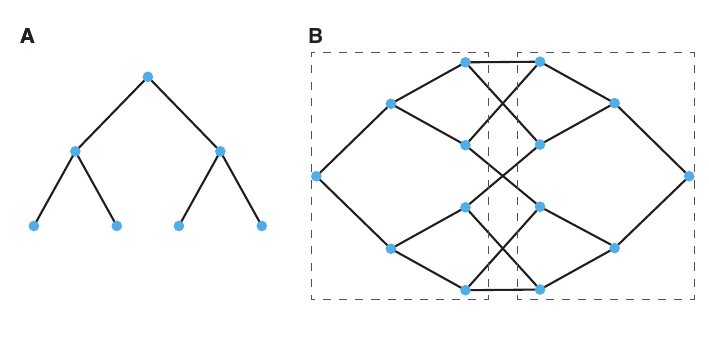}
    \caption{\small \textbf{A.} A height-2 perfect binary tree. \textbf{B.} A glued trees graph consisting of two height-2 binary trees.}
    \label{fig:append_1}
\end{figure}

\subsubsection{Binary tree}
A perfect height-$h$ binary tree has $2^{h+1} - 1$ nodes. In \fig{append_1}A, we illustrate a binary tree of height $2$. The adjacency matrix of this tree is of dimension $7$,
\begin{align}
    A_{\mathrm{tree}} = \begin{bmatrix}
         & 1 & 1 & & & & \\
        1 & & & 1 & 1 & & \\
        1 & & & & & 1 & 1\\
          & 1 & \\
          &  1 & \\
         & & 1 \\
         & & 1
    \end{bmatrix}
\end{align}

The adjacency matrix $A_{\mathrm{tree}}$ does not have a band or circulant structure. We use the one-hot embedding (see \append{one-hot}) to construct an embedding of $A_{\mathrm{tree}}$. To simplify the notation, we define the set of edges of the binary tree $E_{\mathrm{tree}} = \{(1,2),(1,3),(2,4),(2,5),(3,6),(3,7)\}$.

\begin{enumerate}
    \item One-hot embedding (with penalty) of $A_{\mathrm{tree}}$: $\Hpen = \sum^7_{j=1} Z_j, \quad Q = \sum_{(j,k)\in E} X_jX_k$.
    \item Penalty-free one-hot embedding of $A_{\mathrm{tree}}$: $\Hebd = \frac{1}{2}\left(\sum_{(j,k)\in E} X_jX_k + Y_jY_k\right)$.
\end{enumerate}

In \tab{btree_resource}, we summarize the resources required in different Hamiltonian embeddings of $A_{\mathrm{tree}}$.

\begin{table}[!ht]
    \centering
    \begin{tabular}{|p{4cm}|p{3cm}|p{4.5cm}|p{2.5cm}|}
     \hline
     \textbf{Embedding scheme} & \textbf{\# of qubits} & \textbf{\# of 2-local operators} & \textbf{Max. weight}\\
     \hline
     One-hot w/ penalty & $2^{h+1}-1$ & $2^{h+1}-2$ & 2\\
     \hline
     Penalty-free one-hot & $2^{h+1}-1$ & $2^{h+2}-4$ & 2\\
     \hline
    \end{tabular}
    \caption{Resource count for one-hot embeddings of height-$h$ binary trees.}
    \label{tab:btree_resource}
\end{table}

\subsubsection{Glued trees graph}
\label{append:gluedtrees}
In \fig{append_1}B, we show a 6-level glued trees graph (consisting of 2 height-2 binary trees). The adjacency matrix of the glued trees graph is $A_{\mathrm{glued}} = \begin{bmatrix}
        A_{\mathrm{tree}} & W\\
        W^\dagger & A_{\mathrm{tree}}
    \end{bmatrix}$, where $W$ is a $7$-by-$7$ matrix representing the connectivity pattern between the two trees,
$$W = \begin{bmatrix}
    \\
    \\
    \\
    1 & 1 &  & & & &\\
    1 &  & 1 & & & &\\
     & 1 &  & 1& & &\\
     &  & 1 & 1& & &\\
\end{bmatrix}$$

Since the full adjacency matrix $A_{\mathrm{glued}}$ is a $14$-by-$14$ sparse matrix, we can construct Hamiltonian embeddings of $A_{\mathrm{glued}}$ using one-hot embedding. In the standard one-hot embedding (with penalty), we need to use $14$ qubits and 20 two-qubit operators (all $XX$ rotations). In the penalty-free one-hot embedding, we need $14$ qubits and $40$ two-qubit operators (20 $XX$ plus 20 $YY$ rotations).

\subsubsection{Experiments on IonQ}
We experimentally demonstrate quantum walk on the binary tree and glued trees graph using the IonQ Aria-1 processor.
In both cases, we use the penalty-free one-hot embedding and we choose to embed the adjacency matrix rather than the Laplacian.
This choice of Hamiltonian avoids the need for additional Pauli-$Z$ rotations required to embed the diagonal terms of the Laplacian.
In both cases, we demonstrate traversal through the graphs by choosing the initial state to be the root node (in the case of the binary tree) or the entrance node (in the case of the glued trees graph).

The binary tree is chosen to be a perfect binary tree with height 3, i.e. $N=15$ nodes.
The maximum evolution time is $T=3.0$ with $r=6$ Trotter steps.
Each circuit consists of 168 two-qubit gates and a single one-qubit gate for preparing the root node state.
In \fig{experiments_binary_tree}A, we plot the propagation of the wavefunction across the levels of the binary tree, starting from the root (denoted as level 1).

For quantum walk on the glued trees graph, we use two binary trees each of height 2, resulting in a 14-node glued trees graph as shown in \fig{append_1}.
The maximum evolution time is taken to be $T=2.0$ with $r=4$ Trotter steps.
Each circuit consists of 120 two-qubit gates and a single one-qubit gate for preparing the entrance node state.
In both cases, we use the randomized first-order Trotter-Suzuki formula to simulate the Hamiltonian.

\begin{figure}[ht!]
    \includegraphics{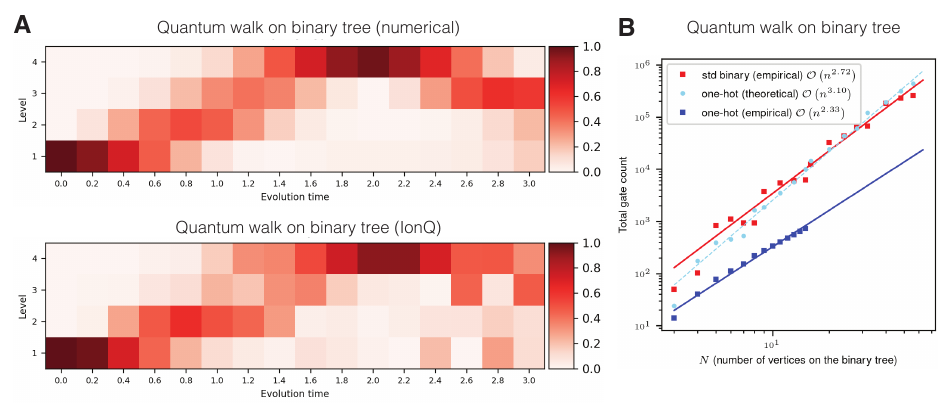}
    \centering
    \caption{\small Quantum walk on a binary tree of height 3. \textbf{A.} Heatmap of the quantum walk evolution (clustered on each layer of the binary tree). \textbf{B.} Estimated total gate count for quantum walk on binary trees with $N$ nodes. For standard binary, we estimate the gate count for all integers $N$ ranging from 3 to 63, but only a selection of points are plotted for visualization.}
    \label{fig:experiments_binary_tree}
\end{figure}

\subsubsection{Resource analysis}
\label{append:gluedtrees-resource}
In this section, we perform a resource analysis for the task of traversing the binary tree and glued trees graph. 

For the binary tree, we are interested in the propagation $O_{\mathrm{prop}}$ defined in exactly the same manner as in \append{resource_lattice}.
Here, we set the target distance to be equal to half of the height of the binary tree, so that
\begin{equation}
    T_N = \min_{t>0} \left\{t \mid \bra{\psi(t)}O_{\mathrm{prop}}\ket{\psi(t)} > \frac{\lfloor \lg N \rfloor}{2}\right\}.
\end{equation}

For the glued tree, we are primarily interested in the traversal from the entrance node to the exit node, which we define as follows.
Let $G_h$ be a glued trees graph with $2(h + 1)$ layers ($h$ being the height of each of the binary trees), and let $A_h$ be the adjacency matrix.
As defined in \cite{childs2003exponential}, let $\ket{\col j}$ be the uniform superposition over vertices in the $j$-th layer, i.e.
\begin{equation}
    \ket{\col j} = \frac{1}{\sqrt{N_j}} \sum_{a \in \text{layer $j$}} \ket{a},
\end{equation}
where
\begin{equation}
    N_j = 
    \begin{cases}
        2^j & 0 \leq j \leq h,\\
        2^{2h+1-j} & h+1 \leq j \leq 2h+1.
    \end{cases}
\end{equation}
for $j=0,1,\dots,2h+1$.
For a constant $s\in (0,1)$, let $T_{h,s}$ be the evolution time needed for the probability on the exit node to be at least $s$, i.e.,
\begin{equation}
    T_{h,s} = \min_{t>0}\left\{t : \left|\bra{\col 0} e^{-i t A_h}\ket{\col 2h + 1}\right|^2 \geq s\right\},
\end{equation}
where the initial state is the entrance node $\ket{\col 0}$.

In \fig{experiments_binary_tree}B, we show estimated total gate counts for quantum walk on the binary tree and glued trees graphs.
The target evolution time is set to $T_{h,s}$ with $s=0.4$.
As in the experimental demonstrations, we use the randomized first-order product formula to simulate the Hamiltonian.
Empirical estimates suggest that one-hot embedding requires significantly fewer gates than the standard binary code, with a slightly better asymptotic scaling.

\section{Spatial search on regular lattices}\label{append:spatial_search}

The spatial search problem is to find a marked item in a structured database. Precisely, given an unweighted graph $G = (\V,\E)$ without self loops, and an oracle function $f\colon \V \to \{0, 1\}$ that returns $1$ if the input vertex $x$ is \textit{marked} and $0$ otherwise. For simplicity, we assume there is only one marked vertex $\vv$, and the oracle function $f$ is given by
\begin{align}
    f(x) = \begin{cases}
        0 & (x \neq \vv)\\
        1 & (x = \vv).
    \end{cases}
\end{align}
The spatial search problem requires finding the marked vertex $\vv$ by querying the oracle $f$.

While no classical algorithms can do better than exhaustive search, quantum algorithms can provide speedups for this problem \cite{aaronson2003quantum,childs2004spatial,tulsi2008faster}. Childs and Goldstone \cite{childs2004spatial} proposed a quantum search algorithm by continuous-time quantum walk, which demonstrates quantum speedups for spatial search on various graphs. This quantum algorithm is formulated as a continuous-time quantum walk:

\begin{align}\label{eqn:spatial-search-original}
    H_{\mathrm{search}} = -\gamma L - H_\vv,
\end{align}
where $L$ is the graph Laplacian of $G$ (for details, see \append{qwalk}), and the \textit{oracle Hamiltonian} is defined as
\begin{align}
    H_\vv = \ket{\vv}\bra{\vv}.
\end{align}
The parameter $\gamma > 0$ is chosen to minimize the spectral gap of $H$. 

To simulate the Hamiltonian $H_{\mathrm{search}}$ on quantum computers, we need to build a Hamiltonian embedding for it. In this section, we consider the spatial search problem on a $d$-dimensional \emph{regular lattice graph}, whose graph Laplacian $L_\mathrm{reg}$ in can be embedded using a Hamiltonian embedding scheme proposed in \append{reg_lattice}. In what follows, we focus on the embedding of the oracle Hamiltonian $H_\vv$.

\subsection{Embedding of the oracle Hamiltonian}\label{append:spatial_search_ebd}

We exemplify the embedding of the oracle Hamiltonian $H_\vv$ using the unary embedding scheme. The same idea can be readily generalized to antiferromagnetic embedding and one-hot embedding.

First, we discuss how to embed the oracle Hamiltonian for a 1D chain graph with $N$ vertices $\{v_1,\dots,v_N\}$:
\begin{enumerate}
    \item If the marked vertex is $v_1$, the unary codeword for $v_1$ is the $(N-1)$-bit all-zero string $\ket{0\dots0}$. This is the only unary codeword with the first bit $0$, so we can construct an embedding of $H_{v_1}$ by marking the first qubit: $\Hebd_{v_1} = I - \hat{n}_1$.
    \item Similarly, if the marked vertex is $v_N$, the unary codeword for $v_N$ is the $(N-1)$-bit all-one string $\ket{1\dots1}$. This is the only unary codeword with the last bit $1$, so we can construct an embedding of $H_{v_1}$ by marking the last qubit: $\Hebd_{v_N} = \hat{n}_{N-1}$.
    \item If the marked vertex is among $\{v_2,\dots,v_{N-1}\}$, i.e., in the interior of the chain, we can distinguish this vertex from others by its $(N-1)$-th and $N$-th bits. In particular, $v_j$ is the only unary codeword whose $(N-1)$-th bit is $1$ and $N$-th bit is $0$. Therefore, we can construct an embedding of $H_{v_j}$ by marking its $(N-1)$-th and $N$-th bits: $\Hebd_{v_j} = (I - \hat{n}_j)\otimes \hat{n}_{j-1}$.
\end{enumerate}

In summary, the oracle Hamiltonian for the spatial search on the 1D chain is at most $2$-local. Based on the embedding of the oracle Hamiltonian for 1D chain, we can construct embeddings of the oracle Hamiltonian for any $d$-dimensional regular lattice:
\begin{align}
    \Hebd_{\xx} = \Hebd_{x_1}\otimes \Hebd_{x_2}\otimes \dots \otimes \Hebd_{x_d},
\end{align}
where the vertex $\xx = (x_1,\dots,x_d)$, each $x_k \in \{v_1,\dots, v_N\}$, and we define
\begin{align}
    \Hebd_{x_k} = \begin{cases}
        I - \hat{n}_1 & (x_k = v_1)\\
        \hat{n}_{N-1} & (x_k = v_N)\\
        (I - \hat{n}_j)\otimes \hat{n}_{j-1} & (x_k = v_j)
    \end{cases}
\end{align}

It turns out that the embedding of the oracle Hamiltonian $H_\xx$ is at least $d$-local and at most $2d$-local.

\begin{remark}
    Similar ideas can be used to mark vertices on the edge and interior of a 2D lattice graph. However, we need a 3- or 4-tensor operator to realize the oracle Hamiltonian because we need 3 qubits to unambiguously distinguish an edge vertex, and 4 qubits for an interior vertex. For example, on a $4$-by-$4$ lattice, the operator that marks the vertex $\vv = (2,4)$ is $H' = \hat{n}_1\hat{n}_2\hat{n}_6$, and the operator that marks the vertex $\vv = (3,3)$ is $H' = (I-\hat{n}_2)(I-\hat{n}_3)(I-\hat{n}_5)(I-\hat{n}_6)$.
\end{remark}

\paragraph{Generalization to periodic lattices.}
Similarly, a marked vertex on the periodic lattice can be perfectly distinguished by a few bits. Consider the codewords for the cycle graph $\texttt{CW-cycle}(n)$ with $n = 3$, we have:
\begin{table}[ht!]
    \centering
    \begin{tabular}{ | m{2cm} | m{1.5cm}| m{1.5cm} | m{1.5cm}| m{1.5cm}| m{1.5cm}| m{1.5cm}| } 
      \hline
      Node & 1 & 2 & 3 & 4 & 5 & 6\\ 
      \hline
      Code word & $ g_1 = {010}$ & $g_2 = 110$ & $g_3 = {100}$ & $g_4 = {101}$ & $g_5 = 001$ & $g_6 = 011$ \\ 
      \hline
    \end{tabular}
    \caption{Code words for the $4$-node cycle graph.}
\end{table}
There are no ``endpoints'' on the cycle graph, and each vertex codeword is distinguished from the others by exactly 2 bits. For example, $g_1$ is the only codeword with the first and last bit being $0$, and $g_2$ is the only codeword with the first and second bit being $1$, etc. This implies that we need a 4-tensor to encode the oracle Hamiltonian $H_\vv$ for any vertex $\vv$ on the 2D periodic lattice. Again, such operators can not be engineered on neutral-atom quantum computers.

\subsection{Experiments on IonQ}\label{append:spatial_search_experiments}
\subsubsection{Experiment setup}
We perform experimental demonstrations of spatial search by continuous time quantum walk using Hamiltonian embedding.
We present two spatial search experiments: the first being on a $4\times 4$ regular lattice with the unary embedding, and the second being on a $5\times 5$ regular lattice with the penalty-free one-hot embedding.

The algorithm for spatial search by continuous time quantum walk requires the initial state to be in a uniform superposition.
For the standard binary encoding, this is trivially done using a layer of Hadamard gates.
However, initial state preparation becomes nontrivial when considering alternate encodings such as unary and one-hot.

We design state preparation circuits for both the unary and one-hot codes.
In \algo{unary_state_prep} and \algo{one_hot_state_prep}, we show the construction of state preparation circuits for the unary and one-hot encodings, assuming real-valued amplitudes.
The idea behind these circuits is to start from the first encoded state and sequentially rotate probability amplitude using two-qubit gates on consecutive qubits.
In the case of complex-valued amplitudes, one may introduce local phases by an additional sequence of diagonal gates, such as single $Z$ rotations or $ZZ$ rotations.

\begin{figure}[ht!]
    \centering
    \begin{minipage}{\hsize}
    \begin{algorithm}[H]
    \setlength{\lineskip}{5pt}
    \caption{1D state preparation with unary encoding}\label{algo:unary_state_prep}
        \KwIn{Real amplitudes $a_1,\dots,a_{n+1}$, input state $\ket{00\dots 0}$}
        \KwOut{Quantum state $\sum_{k=1}^{n+1} a_k \ket{w_k}$}
        Apply $R_y(2\arccos(a_1))$ to qubit 1.\\
        \For{$k=1,\dots,n$}{
            Let $\alpha_{k+1} = a_{k+1} \left(\sum_{j=k+1}^{n+1} |a_j|^2\right)^{-1/2}$.\\
            Apply $R_y\left(2\arccos(\alpha_{k+1})\right)$ to qubit $k+1$ controlled on qubit $k$.
        }
    \end{algorithm}
    \end{minipage}
\end{figure}

\begin{figure}[ht!]
    \centering
    \begin{minipage}{\hsize}
    \begin{algorithm}[H]
    \setlength{\lineskip}{5pt}
    \caption{1D state preparation with one-hot encoding, sequential version}\label{algo:one_hot_state_prep}
        \KwIn{Real amplitudes $a_1,\dots,a_{n}$, input state $\ket{10\dots 0}$}
        \KwOut{Quantum state $\sum_{k=1}^{n} a_k \ket{w_k}$}
        Apply $R_y(2\arccos(a_1))$ to qubit 2.\\
        Apply $\mathrm{CNOT}$ to qubit $1$ controlled on qubit $2$.\\
        \For{$k=1,\dots,n-1$}{
            Let $\alpha_{k+1} = a_{k+1} \left(\sum_{j=k+1}^{n} |a_j|^2\right)^{-1/2}$.\\
            Apply $R_y\left(2\arccos(\alpha_{k+1})\right)$ to qubit $k+1$ controlled on qubit $k$.\\
            Apply $\mathrm{CNOT}$ to qubit $k$ controlled on qubit $k+1$.
        }
    \end{algorithm}
    \end{minipage}
\end{figure}
We note that a version of the state preparation circuit for the one-hot code has appeared in \cite{mathur2021medical}, referred to as a ``diagonal unary loader'' in the context of quantum neural networks.
For the one-hot code specifically, the circuit can additionally be parallelized to achieve $O(\log n)$ circuit depth.
In practice, we implement the parallel version involving the partial swap gate by using two arbitrary-angle MS gates, resulting in fewer gates. 

For both the unary and one-hot embeddings, we mark the corner vertex $\xx=(N,1)$ when constructing the embedding of the oracle Hamiltonian.
This results in a single $\hat{n}_{N-1}\otimes (I-\hat{n}_{1})$ term for the unary embedding and a $\hat{n}_{N}\otimes \hat{n}_{1}$ term for the one-hot embedding.

For spatial search with the unary embedding, we set $N=4$, for a total of 6 qubits.
For the unary embedding, the embedding Hamiltonian is given by
\begin{equation}
    \Hebd_{\mathrm{unary}} = -\gamma\Hebd_{L} - \Hebd_{\xx} + g\Hpen,
\end{equation}
where $\Hebd_{L}$ and $\Hebd_{\xx}$ are embeddings of the graph Laplacian and oracle Hamiltonian, respectively, and $\Hpen$ is the penalty Hamiltonian for the unary embedding.
For the real-machine experiment, we set the penalty parameter to $g=2$.
The embedding Hamiltonian consists of two-qubit $ZZ$ terms, as well as additional single-qubit $X$ and $Z$ terms.

For the one-hot embedding, the embedding Hamiltonian is given by
\begin{equation}
    \Hebd_{\mathrm{one-hot}} = -\gamma\Hebd_{L} - \Hebd_{\xx},
\end{equation}
where $\Hebd_{L}$ and $\Hebd_{\xx}$ are embeddings of the graph Laplacian and oracle Hamiltonian.
No penalty term is needed due to the use of the penalty-free embedding.
Here, the embedding Hamiltonian consists of two-qubit $XX$, $YY$, and $ZZ$ terms as well as additional single-qubit $Z$ terms to encode the diagonal entries of the graph Laplacian.

In both cases, we simulate the Hamiltonian using the second-order Trotter-Suzuki formula. For spatial search on a $4\times 4$ lattice with the unary encoding, we perform simulation for unitless time up to $T \approx 4.894$ with $r=12$ Trotter steps.
The largest resulting circuit consists of 6 qubits, 116 one-qubit gates, and 112 two-qubit gates.
For the $5\times 5$ lattice with the penalty-free one-hot embedding, we perform simulation for unitless time $T \approx 6.569$ with $r=5$ Trotter steps.
Each circuit consists of 10 qubits, 22 one-qubit gates and 181 two-qubit gates.
The evolution times are chosen to achieve a desired threshold success probability and is described in more detail in \append{spatial-resource} as part of the resource analysis.

\subsubsection{Resource analysis}\label{append:spatial-resource}
We conduct a resource analysis comparing the standard binary encoding with the unary and penalty-free one-hot embeddings.
In all cases, we consider spatial search on a $d$-dimensional lattice with $N$ vertices along each axis.
Starting from a uniform initial state, we define the success probability to be $|\bra{\psi(t)}\ket{\vv}|^2$.
The spatial search task is to evolve the initial state according to the Hamiltonian $H_{\text{search}}$ for time
\begin{equation}
    T_p=\min_{t>0}\{t : |\bra{\psi(t)}\ket{\vv}|^2 \geq p\},
\end{equation}
such that the success probability is at least $p$.
Note that for dimension $d=2$, the success probability is bounded by $O((\log (N) / N)^2)$ \cite{childs2004spatial}, so we choose $p=4 (\log (N) / N)^2$.
For reference, $T_p \approx 4.89$ for $N=4$ and $T_p \approx 6.57$ for $N=5$.
The parameter $\gamma$ is numerically computed by minimizing the spectral gap using the Nelder-Mead method.

For the standard binary encoding, we start with the spatial search Hamiltonian for $2^{\lceil \lg N \rceil}$ vertices, where $N$ is the desired number of vertices along a single axis and $d$ is the dimension of the lattice. 
To properly encode the Laplacian of the $d$-dimensional lattice with $N^d$ vertices, we isolate the subgraph with $N^d$ vertices, removing all off-diagonal terms that contribute to interactions between the subgraph and the remaining component of the graph.
The diagonal terms of the target submatrix are modified to match the spatial search Hamiltonian.
This procedure helps preserve the tensor product structure of the target Hamiltonian.

For the unary embedding, we choose the penalty parameter $g=\gamma nd T_{p}$, where $\gamma$ is the parameter for the oracle Hamiltonian.
This choice of penalty parameter yields a unitary fidelity of over $0.99$ for all system sizes tested.
Here, we define the unitary fidelity restricted to the encoding subspace by
\begin{equation}
    f = \frac{1}{Nd}\left\|\tr \left[U_{\text{search}}^{\dagger} (P_{\S} U_{\text{ebd}} P_{\S})\right]\right\|,
\end{equation}
where $U_{\text{search}}$ and $U_{\text{ebd}}$ are the evolution operators for $H_{\text{search}}$ and $\Hebd$, respectively.

For both the unary and one-hot embedding, we apply the second-order Trotter-Suzuki product formula and use Monte Carlo sampling to evaluate the additive Trotter error.

In \fig{fig_3}B we show the estimated gate counts for performing spatial search on $N\times N$ lattices.
The marked vertex is $\xx = (N,1)$, so the embedded oracle Hamiltonian is $\hat{n}_{N-1}\otimes (I-\hat{n}_1)$.
We test system sizes up to $N=11$ for standard binary and up to $N=6$ for the unary and penalty-free one-hot embeddings.
While the unary embedding requires fewer gates empirically than standard binary for the system sizes tested, the asymptotic scaling is slightly worse.
On the other hand, the penalty-free one-hot embedding shows a significant advantage both in the empirically estimated gate counts and the asympotic scaling.

We additionally show resource estimates on 3D lattices in \fig{resource_spatial_search}.
In this case, the marked vertex is set to $\xx=(N,N,N)$, resulting in the embedded oracle Hamiltonian to be $\hat{n}_N \otimes \hat{n}_N \otimes \hat{n}_N$.
We observe advantages for the unary and penalty-free one-hot embeddings in both the constant factor and asymptotic scaling, according to the empirical estimates.
\begin{figure}[ht!]
    \centering
    \includegraphics[width=0.5\linewidth]{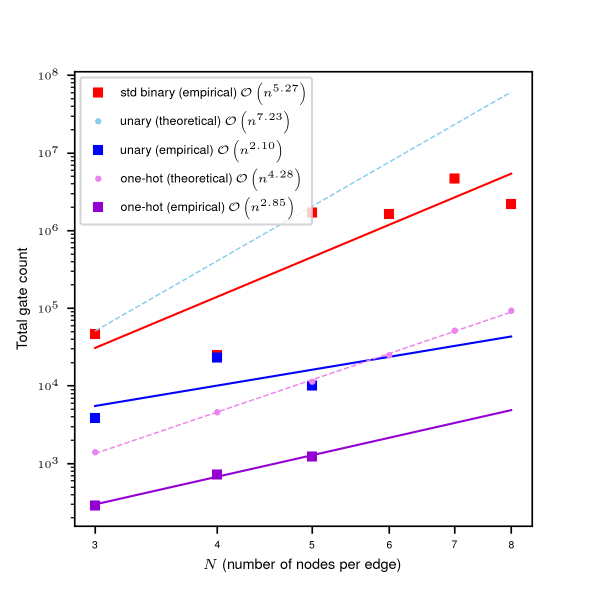}
    \caption{\small Estimated total gate count for spatial search on 3D lattice.}
    \label{fig:resource_spatial_search}
\end{figure}

\section{Simulating quantum dynamics in the real space}\label{append:real_space_simulation}
In this section, we consider the simulation of quantum dynamics in the real space. The quantum dynamics is described by the Schr\"odinger equation over $\R^d$:
\begin{align}\label{eqn:schrodinger}
    i \frac{\partial}{\partial t}\Psi = \left[-\frac{1}{2}\nabla^2 + f(x)\right] \Psi(t,x),
\end{align}
where $\Psi(t,x)\colon [0,T]\times \R^d \to \C$ is the quantum wave function. The Schr\"odinger equation is subject to an initial condition $\Psi(0,x) = \Psi_0(x)$.

\subsection{Experiments on IonQ}\label{append:real-space-ionq}
\subsubsection{Problem formulation}
\begin{lemma}
    Consider the 1D Schr\"odinger equation with quadratic potential field ($a > 0$):
    \begin{align}
        i \frac{\partial}{\partial t}\Psi(t,x) = \left[-\frac{1}{2}\frac{\partial^2}{\partial x^2} + \left(\frac{1}{2}ax^2 + bx\right)\right] \Psi(t,x),
    \end{align}
    subject to a Gaussian initial state
    \begin{align}
        \Psi(0,x) = \left(\frac{1}{2\pi \sigma^2}\right)^{1/4} e^{-\frac{x^2}{4\sigma^2}}.
    \end{align}
    Then, the expectation value of the position observable $\hat{x}$ is given by 
    \begin{align}\label{eqn:expect_X}
        \<\hat{x}\>_t = \frac{b}{a}\left(\cos(\sqrt{a}t) - 1\right).
    \end{align}
\end{lemma}
\begin{proof}
    Since the potential field is quadratic, we invoke the Ehrenfest theorem to obtain the following Hamilton's equation for $\<\hat{x}\>_t$ and $\<\hat{p}\>_t$:
    \begin{align}
        \frac{\d}{\d t}\<\hat{x}\>_t &= \<\hat{p}\>_t,\\
        \frac{\d}{\d t}\<\hat{p}\>_t &= - \<f\>_t = -a \<\hat{x}\>_t - b,
    \end{align}
    where the initial data are given by
    \begin{align}
        \<\hat{x}\>_0 = \<\hat{p}\>_0 = 0.
    \end{align}
    Therefore, the closed-form formula of $\<\hat{x}\>_t$ is obtained by solving the system of ODEs.
\end{proof}

\subsubsection{Fock space truncation method}\label{append:fock_space_truncation}
We can rewrite the real-space Schr\"odinger operator using the position operator $\hat{x}$ and the momentum operator $\hat{p} = -i\frac{\partial}{\partial x}$:
\begin{align}\label{eqn:bosonic}
    H = \frac{1}{2}\sum^d_{j=1}\hat{p}^2_j + f(\hat{x}_1,\dots,\hat{x}_d).
\end{align}
We define two operators:
\begin{align}
    \hat{a} = \frac{1}{\sqrt{2}}(i\hat{p}+\hat{x}),\quad \hat{a}^\dagger = \frac{1}{\sqrt{2}}(-i\hat{p}+\hat{x}).
\end{align}
The operator $\hat{a}^\dagger$ is the raising/creation operator, and $\hat{a}$ is the lowering/annihilation operator. They are infinite-dimensional quantum operators with eigenstates being the number basis (corresponding to the eigenstates of quantum harmonic oscillators):
\begin{align}
    \hat{a}\ket{0} &= 0,~\hat{a}\ket{n} = \sqrt{n} \ket{n-1} ~\text{(for $n \ge 1$)},\\
    \hat{a}^\dagger \ket{n} &= \sqrt{n+1}\ket{n+1} ~\text{(for $n \ge 0$)}.
\end{align}
We can express the creation and annihilation operators in the standard matrix form,
 \begin{align}
    \hat{a}^\dagger = \begin{bmatrix}
        0 & & & & \\
        1 & 0 & & & \\
        & \sqrt{2} & 0 & & \\
        & & \sqrt{3} & 0 & \\
        & & & ...& ...\end{bmatrix}, ~\hat{a} =\begin{bmatrix}
        0 & 1 & & & \\
        & 0 & \sqrt{2} & & \\
        & & 0 & \sqrt{3} & \\
        & & & 0 & 2 \\
        & & & ...& ...\end{bmatrix}.
\end{align}

Using the ladder operators, we can rewrite the position and momentum operators:
\begin{align}
    \hat{x} = \frac{1}{\sqrt{2}}\left(\hat{a}^\dagger + \hat{a}\right),~\hat{p} = \frac{i}{\sqrt{2}} \left(\hat{a}^\dagger - \hat{a}\right).
\end{align}
In the matrix form, the two operators are infinite-dimensional tri-diagonal matrices. 
\begin{align}
    \hat{x} &= \frac{1}{\sqrt{2}}\sum^{\infty}_{j=0} \sqrt{j+1} \left(\ket{j}\bra{j+1}+ \ket{j+1}\bra{j}\right),\\
    \hat{p} &= \frac{i}{\sqrt{2}}\sum^{\infty}_{j=0} \sqrt{j+1} \left(-\ket{j}\bra{j+1}+ \ket{j+1}\bra{j}\right).
\end{align}
Similarly, we can compute the quadratic operators:
\begin{align}
    \hat{p}^2 &= \frac{1}{2}\sum^{\infty}_{j=0}(2j+1) \ket{j}\bra{j} -\frac{1}{2}\sum^{\infty}_{j=0} \sqrt{(j+1)(j+2)} \left(\ket{j}\bra{j+2}+ \ket{j+2}\bra{j}\right),\\
    \hat{x}^2 &= \frac{1}{2}\sum^{\infty}_{j=0}(2j+1) \ket{j}\bra{j} +\frac{1}{2}\sum^{\infty}_{j=0} \sqrt{(j+1)(j+2)} \left(\ket{j}\bra{j+2}+ \ket{j+2}\bra{j}\right),
\end{align}
which are infinite-dimensional matrices with bandwidth $2$.

\vspace{4mm}
To simulate the real-space quantum operator~\eqn{bosonic} on a quantum computer, we need to truncate the Fock space (spanned by the number basis) up to finite dimension $N$. Then, we can construct Hamiltonian embeddings of the truncated quantum operator for real-machine implementation. 

By truncating Fock space up to $N$ levels, we have the truncated operators,
\begin{align}
    \hat{x} &= \frac{1}{\sqrt{2}}\sum^{N-2}_{j=0} \sqrt{j+1} \left(\ket{j}\bra{j+1}+ \ket{j+1}\bra{j}\right),\\
    \hat{p}^2 &= \frac{1}{2}\sum^{N-1}_{j=0}(2j+1) \ket{j}\bra{j} -\frac{1}{2}\sum^{N-3}_{j=0} \sqrt{(j+1)(j+2)} \left(\ket{j}\bra{j+2}+ \ket{j+2}\bra{j}\right),\\
    \hat{x}^2 &= \frac{1}{2}\sum^{N-1}_{j=0}(2j+1) \ket{j}\bra{j} +\frac{1}{2}\sum^{N-3}_{j=0} \sqrt{(j+1)(j+2)} \left(\ket{j}\bra{j+2}+ \ket{j+2}\bra{j}\right).
\end{align}

Therefore, if we use the unary embedding, the encoding Hamiltonians are $(N-1)$-qubit operators:
\begin{align}
    Q^{\mathrm{unary}}_{\hat{x}} &= \frac{1}{\sqrt{2}} \sum^{N-1}_{j=1} \sqrt{j}X_j,\\
    Q^{\mathrm{unary}}_{\hat{p}^2} &= \sum^{N-1}_{j=1} \hat{n}_j - \frac{1}{2}\sum^{N-2}_{j=1} \sqrt{j(j+1)}X_{j+1}X_j,\\
    Q^{\mathrm{unary}}_{\hat{x}^2} &= \sum^{N-1}_{j=1} \hat{n}_j + \frac{1}{2}\sum^{N-2}_{j=1} \sqrt{j(j+1)}X_{j+1}X_j.
\end{align}
It is worth noting that we already discard the global phase in the embedding operators.

Similarly, if we choose the one-hot embedding, the encoding Hamiltonians are $N$-qubit operators:
\begin{align}
    Q^{\mathrm{one-hot}}_{\hat{x}} &= \frac{1}{\sqrt{2}} \sum^{N-1}_{j=1} \sqrt{j}X_{j+1}X_j,\\
    Q^{\mathrm{one-hot}}_{\hat{p}^2} &= \sum^{N}_{j=1} (j-\frac{1}{2})\hat{n}_j - \frac{1}{2}\sum^{N-2}_{j=1} \sqrt{j(j+1)}\frac{\left(X_{j+2}X_j + Y_{j+2} Y_j\right)}{2},\\
    Q^{\mathrm{one-hot}}_{\hat{x}^2} &= \sum^{N}_{j=1} (j-\frac{1}{2})\hat{n}_j + \frac{1}{2}\sum^{N-2}_{j=1} \sqrt{j(j+1)}\frac{\left(X_{j+2}X_j + Y_{j+2} Y_j\right)}{2}.
\end{align}
Again, the global phase is discarded.

With these embedding operators, we can use the rules in \append{building} to construct an embedding of the real-space Schr\"odinger operator with a quadratic potential field.

To simulate the $N$-level truncated Hamiltonian, we require a large penalty coefficient $g$. With a large $g$, the norm of the embedding Hamiltonian becomes so huge that the standard product formula requires a significant number of Trotter steps that is infeasible for near-term devices. To address this issue, we transform the embedding operators to the interaction picture and use the quantum simulation algorithm \texttt{qDRIFT} (for details, see \append{sim_perturbative}).

\subsubsection{Experiment setup and result}\label{append:real-setup}
We perform experiments of the Schr\"odinger equation in one-dimensional real space using the IonQ Aria-1 processor.
We use the penalty-free one-hot embedding to simulate the real-space Schr\"odinger equation with quadratic potential $f(x)=x^2-\frac{1}{2}x$.
The embedding Hamiltonian is then given by
\begin{equation}
    H_{\mathrm{real-space}} = \frac{1}{2} Q_{\hat{p}^2}^{\mathrm{one-hot}} - \frac{1}{2}Q_{\hat{x}^2}^{\mathrm{one-hot}} + \frac{1}{2} Q_{\hat{x}}^{\mathrm{one-hot}},
\end{equation}
as described in \append{fock_space_truncation}.

We truncate the Hamiltonian to $N=5$ levels, corresponding to 5 qubits using the one-hot embedding.
The initial state is the first encoded state of the one-hot code, prepared using a single $X$ rotation gate.
We perform quantum simulation using the randomized first-order Trotter-Suzuki formula for evolution times up to $T=5$ with $r=11$ Trotter steps.
For each circuit, the total gate count is 154 two-qubit gates in addition to a single one-qubit gate for preparing the initial state.
To compute the expected position observable $\hat{x}$, we measure in $x$-basis using 1500 shots.
For computing the expected kinetic energy $\frac{1}{2}\hat{p}^2$, we also measure in the computational basis using 500 shots.

\begin{figure}[ht!]
    \centering
    \includegraphics[width=9cm]{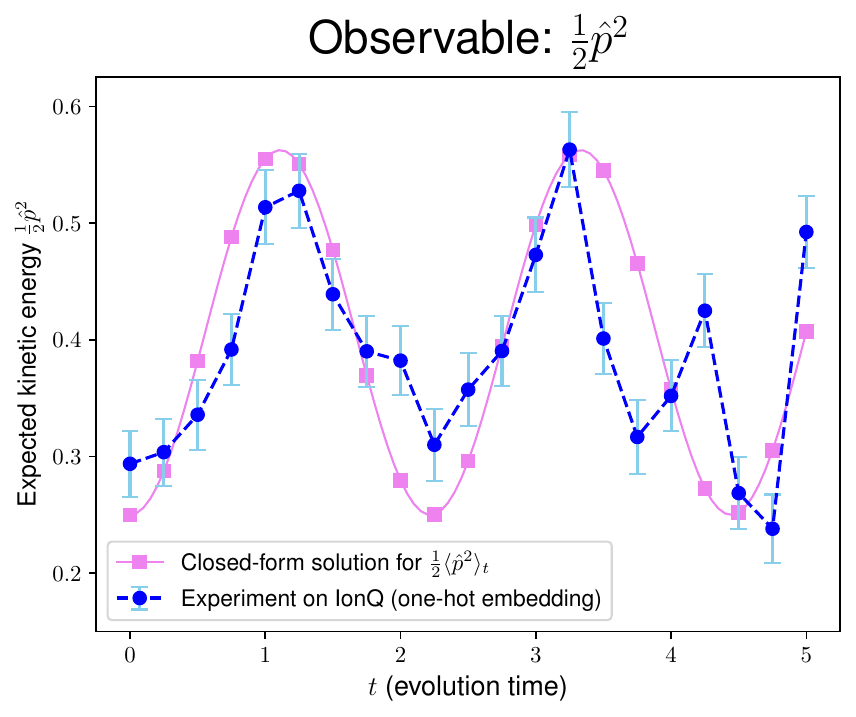}
    \caption{\small Expectation value of kinetic energy $\frac{1}{2}\hat{p}^2$ depicted as a function of evolution time $t$.
    The violet curve denotes the analytically computed expected kinetic energy with squares drawn at the same evolution times as the experiment data for comparison.
    Blue dots denote experimental results obtained from the IonQ Aria-1 processor using 1500 shots for $x$-basis measurements and 500 shots for $z$-basis measurements for each data point. Error bars denote standard error. }
    \label{fig:real_space_kinetic_energy}
\end{figure}

In \fig{fig_4}C, we plot the expected value of the position observable.
In violet, we plot the closed form solution, given by $\langle \hat{x} \rangle_t = \frac{1}{4}(1-\cos(\sqrt{2}t))$.
In blue, we plot the real-machine data obtained from the IonQ processor.
The experimentally observed values show a sinusoidal trend consistent with the analytically predicted values.

In addition, we present the expected kinetic energy $\frac{1}{2}\langle p^2 \rangle$ in \fig{real_space_kinetic_energy}.
The analytical solution, shown in violet, is evaluated as $\frac{1}{2}\langle p^2 \rangle = -\frac{5}{16}\cos(\sqrt{2}t)^2 + \frac{9}{16}$.
The experimental data is depicted by blue dots, showing a sinusoial trend consistent with the analytically predicted solution.

\subsubsection{Resource analysis}\label{append:real-resource}
We conduct a resource analysis for performing a simulation of the real-space Schr\"odinger equation via the truncated Fock space method.
The task is to perform quantum simulation for a fixed evolution time $T$.
To be consistent with the real-machine experiment, we fix the evolution time to be $T=5$ and the potential function is $f(x)=\frac{1}{2} ax^2 + bx$ with $a=2$ and $b=-1/2$.
For this task, we apply the randomized first-order Trotter-Suzuki product formula.

For the standard binary encoding, we start with the desired quantum operator with $2^{\lceil \lg N \rceil}$ levels, in order to preserve the structure of the problem.
To simulate the operator with $N$ levels, we only use the first $N$ basis states as the encoding subspace.
To this end, we remove all matrix entries that contribute to interactions between the encoding subspace and the orthogonal complement.
We use Qiskit to compute the Pauli decomposition of the resulting matrix and convert the circuit to single-qubit Pauli rotations and $XX$-rotation gates.

In \fig{fig_4}B, we plot estimates of the total gate count required for the standard binary code and the penalty-free one-hot embedding.
For standard binary, we observe an asymptotic gate complexity of $O(n^{3.16})$.
For the one-hot embedding, the complexity is estimated as $O(n^{2.67})$, showing a modest advantage in the asymptotic scaling.
In addition, we observe an order of magnitude advantage in the constant factor for the one-hot embedding.

\subsection{Experiments on QuEra}\label{append:quera_experiments}
\subsubsection{Finite difference method}\label{append:finite_difference}
We now discuss how to simulate two-dimensional Schr\"odinger equations on the QuEra machine. Consider the following partial differential equation defined on the unit square $\Omega = [0,1]^2$,
\begin{align}\label{eqn:2d_sch}
        i \frac{\partial}{\partial t}\Psi(t,\vect{x}) = \left(-\frac{1}{2}\nabla^2 + V(\vect{x})\right) \Psi(t,\vect{x}),
\end{align}
where $\Psi(t,\vect{x})\colon [0,t]\times \Omega \to \C$ is the quantum wave function with Dirichlet boundary conditions (i.e., $\Psi(t,\vect{x}) = 0$ for $\vect{x} \in \partial \Omega$). $\nabla^2 = \frac{\partial^2}{\partial x^2} + \frac{\partial^2}{\partial y^2}$ is the Laplacian operator, and $V(\vect{x})\colon \Omega \to \R$ is the potential field.

In our QuEra experiment, we simulate the Schr\"odinger equation \eqn{2d_sch} using the finite difference method. Suppose that we discretize the unit interval $[0,1]$ into $N$ pieces with grid points $\{z_k = k/(N-1) \colon k = 0,\dots, N-1\}$. For any function $u(x)$ defined on $[0,1]$ with vanishing boundary condition $u(0) = u(1) = 0$, we can discretize $u$ on the mesh $\{z_k\}$ and represent it as an $N$-dimensional vector:
$$\vec{u} = [u(z_0),\dots,u(z_{N-1})]\trans.$$
Using the central finite difference approximation, i.e., $u''(x) \approx \frac{u(x+h) - 2u(x)+u(x-h)}{h^2}$, we can approximate the function $u''$ by a new $N$-dimensional vector $D \vec{u}$ (with $h = 1/(N-1)$), where the matrix $D$ is given by
\begin{align}\label{eqn:discrete_D}
    D = \frac{1}{h^2}\begin{bmatrix} -2 & 1 & & \\
    1 & -2 & 1 & \\
    ...& ... & ... &...\\
    & 1 & -2 & 1\\
    & & 1 & -2\\
\end{bmatrix}.
\end{align}

The matrix $D$ is often called the finite difference discretization of the differential operator $\frac{\partial^2}{\partial x^2}$. Note that the matrix $D$ is tri-diagonal. We can use either unary embedding or antiferromagnetic embedding for the matrix $D$ (see \append{band}).

Similarly, the finite difference method can be applied to the 2D Schr\"odinger equation \eqn{2d_sch}. The discretized quantum Hamiltonian reads,
\begin{align}\label{eqn:fdm_operator}
    \hat{H} = -\frac{1}{2} (D \otimes I) + I \otimes D) + U,
\end{align}
where $I$ is the $N$-by-$N$ identity matrix and $U$ is a $N^2$-by-$N^2$ diagonal matrix such that 
$$U\ket{z_j}\ket{z_k} = V(z_j,z_k)\ket{z_j}\ket{z_k}.$$

\subsubsection{Neutral-atom simulator and antiferromagnetic embedding}\label{append:quera-antiferro}
The QuEra quantum simulator is based on neutral atoms. The atoms are placed individually and deterministically on a two-dimensional plane by optical tweezers \cite{labuhn2016tunable,lienhard2018observing}. On the QuEra quantum simulator, a qubit is realized by an atom with an internal ground state $\ket{0}$ and an excited Rydberg state $\ket{1}$. The Rydberg atoms are long-lived and can be coherently controlled by external laser pulses. The evolution of the neutral atoms is described by the Schr\"odinger equation $i \hbar \frac{\d}{\d t} \ket{\psi(t)} = \hat{H}(t)\ket{\psi(t)}$, where the system Hamiltonian is given by
\begin{align}\label{eqn:rydberg-ham}
    \frac{\hat{H}(t)}{\hbar} = \sum_j\left( \frac{\Omega_j(t)}{2} \hat{X}_j - \Delta_j(t) \hat{n}_j \right) + \sum_{j < k} \frac{C_6}{|\mathbf{r}_j - \mathbf{r}_k|^6}\hat{n}_j \hat{n}_k.
\end{align}
In the Hamiltonian, $\hat{X}_j = \begin{bmatrix}0&1\\1&0\end{bmatrix}$ and $\hat{n}_j=\begin{bmatrix}0&0\\0&1\end{bmatrix}$ are the Pauli-$X$ operator and the number operator acting on the $j$-th qubit, respectively. $\Omega_j(t)$, $\Delta_j(t)$ denotes the Rabi frequency and local detuning of the driving laser field on qubit $j$. $C_6$ is the Rydberg interaction constant that depends on the particular Rydberg atom used. $\mathbf{r}_j$ denotes the position vector of the $j$-th qubit. In our experiment, we use the neutral-atom analog quantum computer fabricated by QuEra Computing Inc. In each simulation, we can initiate no more than 256 atoms, and the Rydberg interaction constant is $C_6 = 862690\times 2\pi~\SI{}{\mega\hertz}\cdot \SI{}{\micro\meter}^6$.

In the system Hamiltonian \eqn{rydberg-ham}, we have the atom-atom interactions $\hat{n}_j\hat{n}_k$ with positive coefficients. This structure is very suitable for antiferromagnetic embedding. In~\lem{penalty-chain}, we show how to engineer a penalty Hamiltonian for antiferromagnetic embedding by tuning the local detuning of neutral atoms.

\begin{lemma}\label{lem:penalty-chain}
    Given an integer $N > 1$, let $n = N-1$.
    Define the penalty Hamiltonian:
    \begin{align}\label{eqn:chain-penalty}
        \hat{H}_{\mathrm{pen}} = \sum_{1 \le j< k\le n} \frac{C_6}{|\mathbf{r}_j - \mathbf{r}_k|^6} \hat{n}_j \hat{n}_k - \sum^n_{j=1} \Delta_j \hat{n}_j,
    \end{align}
    where the atoms are placed on a straight line with equal distance $r>0$, i.e., $\mathbf{r}_j = (0, (j-1)r)$, and the local detuning values are:
    \begin{align}\label{eqn:detuning-chain}
        \Delta_j = \begin{cases}
        \frac{C_6}{r^6} \left(\sum^{\lfloor (n+1-j)/2\rfloor}_{m=1}\frac{1}{(2m-1)^6} +\sum^{\lfloor j/2\rfloor}\frac{1}{(2m)^6} \right), & (\text{for odd~}j)\\
        \frac{C_6}{r^6}\left(\sum^{\lfloor (n-j)/2\rfloor}_{m=1}\frac{1}{(2m-1)^6} +\sum^{\lfloor (j-1)/2\rfloor}\frac{1}{(2m)^6} \right). & (\text{for even}~j)
        \end{cases}
    \end{align}
    Then, the penalty Hamiltonian \eqn{chain-penalty} has an $N$-fold degenerate ground-energy subspace $\S$ which is spanned by all the $n$-bit antiferromagnetic code (see \tab{antiferro_code}).
\end{lemma}
\begin{proof}
    We use $\ket{w_k}$ to denote the antiferromagnetic code for $k = 1,\dots,N$. Define the following two operators
    $$H'_Z = \sum_{1\le j< k\le n} \frac{C_6}{|\mathbf{r}_j - \mathbf{r}_k|^6} \hat{n}_j \hat{n}_k,\quad H''_Z = -\sum^n_{j=1} \Delta_j \hat{n}_j,$$
    where $\mathbf{r}_j = \left(0, (j-1)r\right)$, and we consider an ansatz
    $$\hat{H}_{\mathrm{pen}} = H'_Z + H''_Z.$$
    Since the Hamiltonian $H'_Z$ is in antiferromagnetic ordering, its low-energy subspace is spanned by antiferromagnetic code. However, these low-energy states are not degenerate. To get the $N$-degenerate ground-energy subspace, we solve a linear system to compute $\Delta_j$ in $\hat{H}_{\mathrm{pen}}$.

    First, we try to match the energy of $\ket{w_1}$ and $\ket{w_2}$. This requires the following identity holds:
    \begin{align}\label{eqn:w1-w2-equal-E}
        \bra{w_1}\hat{H}_{\mathrm{pen}}\ket{w_1} = \bra{w_2}\hat{H}_{\mathrm{pen}}\ket{w_2}.
    \end{align}
    Note that $w_1 = 0101\dots$ and $w_2 = 1101\dots$, and the 2-tensor $\hat{n}_j\hat{n}_k = \ket{11}_{jk}\bra{11}_{jk}$ is visible only when the atoms $j$ and $k$ are simultaneously excited, so we have that,
    $$\bra{w_2}H'_Z\ket{w_2} - \bra{w_1}H'_Z\ket{w_1} = \sum^{\lfloor n/2 \rfloor}_{m=1} \frac{C_6}{r^6 (2m-1)^6}.$$
    
    Meanwhile, we have $\bra{w_2}H''_Z\ket{w_2} - \bra{w_1}H''_Z\ket{w_1} = -\Delta_1$. To have the identity \eqn{w1-w2-equal-E}, we must have
    $$\Delta_1 = \frac{C_6}{r^6} \sum^{\lfloor n/2 \rfloor}_{m=1} \frac{1}{(2m-1)^6}.$$

    Similarly, we can compute $\Delta_2$. Comparing $w_2 = 1101\dots$ and $w_3 = 1001\dots$, we have
    $$\bra{w_2}H^{(1)}_Z\ket{w_2} - \bra{w_3}H^{(1)}_Z\ket{w_3} = \frac{C_6}{r^6} + \sum^{\lfloor n/2 \rfloor-1}_{m=1} \frac{C_6}{r^6 (2m)^6}.$$
    Also, we have $\bra{w_2}H''_Z\ket{w_2} - \bra{w_3}H''_Z\ket{w_3} = -\Delta_2$, it follows that
    $$\Delta_2 = \frac{C_6}{r^6}\left(\frac{1}{1^6} + \sum^{\lfloor n/2 \rfloor-1}_{m=1} \frac{1}{ (2m)^6}\right).$$

    We can compute all local detuning $\Delta_j$ following the same procedure. It turns out that if we choose $\Delta_j$ as in \eqn{detuning-chain}, the ground-energy subspace of $\hat{H}_{\mathrm{pen}}$ is indeed spanned by the codeword states $\ket{w_j}$ for $j = 1,\dots, N$.
\end{proof}

\subsubsection{Experiment setup and result}
We perform experiments of simulating the 2D Schr\"odinger equation using the QuEra Aquila processor. The programmable quantum Hamiltonian on QuEra Aquila reads~\cite{quera_doc},
\begin{align}\label{eqn:quera_limited}
    \frac{\hat{H}(t)}{\hbar} = \frac{\Omega(t)}{2} \left(\sum_j e^{i\phi(t)}\ket{0}\bra{1}_j + e^{-i\phi(t)}\ket{1}\bra{0}_j\right) - \Delta(t) \left(\sum_j \hat{n}_j\right) + \sum_{j < k} \frac{C_6}{|\mathbf{r}_j - \mathbf{r}_k|^6}\hat{n}_j \hat{n}_k,
\end{align}
where $\Omega(t)$, $\phi(t)$, and $\Delta(t)$ denote the Rabi frequency, laser phase, and the detuning of the driving laser field on atom (qubit) $j$ coupling the two states $\ket{0}$ (ground state) and $\ket{1}$ (Rydberg state). It is worth noting that all programmable control parameters are \emph{global}.

\paragraph{Hamiltonian embedding and the engineering of the potential field.}
We use the antiferromagnetic embedding the simulate the quantum dynamics. Ideally, we want to use \lem{penalty-chain} to engineer a perfect antiferromagnetic embedding of the discretized Laplacian operator $-\frac{1}{2}\left(D\otimes I + I\otimes D\right)$ (see \eqn{fdm_operator}). However, currently (as of the fall of 2023), the QuEra Aquila processor does not support user-specified local detuning (see \eqn{quera_limited}). The lack of local detuning poses a significant restriction on the Hamiltonian we can engineer on the QuEra Aquila processor. 

In our QuEra experiment, we follow the intuition of \lem{penalty-chain} and arrange $12$ neutral atoms as in \fig{fig_4}D. The $12$ neutral atoms are grouped into 2 vertical chains. Each chain has 6 atoms. The (sub-)system Hamiltonian for each chain reads
\begin{align}\label{eqn:quera_real}
    H_0 = \sum_{1\le j < k \le 6} \frac{C_6}{(k-j)^6r^6}\hat{n}_j\hat{n}_k - \Delta \sum^6_{j=1}\hat{n}_j,
\end{align}
where $r = \SI{5.9}{\micro\meter}$ is the distance between two adjacent atoms, and $\Delta$ is a global detuning.

It is clear that the actual Hamiltonian \eqn{quera_real} is similar to the penalty Hamiltonian specified in \eqn{chain-penalty}, except that we can not specify individual detuning $\Delta_j$ for each atom. In this case, the codeword subspace of antiferromagnetic embedding
$$\S = \{\ket{010101},\ket{010100},\ket{010110},\ket{010010},\ket{011010},\ket{001010},\ket{101010}\}$$
coincide with the low-energy subspace of $H_0$ but not exactly the ground-energy subspace. This non-degeneracy can be realized as an \emph{effective} potential field. Namely, we may regard the operator \eqn{chain-penalty} as the \emph{desired} penalty Hamiltonian (as in \eqn{chain-penalty}) plus an operator that embeds \emph{a} potential field (although this potential field is not very smooth and it does not correspond to a familiar closed-form function). Similarly, the system Hamiltonian of the two atom chains (\fig{fig_4}D) can be regarded as the penalty Hamiltonian for 2D antiferromagnetic embedding plus an effective potential operator. In \fig{fig_10}A, we plot the effective 2D potential function at $\Delta = 5\times 10^7$ rad/s. The potential is shown as a discretized function on the 2D mesh $\{(z_j,z_k):j,k = 1,\dots, N\}$.

\begin{figure}[ht!]
    \centering
    \includegraphics{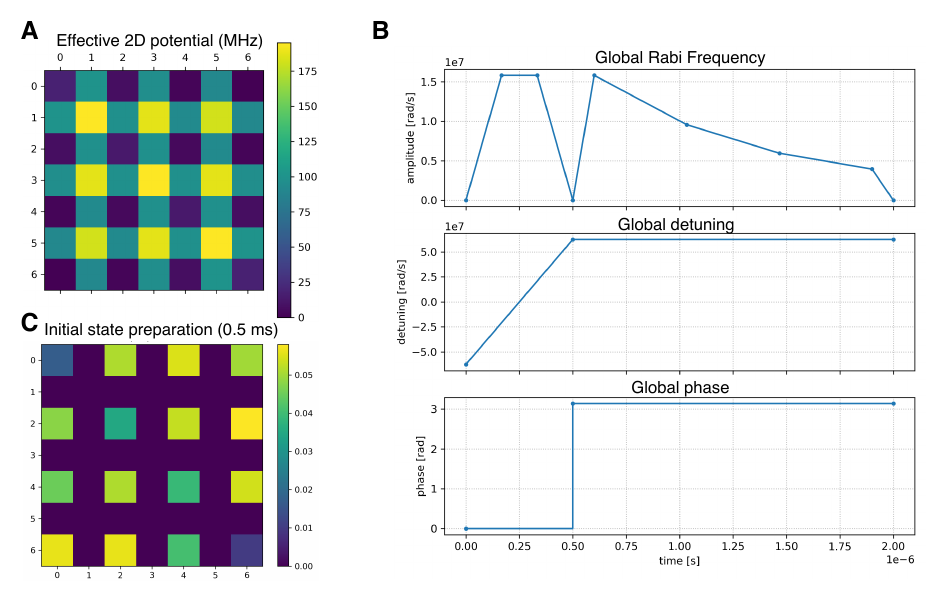}
    \caption{\small Simulating real-space dynamics on QuEra Aquila.
    \textbf{A.} The effective potential field engineered by Rydberg interactions. The unit of the potential is MHz. The potential field is normalized such that the minimal potential is 0.}
    \label{fig:fig_10}
\end{figure}

\paragraph{Initial state preparation.}
The QuEra Aquila machine is initialized to the all-zeros state. This default initial state is not in the embedding subspace of the antiferromagnetic embedding. In our experiment, we use a \emph{3-stage} state preparation subroutine~\cite{ebadi2022quantum} to prepare an initial state that has a significant overlap (around $73\%$) with the embedding subspace.
Since the initial all-zeros state is the ground state when $\Omega=\SI{0}{\mega\hertz}$ and $\Delta<\SI{0}{\mega\hertz}$, this state preparation pulse is quasi-adiabatic and prepares a state with high overlap with the antiferromagnetic encoding subspace.
In our pulse, we use a trapezoidal pulse for the Rabi frequency with maximum frequency $\Omega_{\max}=\SI{15.8}{\mega\hertz}$ and a linear ramp for the detuning from $\SI{-50}{\mega\hertz}$ to $\SI{50}{\mega\hertz}$.
The duration of initial state preparation is $0.5$ $\mu$s. The initial state preparation pulses are shown in \fig{fig_10}C (see the first $0.5$ $\mu$s). The engineered initial state (e.g., the quantum state at $t = 0.5$ $\mu$s) is shown in \fig{fig_10}B.

\paragraph{Time-rescaling and parameters for analog quantum simulation.}
On the QuEra Aquila device, all parameters are specified with a physical unit. We conducted our experiments through Amazon Braket, where the unit of the Rabi frequency $\Omega$ is radians per second (rad/s) and the evolution is in the unit of second (s)~\cite[Section 1.3]{wurtz2023aquila}. Note that $1\ \mathrm{rad/s} = \frac{1}{2\pi}\ \mathrm{Hz}$. 

We can establish a relation between these QuEra physical parameters with the abstract unitless model Hamiltonian \eqn{fdm_operator} via time-rescaling. We assume the quantum simulation is restricted to the unit square $[0,1]^2$ (with Dirichlet boundary conditions), and each dimension is discretized into $8$ pieces. Therefore, the discretized differential operator is $D$ in \eqn{discrete_D} with $h = 1/8$. Therefore, simulating a physical Hamiltonian (on QuEra Aquila) $Q = - \frac{\Omega(t)}{2}\sum_j X_j$ for time $T$ (seconds) is equivalent to simulating the Hamiltonian $Q = -\varphi(t) \frac{1}{2h^2}\sum_j X_j$ (i.e., the antiferromagnetic embedding of the discretized Laplacian $\varphi(t) D/2h^2$) for (unitless) time 
$$T_{\mathrm{eff}} = 10^6 \cdot T,$$
where $\varphi(t) = \frac{h^2\Omega(t)}{2\pi \cdot 10^6}$. 
In our experiments, we run the simulation on QuEra Aquila for $T = 0.5, 1, 1.5$ $\mu$s (plus an initial state preparation duration of $0.5$ $\mu$s). Effectively, we are simulating the real-space quantum dynamics for $T_{\mathrm{eff}} = 0.5, 1, 1.5$.

\paragraph{Post-selection.}
The global Rabi frequency corresponds to the time-dependent function (i.e., $\varphi(t)$) in the kinetic energy. We gradually decrease the global Rabi frequency (see \fig{fig_10}B for the control pulses with total evolution time $T = 2.0$ $\mu$s) so the quantum state will eventually find the low-energy configuration regarding the effective potential field. We consider a measured sample \emph{legit} if it is in the embedding subspace. In our experiment on QuEra Aquila, approximately half of the samples in the measurement are legit. 
In \fig{fig_4}E, we plot the distribution of the legit samples for evolution time $T = 0.5, 1, 1.5$ (corresponding to total physical evolution time $T = 1, 1.5, 2$ $\mu$s).